\documentclass[final,reqno]{amsart}

\usepackage[colorlinks=true,allcolors=blue]{hyperref}
\usepackage{amsaddr}
\usepackage{amssymb,amsfonts,latexsym,amsmath,wrapfig,graphicx,amsthm}
\usepackage[square,numbers]{natbib}
\usepackage{bm}
\usepackage{enumerate}
\usepackage{chngcntr}
\usepackage[dvipsnames]{xcolor}
\usepackage{tikz}
\usetikzlibrary{patterns,decorations.pathreplacing}
\usepackage{physics}
\usepackage{bbold}
\usepackage[makeroom]{cancel}
\usepackage{soul}

\usepackage{stmaryrd}

\usepackage{mathtools}

\usepackage{mathrsfs}

\undef\ul

\undef\Span
\DeclareMathOperator{\Span}{Span}

\undef\div
\DeclareMathOperator{\div}{div}
\DeclareMathOperator{\supp}{supp}

\DeclareMathOperator{\spec}{spec}

\newtheorem{theorem}{Theorem}[section]
\newtheorem{lemma}[theorem]{Lemma}

\newtheorem*{claim}{Claim}

\theoremstyle{remark}
\newtheorem{remark}{Remark}

\theoremstyle{definition}
\newtheorem{definition}[theorem]{Definition}

\renewcommand{\le}{\leqslant}
\renewcommand{\ge}{\geqslant}

\newcommand{\RR}{\mathbb{R}}
\newcommand{\CC}{\mathbb{C}}
\newcommand{\ZZ}{\mathbb{Z}}
\newcommand{\NN}{\mathbb{N}}

\newcommand{\iden}{\mathbb{1}}

\def\XXint#1#2#3{{\setbox0=\hbox{$#1{#2#3}{\int}$ }
\vcenter{\hbox{$#2#3$ }}\kern-.6\wd0}}

\newcommand{\ul}[1]{\underline{#1}}

\newcommand{\TC}{\mathcal{T}}

\newcommand{\EC}{\mathcal{E}}

\newcommand{\Cc}{\mathcal{C}}

\newcommand{\BC}{\mathcal{B}}

\newcommand{\TF}{\mathrm{TF}}

\newcommand{\UEG}{\mathrm{UEG}}
\newcommand{\Jel}{\mathrm{Jel}}

\newcommand{\Fock}{\mathcal{F}}

\newcommand{\Hil}{\mathfrak{H}}

\let\phi=\varphi

\newcommand{\vGamma}{{\bm{\Gamma}}}

\undef\vb
\newcommand{\vb}{{\mathbf{b}}}

\renewcommand{\epsilon}{\varepsilon}

\newcommand{\tetra}{\mathbb{\Delta}}
\newcommand{\SO}{\mathrm{SO}}

\title[Equality of Jellium and UEG]{Equality of the Jellium and the\\Uniform Electron Gas Energies}

\author{Mih\'aly A. Csirik}
\address{Department of Mathematical Sciences, Chalmers University of
Technology and the University of Gothenburg, G\"oteborg, Sweden}
\email{csirik@gmail.com}
\hypersetup{
  pdftitle={Equality of the Jellium and the Uniform Electron Gas Energies},
  pdfauthor={Mih\'aly A. Csirik}
}

\begin{document}
\raggedbottom

\begin{abstract}
Jellium, a system of electrons in a uniform neutralizing background, and the uniform electron gas (UEG), a system of electrons with constant density and no background, are two standard models of the electron gas. We prove that the ground-state energy per unit volume of Jellium, as defined by Lieb and Narnhofer, coincides with the UEG energy defined by Lewin, Lieb and Seiringer, i.e.
\[
e^\hbar_\Jel(\rho_0)=e^\hbar_\UEG(\rho_0).
\]
This equality was previously established in the classical case by Lewin, Lieb and Seiringer and independently by Cotar and Petrache.
\end{abstract}

\maketitle


\section{Introduction}

The electron gas, an infinite system of interacting electrons spread uniformly over space, is one of the most studied models of many-body quantum mechanics~\cite{giuliani2005quantum}. In the physics literature the terms ``electron gas'', (interacting) ``Fermi gas'' and ``Jellium'' are mostly used interchangeably. Jellium is a system of electrons moving in a uniform background of positive charge which exactly neutralizes them. The uniformity of Jellium is, however, only that of the background. Wigner predicted that the electrons crystallize at low density~\cite{wigner1934interaction}, and so the ground state of Jellium is not expected to have constant density~\cite{lewin2019floating}. Mathematically, the thermodynamic limit of Jellium was established by Lieb and Narnhofer~\cite{lieb1975thermodynamic}, building on the work of Lieb and Lebowitz~\cite{lieb1972constitution}. For the ground-state energy their result reads as follows. Let $\{\Omega_N\}$ be a sequence of bounded domains that grow regularly to $\RR^3$, for instance balls, with $\rho_0|\Omega_N|=N$. Then the limit
\begin{multline*}
e^\hbar_\Jel(\rho_0)=\lim_{N\to\infty}\frac1{|\Omega_N|}\inf\spec\bigg(-\hbar^2\sum_{j=1}^N\Delta_{x_j}+\sum_{1\le j<k\le N}\frac{1}{|x_j-x_k|}\\
-\rho_0\sum_{j=1}^N\int_{\Omega_N}\frac{\mathrm dy}{|x_j-y|}+\frac{\rho_0^2}2\iint_{\Omega_N\times\Omega_N}\frac{\mathrm dx\,\mathrm dy}{|x-y|}\bigg)
\end{multline*}
exists and does not depend on the shape of the domains. Here the Hamiltonian acts on antisymmetric functions of $N$ variables in $\Omega_N$, each carrying $q$ spin states, with homogeneous Dirichlet boundary conditions on $\partial\Omega_N$. The third term is the attraction between the electrons and the background, and the last one is the electrostatic self-energy of the background. The system is thus \emph{exactly neutral} as the number of electrons equals the total charge of the background.

The approach of Lewin, Lieb and Seiringer~\cite{lewin2018statistical} does not involve the ground-state problem. They \emph{define} the uniform electron gas (UEG) as a system of electrons whose density is constant everywhere, without any background. In the quantum case its energy is expressed through the grand-canonical Levy--Lieb functional~\cite{lewin2019universal}, the minimal kinetic and interaction energy of electrons with a prescribed density. Since quantum densities cannot be cut off sharply, the indicator function of the domain has to be smoothed. In the convenient form of~\cite{lewin2020local}, the UEG energy per unit volume is
\begin{multline*}
e^\hbar_\UEG(\rho_0)=\\
\lim_{N\to\infty}\frac1{|\Omega_N|}\Bigg(\inf_{\rho_{\vGamma}=\rho_0\iden_{\Omega_N}*\chi_{\delta_N}}\sum_{n\ge1}\Tr\bigg[\bigg(-\hbar^2\sum_{j=1}^n\Delta_{x_j}+\sum_{1\le j<k\le n}\frac1{|x_j-x_k|}\bigg)\Gamma^{(n)}\bigg]\\
-\frac{\rho_0^2}2\iint_{\RR^3\times\RR^3}\frac{(\iden_{\Omega_N}*\chi_{\delta_N})(x)\,(\iden_{\Omega_N}*\chi_{\delta_N})(y)}{|x-y|}\,\mathrm dx\,\mathrm dy\Bigg),
\end{multline*}
where $\chi_\delta(x)=\delta^{-3}\chi(\delta^{-1}x)$ for a fixed radial function $0\le\chi\in C_c^\infty(\RR^3)$ with $\supp\chi\subset B_1$ and $\int_{\RR^3}\chi=1$. The domains satisfy the hypotheses of Theorem~\ref{thm:UEG}(ii), and the smoothing scales $\delta_N>0$ satisfy $\delta_N/|\Omega_N|^{1/3}\to0$ and $\delta_N|\Omega_N|^{1/3}\to\infty$. The infimum is taken over all grand-canonical fermion states $\vGamma=(\Gamma^{(n)})_{n\ge0}$, satisfying the constraint $\rho_{\vGamma}=\rho_0\iden_{\Omega_N}*\chi_{\delta_N}$. The limit exists~\cite{lewin2018statistical,lewin2020local}, and it has the expected behavior at high density~\cite{lewin2018statistical}.

For the classical UEG (which is obtained by taking $\hbar=0$, dropping $\chi_\delta$ and optimizing over classical $N$-particle probability measures), it was shown that the Jellium and the UEG energies coincide~\cite{lewin2019floating,cotar2017equality}. Averaging a crystal over translations (the ``floating crystal'') gives a constant density, but the charge fluctuations at the boundary of the sample raise the energy by an amount proportional to the volume~\cite[Appendix~B]{lewin2015improved}. This issue was overcome by Lewin, Lieb and Seiringer, who constructed a floating Wigner crystal with no boundary charge fluctuations by introducing a liquid layer at the boundary~\cite{lewin2019floating}, and independently by Cotar and Petrache~\cite{cotar2017equality} by a different method.

In the quantum case, it is believed that the two energies coincide as well~\cite[Sec.~2.2]{lewin2020local}. The easy inequality $e_\UEG^\hbar(\rho_0)\ge e_\Jel^\hbar(\rho_0)$ was proved in~\cite{lewin2018statistical} for a grand-canonical version of Jellium. The reverse inequality is harder, because Jellium states, which are expected to be crystalline, have to be turned into states with an exactly constant density at a cost which is negligible compared with the volume.

In this paper we prove that the Jellium and the UEG energies coincide also in the quantum case. Our main result, Theorem~\ref{thm:main} below, states that
\[
\boxed{e^\hbar_\Jel(\rho_0)=e^\hbar_\UEG(\rho_0).}
\]
In particular, the known properties of the UEG carry over to Jellium (Remark~\ref{rem:main_consequences}). For instance, at fixed $\hbar>0$,
\[
e^\hbar_\Jel(\rho_0)=\frac35\Big(\frac{6\pi^2}{q}\Big)^{2/3}\hbar^2\rho_0^{5/3}-\frac34\Big(\frac{6}{q\pi}\Big)^{1/3}\rho_0^{4/3}+o(\rho_0^{4/3})\qquad\text{as }\rho_0\to\infty.
\]
The first term is the kinetic energy of the free Fermi gas and the second one is the Dirac exchange energy. Moreover, $\rho_0^{-4/3}e^\hbar_\Jel(\rho_0)$ converges to the classical Jellium energy at unit density as $\hbar^2\rho_0^{1/3}\to0$ (see Figure~\ref{fig:energy_shape}).
\begin{figure}[ht]
\centering
\begin{tikzpicture}[>=stealth,xscale=1.75,yscale=1.75,
  declare function={
    fg(\x)=\x^(5/3)-\x^(4/3);
    cl(\x)=-1.9553*\x^(4/3);
    ee(\x)=\x^(5/3)-\x^(4/3)-0.9553*\x^(4/3)/(1+3.443*\x^(1/6));}]
  \draw[->] (0,0) -- (2.75,0) node[below left] {$\rho_0$};
  \draw[->] (0,-0.95) -- (0,1.55);
  \draw[densely dashed,domain=0.0001:2.62,samples=120,smooth] plot (\x,{fg(\x)});
  \draw[dash dot,domain=0.0001:0.62,samples=60,smooth] plot (\x,{cl(\x)});
  \draw[very thick,domain=0.0001:2.62,samples=160,smooth] plot (\x,{ee(\x)});
  \node[right] at (2.62,{fg(2.62)}) {$\hbar^2c_\TF\rho_0^{5/3}-c_{\mathrm D}\rho_0^{4/3}$};
  \node[right] at (2.62,{ee(2.62)}) {$e^\hbar_\UEG(\rho_0)=e^\hbar_\Jel(\rho_0)$};
  \node[below right] at (0.62,{cl(0.62)}) {$e_\UEG\,\rho_0^{4/3}$};
  \node[left] at (0,0) {$0$};
\end{tikzpicture}
\caption{Sketch of $e^\hbar_\UEG(\rho_0)=e^\hbar_\Jel(\rho_0)$ as a function of $\rho_0$.}
\label{fig:energy_shape}
\end{figure}
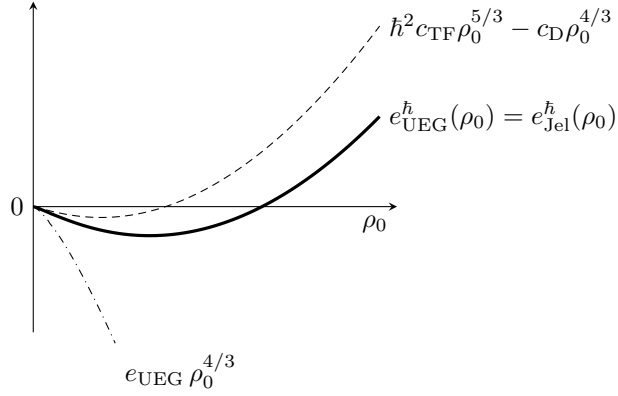
\vspace{.3em}\\
\noindent\emph{Organization of the paper.} In Section~\ref{sec:setting}, we fix our notation (Section~\ref{sec:notation}) and recall the definitions and the known results for Jellium (Section~\ref{sec:LN}) and for the UEG (Section~\ref{sec:UEG}). Section~\ref{sec:main} contains our main result, Theorem~\ref{thm:main}, its consequences for Jellium (Remark~\ref{rem:main_consequences}) and an outline of its proof (Section~\ref{sec:main_strategy}). The proofs are given in Section~\ref{sec:proofs}: after collecting some general facts in Section~\ref{sec:basic}, we prove the inequality $e^\hbar_\Jel(\rho_0)\le e^\hbar_\UEG(\rho_0)$ in Section~\ref{sec:lower} and the inequality $e^\hbar_\UEG(\rho_0)\le e^\hbar_\Jel(\rho_0)$ in Section~\ref{sec:upper}. Appendix~\ref{sec:LN_proof} sketches the proof of the thermodynamic limit of Jellium on balls at zero temperature, and Appendix~\ref{sec:canonical_completion} contains the neutralization construction used in Section~\ref{sec:lower}.

\subsection*{Acknowledgments}
The author has received funding from the ERC-2021-STG under grant agreement No.~101041487 REGAL.

\subsection*{Use of AI tools}
The author made substantial use of the LLMs ChatGPT (OpenAI) and Claude (Anthropic) in the preparation of this work, including the development of the proofs, writing and making figures. The author has checked all the proofs by hand, and takes full responsibility for the content of the article.

\subsection*{Lean verification}
The formalization in \texttt{lean\_gpt}\footnote{\url{https://github.com/xyzw/jellium_ueg}}, produced with the help of the same AI tools, checks the lower and upper comparisons of Sections~\ref{sec:lower} and~\ref{sec:upper} and the neutralization construction of Appendix~\ref{sec:canonical_completion} in Lean~4 with Mathlib.
The verification is relative to explicit hypotheses on abstract states and charges. The cited literature results, including the physical Lieb--Narnhofer thermodynamic limit, and the general facts of Section~\ref{sec:basic} remain inputs. The representation of quantum states by density operators on Hilbert space and the general identities for their kinetic and Coulomb energies are assumed rather than formalized.

\section{Setting}\label{sec:setting}

\subsection{Notation}\label{sec:notation}

For $\ell>0$ we write $C_\ell=(-\ell/2,\ell/2)^3$ for the open cube of side length $\ell$ centered at the origin. For functions $f,g:\RR^3\to\RR$  we write
\[
D(f,g):=\frac12\int_{\RR^3}\int_{\RR^3}\frac{f(x)g(y)}{|x-y|}\,\mathrm{d}x\,\mathrm{d}y,\qquad D(f):=D(f,f)
\]
for the Coulomb inner product and the direct Coulomb energy. For an integer spin multiplicity $q\ge1$, the Thomas--Fermi and Dirac constants are
\[
c_\TF=\frac35\Big(\frac{6\pi^2}{q}\Big)^{2/3},\qquad c_{\mathrm{D}}=\frac34\Big(\frac{6}{q\pi}\Big)^{1/3}.
\]
Let $\Hil^n:=\bigwedge^nL^2(\RR^3\times\{1,\dots,q\})$, with $\Hil^0=\CC$. To us, a grand-canonical fermion state is a family $\vGamma=(\Gamma^{(n)})_{n\ge0}$ of operators $0\le\Gamma^{(n)}=(\Gamma^{(n)})^*$ on $\Hil^n$ with $\sum_{n\ge0}\Tr\Gamma^{(n)}=1$, that is, a state on the Fock space $\bigoplus_{n\ge0}\Hil^n$ that commutes with the particle-number operator, i.e. a number-conserving state. Writing $\mathbf{x}=(x,s)\in\RR^3\times\{1,\dots,q\}$, the fermionic statistics means that the integral kernel of $\Gamma^{(n)}$ is antisymmetric in each group of variables:
\[
\Gamma^{(n)}(X_\sigma;Y)=\Gamma^{(n)}(X;Y_\sigma)=\operatorname{sgn}(\sigma)\,\Gamma^{(n)}(X;Y)
\]
for all $X=(\mathbf{x}_1,\dots,\mathbf{x}_n)$, $Y=(\mathbf{y}_1,\dots,\mathbf{y}_n)$ and every permutation $\sigma\in\mathfrak{S}_n$, where $X_\sigma:=(\mathbf{x}_{\sigma(1)},\dots,\mathbf{x}_{\sigma(n)})$. The density of $\vGamma$ is $\rho_{\vGamma}:=\sum_{n\ge1}\rho_{\Gamma^{(n)}}$, where
\[
\rho_{\Gamma^{(n)}}(x):=n\sum_{s=1}^q\int\Gamma^{(n)}(x,s,\mathbf{x}_2,\dots,\mathbf{x}_n;x,s,\mathbf{x}_2,\dots,\mathbf{x}_n)\,\mathrm{d}\mathbf{x}_2\cdots\mathrm{d}\mathbf{x}_n,
\]
with the convention $\int\mathrm{d}\mathbf{x}:=\sum_{s=1}^q\int_{\RR^3}\mathrm{d}x$. Thus $\int_{\RR^3}\rho_{\vGamma}=\sum_{n\ge1}n\Tr\Gamma^{(n)}$ is the expected number of particles.

\subsection{Jellium}\label{sec:LN}

Here we summarize the results of Lieb and Narnhofer~\cite{lieb1975thermodynamic} on the thermodynamic limit of Jellium, restricted to what concerns the \emph{quantum ground-state energy per unit volume} $e^\hbar_\Jel(\rho_0)$.

Let $\rho_0>0$. For a bounded measurable set $\Omega\subset\RR^3$ and $N\ge1$, the classical Jellium energy of $N$ points $X=(x_1,\dots,x_N)\in(\RR^3)^N$ in the background $\rho_0\iden_\Omega$ is
\begin{equation}\label{eq:LN_U}
U^{\rho_0}_\Omega(X):=\sum_{1\le j<k\le N}\frac{1}{|x_j-x_k|}-\rho_0\sum_{j=1}^N\int_\Omega\frac{\mathrm{d}y}{|x_j-y|}+D(\rho_0\iden_\Omega).
\end{equation}
Here the points need not lie in $\Omega$, and $N$ need not equal $\rho_0|\Omega|$. Now let $\Omega$ be a bounded open set with $N:=\rho_0|\Omega|\in\NN$ (neutrality). The quantum Jellium Hamiltonian is formally
\begin{equation}\label{eq:LN_H}
H^{\hbar,\rho_0}_\Omega:=\sum_{j=1}^N\big(-\hbar^2\Delta_{x_j}\big)+U^{\rho_0}_\Omega(x_1,\dots,x_N),
\end{equation}
acting on $\bigwedge_{j=1}^NL^2(\Omega\times\{1,\dots,q\})$. More precisely, for $\hbar>0$, $H^{\hbar,\rho_0}_\Omega$ denotes the Friedrichs extension of this operator from $\bigwedge^NC_c^\infty(\Omega;\CC^q)$, whose form domain consists of antisymmetric functions in $H^1_0(\Omega^N;(\CC^q)^{\otimes N})$; this corresponds to \emph{homogeneous Dirichlet boundary conditions} on $\partial\Omega$. Its ground-state energy is
given by the variational principle
\[
\boxed{E^\hbar_\Jel(\rho_0,\Omega):=\inf_{\|\Psi\|=1}\big\langle\Psi,H^{\hbar,\rho_0}_\Omega\Psi\big\rangle.}
\]
For $\hbar=0$ this is the classical energy $E^0_\Jel(\rho_0,\Omega)=\inf_{X\in\Omega^N}U^{\rho_0}_\Omega(X)$, see~\cite{lieb1975thermodynamic,sari1976nu,lewin2018statistical,lewin2022coulomb}. The energy per unit volume is
\[
e^\hbar_\Jel(\rho_0,\Omega):=\frac{E^\hbar_\Jel(\rho_0,\Omega)}{|\Omega|}.
\]
The class of admissible domains $\Omega$ for taking the thermodynamic limit is essentially that of Lieb and Lebowitz~\cite[Sec.~V]{lieb1972constitution} (conditions (A), (B) below, which we paraphrase), together with four extra conditions introduced for the quantum case in~\cite[Sec.~5.2]{lieb1975thermodynamic}. For $h>0$, consider the inner and outer layers of $\Omega_N$ of width $h$,
\begin{align*}
\Sigma^{-1}_{N,h}&:=\{x\in\Omega_N:\ d(x,\partial\Omega_N)\le h\}, & \Sigma^{-2}_{N,h}&:=\{x\in\Omega_N:\ h<d(x,\partial\Omega_N)\le2h\},\\
\Sigma^{+1}_{N,h}&:=\{x\notin\Omega_N:\ d(x,\Omega_N)\le h\}, & \Sigma^{+2}_{N,h}&:=\{x:\ h<d(x,\Omega_N)\le2h\}
\end{align*}
(see Figure~\ref{fig:LN_layers}), and the neutral double layers below, with the inner one defined whenever $|\Sigma^{-2}_{N,h}|>0$,
\[
\sigma^-_{N,h}:=\iden_{\Sigma^{-1}_{N,h}}-\frac{|\Sigma^{-1}_{N,h}|}{|\Sigma^{-2}_{N,h}|}\iden_{\Sigma^{-2}_{N,h}},\qquad
\sigma^+_{N,h}:=\iden_{\Sigma^{+1}_{N,h}}-\frac{|\Sigma^{+1}_{N,h}|}{|\Sigma^{+2}_{N,h}|}\iden_{\Sigma^{+2}_{N,h}}.
\]

\begin{figure}[ht]
\centering
\begin{tikzpicture}[>=stealth,x=0.9cm,y=0.9cm,dom/.style={rounded corners=0.45cm,line join=round}]
  \def\dom{(-2.6,-1.6) -- (2.4,-1.6) -- (2.6,0.4) -- (1.0,0.6) -- (0.6,1.8) -- (-1.8,1.6) -- (-2.8,0.2) -- cycle}
  \draw[dom,red!15,line width=1.08cm] \dom;
  \draw[dom,red!40,line width=0.54cm] \dom;
  \fill[dom,white] \dom;
  \begin{scope}[dom]
    \clip \dom;
    \draw[dom,blue!15,line width=1.08cm] \dom;
    \draw[dom,blue!40,line width=0.54cm] \dom;
  \end{scope}
  \draw[dom,thick] \dom;
  \node at (-1.0,0.4) {$\Omega_N$};
  \draw (-1.0,-2.2) -- (-1.0,-1.0);
  \foreach \y in {-2.2,-1.9,-1.6,-1.3,-1.0} \draw (-1.08,\y) -- (-0.92,\y);
  \foreach \y in {-2.05,-1.75,-1.45,-1.15} \node[left,inner sep=1pt] at (-1.08,\y) {\scriptsize$h$};
  \draw[->] (4.2,1.0) node[right] {$\Sigma^{-2}_{N,h}$} -- (2.087,-0.205);
  \draw[->] (4.2,0.2) node[right] {$\Sigma^{-1}_{N,h}$} -- (2.361,-0.485);
  \draw[->] (4.2,-0.6) node[right] {$\Sigma^{+1}_{N,h}$} -- (2.634,-0.765);
  \draw[->] (4.2,-1.4) node[right] {$\Sigma^{+2}_{N,h}$} -- (2.908,-1.045);
\end{tikzpicture}
\caption{Two-dimensional sketch of the inner layers $\Sigma^{-1}_{N,h},\Sigma^{-2}_{N,h}$ (blue) and the outer layers $\Sigma^{+1}_{N,h},\Sigma^{+2}_{N,h}$ (red) of $\Omega_N$, each of width $h$. The double layer $\sigma^+_{N,h}$ is positive on $\Sigma^{+1}_{N,h}$ and negative on $\Sigma^{+2}_{N,h}$, and similarly for $\sigma^-_{N,h}$.}
\label{fig:LN_layers}
\end{figure}
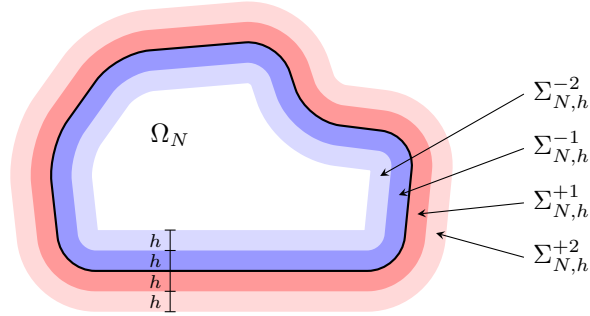

\begin{definition}\label{def:LN_domains}
A sequence $\{\Omega_N\}_{N\ge1}$ of bounded open sets with $\rho_0|\Omega_N|=N$ for every integer $N\ge1$ is \emph{LN-admissible} if
\begin{enumerate}
\item[(A)] (Van Hove) $\big|\{x:\ d(x,\partial\Omega_N)\le h\}\big|=o(|\Omega_N|)$ for every $h>0$;
\item[(B)] (ball condition) $|\Omega_N|\ge\delta\,|B(\Omega_N)|$ for some $\delta>0$ independent of $N$, where $B(\Omega_N)$ is the smallest ball containing $\Omega_N$;
\item[(i)] $\sup_N|\Sigma^{+1}_{N,h}|/|\Sigma^{+2}_{N,h}|<\infty$ for every $h>0$;
\item[(ii)] the potential $\varphi_{N,h}(x):=\int_{\RR^3}\sigma^+_{N,h}(y)|x-y|^{-1}\,\mathrm dy$ is $o(1)$
as $h\to0$, uniformly in $N\ge1$ and $x\in\Omega_N\cup\Sigma^{+1}_{N,h}\cup\Sigma^{+2}_{N,h}$;
\item[(iii)] $D(\sigma^+_{N,h})=o(|\Omega_N|)$ for every $h>0$;
\item[(iv)] for every $h>0$, $|\Sigma^{-2}_{N,h}|>0$ for all sufficiently large $N$, and $D(\sigma^-_{N,h})=o(|\Omega_N|)$ as $N\to\infty$.
\end{enumerate}
\end{definition}

For comparison, the Fisher $\eta$-regularity used in~\cite{lewin2018statistical} implies (A), and together with connectedness it implies (B) by~\cite[Sec.~V and Appendix~A]{lieb1972constitution}, as used in the proof of~\cite[Thm.~2.6]{lewin2018statistical}. Neutral balls form an LN-admissible sequence (Remark~\ref{rem:LN_balls} in Appendix~\ref{sec:LN_proof}).

\begin{theorem}[Stability {\cite[Appendix]{lieb1975thermodynamic}}]\label{thm:LN_stability}
Let $\rho_0>0$ and let $\Omega\subset\RR^3$ be a bounded measurable set.
\begin{enumerate}[(i)]
\item (Classical bound) For every $N\ge1$ and every $X\in(\RR^3)^N$,
\begin{equation}\label{eq:LN_Onsager}
U^{\rho_0}_\Omega(X)\ge-\frac{9}{10}\,\frac{N}{r_s}=-\frac{9}{10}\Big(\frac{4\pi}{3}\Big)^{1/3}\rho_0^{1/3}N,\qquad r_s=\Big(\frac{3}{4\pi\rho_0}\Big)^{1/3}.
\end{equation}
In particular, if $N=\rho_0|\Omega|$, then $E^\hbar_\Jel(\rho_0,\Omega)\ge E^0_\Jel(\rho_0,\Omega)\ge-c_{\rm LN}\,\rho_0^{4/3}|\Omega|$, where
\[
c_{\rm LN}:=\frac35\Big(\frac{9\pi}{2}\Big)^{1/3}=\frac{9}{10}\Big(\frac{4\pi}{3}\Big)^{1/3}\simeq1.4508.
\]
\item (Quantum bound) If $\Omega$ is open and $N=\rho_0|\Omega|$, then
\begin{equation}\label{eq:LN_quantum_stability}
E^\hbar_\Jel(\rho_0,\Omega)\ge\big(\hbar^2c_\TF\,\rho_0^{5/3}-c_{\rm LN}\,\rho_0^{4/3}\big)|\Omega|.
\end{equation}
\end{enumerate}
\end{theorem}

By Theorem~\ref{thm:LN_stability}, the energy per unit volume $e^\hbar_\Jel(\rho_0,\Omega)$ is bounded below uniformly in $\Omega$. The main result of~\cite{lieb1975thermodynamic} is that it converges along LN-admissible sequences, to a limit which does not depend on the shape of the domains.

\begin{theorem}[Jellium {\cite[Secs.~5.1--5.4]{lieb1975thermodynamic}}]\label{thm:LN_main}
Let $\hbar>0$ and $\rho_0>0$.
\begin{enumerate}[(i)]
\item (Balls) Along the standard sequence of balls $\{\mathcal B_K\}$ of Lemma~\ref{lem:LN_packing} (which depends on $\rho_0$), the limit $\lim_{K\to\infty}e^\hbar_\Jel(\rho_0,\mathcal B_K)$ exists.
\item (General domains) For every LN-admissible sequence $\{\Omega_N\}$, the limit
\[
\boxed{e^\hbar_\Jel(\rho_0):=\lim_{N\to\infty}e^\hbar_\Jel(\rho_0,\Omega_N)}
\]
exists and is equal to the limit in item~(i). In particular, it does not depend on the shape of the domains.
\item (Scaling) $e^\hbar_\Jel(\rho_0)=\rho_0^{4/3}\,\tilde e_\Jel\big(\hbar^2\rho_0^{1/3}\big)$, where $\tilde e_\Jel(\lambda):=e^{\sqrt\lambda}_\Jel(1)$, $\lambda>0$, is the energy at unit density.
\item (Regularity) $\lambda\mapsto \tilde e_\Jel(\lambda)$ is concave and non-decreasing on $(0,\infty)$, hence continuous. Therefore $\rho_0\mapsto e^\hbar_\Jel(\rho_0)$ is continuous on $(0,\infty)$. Moreover, for a fixed LN-admissible sequence at unit density, the convergence $e^{\sqrt\lambda}_\Jel(1,\Omega_N)\to \tilde e_\Jel(\lambda)$ is locally uniform in $\lambda\in(0,\infty)$.
\end{enumerate}
\end{theorem}

For $\hbar=0$ we write $e^0_\Jel(\rho_0)$ for the limit in item~(ii) whenever it exists for all LN-admissible sequences and is independent of them, and $e_\Jel:=e^0_\Jel(1)$. (The Jellium energy of~\cite{lewin2018statistical} is defined through a slightly different grand-canonical problem; see Remark~\ref{rem:UEG_vs_Jel}.) By Theorem~\ref{thm:LN_stability}, $e_\Jel\ge-c_{\rm LN}$, $\tilde e_\Jel(\lambda)\ge c_\TF\lambda-c_{\rm LN}$ for all $\lambda>0$, and $e^\hbar_\Jel(\rho_0)\ge-\rho_0\,c_{\rm LN}^2/(4\hbar^2c_\TF)$ for all $\rho_0>0$.

Items (i)--(ii) are proved in~\cite[Secs.~5.1--5.2]{lieb1975thermodynamic} for the free energy at $T>0$; the ground-state version is stated in~\cite[Sec.~5.4(iii)]{lieb1975thermodynamic}. A sketch of the proof of item~(i) at $T=0$ is given in Appendix~\ref{sec:LN_proof}. Item (iii) is the zero-temperature form of~\cite[Eqs.~(62)--(64)]{lieb1975thermodynamic}. Item (iv) is the zero-temperature analog of~\cite[Secs.~3.4.1 and~5.3]{lieb1975thermodynamic}: at $T>0$, continuity in $\rho_0$ is derived there from the joint convexity of $(\beta,\gamma)\mapsto\log\Tr e^{-\beta K-\gamma U}$ and uniformity from Dini's theorem. At $T=0$ the same conclusions follow from concavity in $\lambda=\hbar^2$, which is the argument of~\cite{lewin2018statistical} for $\tilde e_\UEG(\lambda)$.

\subsection{The Uniform Electron Gas}\label{sec:UEG}

We briefly recall the definition of the quantum UEG given by Lewin, Lieb and Seiringer, in two versions. The first is from~\cite[Sec.~5]{lewin2018statistical}, where the density is optimized over a transition region near the boundary. The second is from~\cite[Sec.~2]{lewin2020local}, where the density is a fixed smoothing of $\rho_0\iden_\Omega$. In contrast to Jellium there is no background; instead, the one-particle density is prescribed.

Let $\rho\ge0$ with $\sqrt\rho\in H^1(\RR^3)$. The grand-canonical Levy--Lieb functional (\cite[Eq.~(5.1)]{lewin2018statistical}, \cite[Sec.~2.1]{lewin2020local}) is
\[
F^\hbar_{\rm LL}(\rho):=\inf_{\rho_{\vGamma}=\rho}\sum_{n\ge1}\Tr_{\Hil^n}\Bigg(-\hbar^2\sum_{j=1}^n \Delta_{x_j}+\sum_{1\le j<k\le n}\frac{1}{|x_j-x_k|}\Bigg)\Gamma^{(n)},
\]
and the indirect energy of $\rho$ is
\begin{equation}\label{eq:UEG_E}
\boxed{E^\hbar(\rho):=F^\hbar_{\rm LL}(\rho)-D(\rho).}
\end{equation}
\begingroup\emergencystretch=1em
The number $\int_{\RR^3}\rho$ need not be an integer.
\par\endgroup

Let $\chi\ge0$ be radial with compact support, $\int_{\RR^3}\chi=1$ and $\int_{\RR^3}|\nabla\sqrt\chi|^2<\infty$, and set $\chi_\delta:=\delta^{-3}\chi(\cdot/\delta)$ for $\delta>0$. For a bounded set $\Omega\subset\RR^3$, the UEG energy per unit volume with smoothing $\chi_\delta$ is
\[
e^\hbar_\UEG(\rho_0,\Omega,\chi_\delta):=\frac{E^\hbar\big(\rho_0\iden_\Omega\ast\chi_\delta\big)}{|\Omega|}.
\]

The following theorem gives the common thermodynamic limit of the two UEG constructions, together with its scaling, regularity and limiting properties.

\begin{theorem}[Quantum UEG {\cite[Thm.~5.1]{lewin2018statistical}, \cite[Thm.~1]{lewin2020local}}]\label{thm:UEG}
Let $\rho_0>0$ and $\hbar>0$.
\begin{enumerate}
\item[(i)]  Let $\{\Omega_N\}$ be bounded connected domains with $|\Omega_N|\to\infty$ which have an $\eta$-regular boundary in the sense of Fisher~\cite{fisher1964free}, for a fixed $\eta$:
\[
\big|\{x:\ d(x,\partial\Omega_N)\le|\Omega_N|^{1/3}t\}\big|\le|\Omega_N|\,\eta(t),\qquad0\le t\le t_0,
\]
with $\eta$ continuous and $\eta(0)=0$. Let $\ell_N\to\infty$ with $\ell_N/|\Omega_N|^{1/3}\to0$, and
\begin{align*}
&\Omega_N^-:=\{x\in\Omega_N:\ d(x,\partial\Omega_N)\ge\ell_N\},\\
&\Omega_N^+:=\Omega_N\cup\{x:\ d(x,\partial\Omega_N)\le\ell_N\}.
\end{align*}
Then the limit
\begin{equation}\label{eq:UEG_LLS18}
\boxed{e^\hbar_\UEG(\rho_0):=\lim_{N\to\infty}\ \inf_{\substack{\sqrt\rho\in H^1(\RR^3)\\ \rho_0\iden_{\Omega_N^-}\le\rho\le\rho_0\iden_{\Omega_N^+}}}\frac{E^\hbar(\rho)}{|\Omega_N|}}
\end{equation}
exists and is independent of $\{\Omega_N\}$ and $\{\ell_N\}$.
\item[(ii)] Let $\{\Omega_N\}$ be bounded connected domains with $|\Omega_N|\to\infty$ and
\[
|\partial\Omega_N+B_r|\le Cr|\Omega_N|^{2/3}\qquad\text{for all }0<r\le|\Omega_N|^{1/3}/C,
\]
for some $C>0$. Let $\delta_N>0$ with $\delta_N/|\Omega_N|^{1/3}\to0$ and $\delta_N|\Omega_N|^{1/3}\to\infty$. Then
\begin{equation}\label{eq:UEG_LDA}
\boxed{e^\hbar_\UEG(\rho_0)=\lim_{N\to\infty}e^\hbar_\UEG(\rho_0,\Omega_N,\chi_{\delta_N}),}
\end{equation}
and the limit is independent of $\{\Omega_N\}$, $\{\delta_N\}$ and $\chi$.
\item[(iii)] (Scaling) $e^\hbar_\UEG(\rho_0)=\rho_0^{4/3}\,\tilde e_\UEG\big(\hbar^2\rho_0^{1/3}\big)$, where $\tilde e_\UEG(\lambda):=e^{\sqrt\lambda}_\UEG(1)$, $\lambda>0$, is the energy at unit density.
\item[(iv)] (Regularity) $\lambda\mapsto \tilde e_\UEG(\lambda)$ is concave and non-decreasing on $(0,\infty)$, hence continuous. Therefore $\rho_0\mapsto e^\hbar_\UEG(\rho_0)$ is continuous on $(0,\infty)$.
\item[(v)] (Classical limit) $\lim_{\lambda\to0}\tilde e_\UEG(\lambda)=e_\UEG$, the classical UEG energy of~\cite[Thm.~2.6]{lewin2018statistical}.
\item[(vi)] (High-density limit) $\tilde e_\UEG(\lambda)=c_\TF\lambda-c_{\mathrm{D}}+o(1)$ as $\lambda\to\infty$.
\end{enumerate}
\end{theorem}

Items (i) and (iii)--(vi) are~\cite[Thm.~5.1]{lewin2018statistical}, and the limit in item~(ii) is~\cite[Thm.~1]{lewin2020local}. The two limits coincide, since both reduce to the limit for tetrahedra with a fixed smearing, namely~\cite[Rem.~5.2 and Lemma~5.6]{lewin2018statistical} and~\cite[Prop.~17]{lewin2020local} at fixed $\delta$.

In item~(i), the density is held at the fixed constant $\rho_0$ well inside $\Omega_N$ and to $0$ far outside. In the transition layer of width $2\ell_N$ around $\partial\Omega_N$ it is only required that $0\le\rho\le\rho_0$, and the energy is minimized over all admissible $\rho$.

In item~(ii), the condition $\delta_N/|\Omega_N|^{1/3}\to0$ makes the smeared region around $\partial\Omega_N$ negligible in volume. The condition $\delta_N|\Omega_N|^{1/3}\to\infty$ makes the kinetic energy of the transition region negligible, since $\int_{\RR^3}|\nabla\sqrt{\iden_{\Omega_N}\ast\chi_{\delta_N}}|^2\le C|\Omega_N|^{2/3}/\delta_N$. In particular $\delta_N$ may tend to zero, provided $\delta_N\gg|\Omega_N|^{-1/3}$. Under the weaker assumption that $\Omega_N$ is $\eta$-regular, the same holds with the second condition replaced by $\delta_N^{-2}\eta(\delta_N|\Omega_N|^{-1/3})\to0$~\cite[Rem.~2]{lewin2020local}.

\begin{remark}[UEG versus Jellium]\label{rem:UEG_vs_Jel}
It is noted in~\cite[Sec.~2.2]{lewin2020local} that $e_\UEG$ ``is believed to be the same as the ground state energy of Jellium''. In the classical case, the equality $e_\UEG=e_\Jel$ has been proved~\cite{lewin2019floating,cotar2017equality}; in the quantum case the equality was open. The inequality $\tilde e_\UEG(\lambda)\ge\tilde e^{\rm GC}_\Jel(\lambda)$ is proved in~\cite[proof of Thm.~5.1, Step~4]{lewin2018statistical}, where $\tilde e^{\rm GC}_\Jel(\lambda)$ is the thermodynamic limit of grand-canonical Jellium, with free particle number and without confinement to the support of the background. This energy is attributed there to~\cite{lieb1975thermodynamic}, but~\cite{lieb1975thermodynamic} treats the canonical Jellium with Dirichlet boundary conditions and does not prove the identification of the two limits. Theorem~\ref{thm:lower} proves the inequality for the Jellium energy of Theorem~\ref{thm:LN_main}.
\end{remark}

\section{Main result: the Jellium and UEG energies coincide}\label{sec:main}

Sections~\ref{sec:LN} and~\ref{sec:UEG} define two energies per unit volume at density $\rho_0$. The Jellium energy $e^\hbar_\Jel(\rho_0)$ is canonical and exactly neutral, with Dirichlet boundary conditions and a fixed background, while the particle density is free. On the other hand, the UEG energy $e^\hbar_\UEG(\rho_0)$ is grand-canonical, with a prescribed density and no background. Our main result is that the two coincide.

\begin{theorem}[Jellium $=$ UEG]\label{thm:main}
For every integer spin multiplicity $q\ge1$, every $\rho_0>0$ and every $\hbar>0$,
\begin{equation}\label{eq:main}
e^\hbar_\Jel(\rho_0)=e^\hbar_\UEG(\rho_0).
\end{equation}
Equivalently, $\tilde e_\Jel(\lambda)=\tilde e_\UEG(\lambda)$ for all $\lambda>0$.
\end{theorem}

The two forms are equivalent by the scaling relations $e^\hbar_\bullet(\rho_0)=\rho_0^{4/3}\tilde e_\bullet(\hbar^2\rho_0^{1/3})$ (Theorems~\ref{thm:LN_main}(iii) and~\ref{thm:UEG}(iii)). Theorem~\ref{thm:main} is the quantum analog of the classical equality $e_\Jel=e_\UEG$~\cite{lewin2019floating,cotar2017equality}, and it confirms the expectation stated in~\cite[Sec.~2.2]{lewin2020local} (Remark~\ref{rem:UEG_vs_Jel}). In the quantum case, the literature only contained the inequality $\tilde e^{\rm GC}_\Jel(\lambda)\le\tilde e_\UEG(\lambda)$ for the grand-canonical Jellium of~\cite{lewin2018statistical} (see Remark~\ref{rem:UEG_vs_Jel}). The proof is given in Section~\ref{sec:proofs}, and its strategy is outlined in Section~\ref{sec:main_strategy}.

\begin{remark}[Consequences for Jellium]\label{rem:main_consequences}
By Theorem~\ref{thm:main}, the properties of $\tilde e_\UEG$ in Theorem~\ref{thm:UEG}(iii)--(vi) hold for the Jellium energy of Theorem~\ref{thm:LN_main}. In particular, $\tilde e_\Jel(\lambda)=c_\TF\lambda-c_{\mathrm D}+o(1)$ as $\lambda\to\infty$, and $\tilde e_\Jel(\lambda)\to e_\UEG$ as $\lambda\to0$, where $e_\UEG$ is the classical UEG energy of~\cite[Thm.~2.6]{lewin2018statistical}. The latter coincides with the classical Jellium energy~\cite{lewin2019floating,cotar2017equality}. The general shape of $e^\hbar_\UEG(\rho_0)=e^\hbar_\Jel(\rho_0)$ as a function of $\rho_0$ is sketched in Figure~\ref{fig:energy_shape}.
\end{remark}

\subsection{Strategy of the proof}\label{sec:main_strategy}

Section~\ref{sec:proofs} proves the two inequalities separately, each involving the canonical neutral Dirichlet energy of Theorem~\ref{thm:LN_main}.

\medskip
\noindent\emph{The lower bound $e^\hbar_\Jel(\rho_0)\le e^\hbar_\UEG(\rho_0)$} (Theorem~\ref{thm:lower}).
We start from minimizers of the Levy--Lieb functional at the density $\rho_0\iden_{B_{\ell-2\delta_\ell}}*\chi_{\delta_\ell}$, where $\supp\chi\subset B_1$ and $\delta_\ell=\ell^{-1/2}$. These are states supported in $B_\ell$ satisfying homogeneous Dirichlet boundary conditions. Their indirect energy is $e^\hbar_\UEG(\rho_0)|B_\ell|+o(\ell^3)$ by Theorem~\ref{thm:UEG}(ii), and the self-energy of their discrepancy from $\rho_0\iden_{B_\ell}$ is $o(\ell^3)$. Thus~\eqref{eq:conj_identity} gives the same upper bound on their Jellium energy.

Since Jellium energy is affine, one particle-number sector has energy no greater than this average. Its excess charge is $\mathcal O(\ell^2)$ by the Onsager bound with the positive Coulomb term retained. Lemma~\ref{lem:lower_completion} then adds the missing background and particles in an annulus of fixed thickness, with an upper energy cost $\mathcal O(\ell^2)$. The completed state is exactly neutral and canonical in a ball of radius $\ell+\mathcal O(1)$. The canonical limit on balls of~\cite{lieb1975thermodynamic} proves the lower bound. The shell construction is given in Appendix~\ref{sec:canonical_completion} and does not require the density of the selected sector to be radial.

\medskip
\noindent\emph{The upper bound $e^\hbar_\UEG(\rho_0)\le e^\hbar_\Jel(\rho_0)$} (Theorem~\ref{thm:upper}).
Here Jellium states must be turned into states with a prescribed constant density, at no extensive cost in energy. The Jellium ground state is not expected to have constant density; at low density it should form a Wigner crystal. Averaging a crystal over translations (the ``floating crystal'') gives a constant density, but because of boundary charge fluctuations, the indirect energy per unit volume of a floating crystal can be strictly larger than its Jellium energy~\cite[Appendix~B]{lewin2015improved}. In the classical case this obstruction was overcome by constructing a floating Wigner crystal without boundary charge fluctuations~\cite{lewin2019floating}.

The proof in Section~\ref{sec:upper} does not use a crystal. A large cube is filled, up to a fraction $\varepsilon$ of its volume, by a periodic family of balls with room between them (Lemma~\ref{lem:upper_packing}). Each ball carries a canonical Dirichlet Jellium ground state, averaged over rotations, so that its charge relative to the background is neutral and radial. The rest of the cube is filled by a quasi-free state with prescribed density, taken from~\cite{lewin2020local} (Theorem~\ref{thm:basic_fluid}). Its density must vanish at the boundary of each Dirichlet ball. A compensated shell (Lemma~\ref{lem:upper_collar}) lets this density vanish smoothly near the ball and returns the missing mass in its outer part. The shell charge is neutral and radial, and its kinetic and Coulomb energies are proportional to the surface area of the ball. Consequently the cross Coulomb terms between the resulting neutral charge distributions vanish, and the Jellium energy of the product state is given by the exact identity~\eqref{eq:upper_exact_identity}.

The density of each such state is not constant. It becomes constant in the interior after averaging over translations of the packing, with the outer background held fixed. The average has Jellium energy equal to the average of the Jellium energies, and~\eqref{eq:conj_identity} bounds its indirect energy by this quantity. Its density equals $\rho_0$ except in a boundary layer of fixed width, where it is bounded by a fixed multiple of $\rho_0$. Lemma~\ref{lem:upper_bounded_layer}, which relies on the localization bound of~\cite{lewin2020local}, shows that such a layer does not lower the UEG energy per unit volume. The parameters are chosen in a fixed order. For a given accuracy, a minimum ball radius is fixed first, and only then a finite packing whose balls all have at least this radius. A finer packing therefore needs more balls and a larger period cell, but never smaller balls: it is obtained by dilating a packing of the unit cube until its smallest ball has the minimum radius. This is crucial because the kinetic and Coulomb energies of the shells are proportional to the total surface area of the balls, and this area per unit volume is at most a fixed multiple of the inverse minimum radius, independently of the packing. In the thin gaps between the shells the quasi-free state has constant density, so the gradient term in the kinetic energy bound of Theorem~\ref{thm:basic_fluid} vanishes there, and the remaining Thomas--Fermi term is small because the uncovered volume is small. The quantities that do depend on the packing, such as the period and the width of the boundary layer, are harmless, because the size of the cube tends to infinity with the packing held fixed. Only afterwards is the accuracy improved, with the minimum radius sent to infinity.

\section{Proofs}\label{sec:proofs}

In this section we prove Theorem~\ref{thm:main}, following the outline in Section~\ref{sec:main_strategy}. Section~\ref{sec:basic} collects the general facts used in the proofs. The lower bound is Theorem~\ref{thm:lower} and the upper bound is Theorem~\ref{thm:upper}. Both concern exactly the canonical neutral Dirichlet energy of Theorem~\ref{thm:LN_main}, so together they prove~\eqref{eq:main} for every integer $q\ge1$ and $\rho_0>0$ at $\hbar=1$, which suffices by the reduction below. No identification of unconfined grand-canonical Jellium with the canonical Dirichlet problem is assumed.

By the scaling relations $e^\hbar_\bullet(\rho_0)=\hbar^{-8}e^1_\bullet(\hbar^6\rho_0)$ for both energies (Theorems~\ref{thm:LN_main}(iii) and~\ref{thm:UEG}(iii)), it suffices to prove~\eqref{eq:main} for $\hbar=1$ and all $\rho_0>0$. To reduce clutter, in this section and in Appendices~\ref{sec:LN_proof} and~\ref{sec:canonical_completion} we set $\hbar=1$ and drop it from the notation, writing $E(\rho)$, $F_{\rm LL}(\rho)$, $H^{\rho_0}_\Omega$, $E_\Jel(\rho_0,\Omega)$, $e_\Jel(\rho_0)$ and $e_\UEG(\rho_0)$.

\subsection{Basic facts}\label{sec:basic}

For a grand-canonical state $\vGamma$, write
\[
\TC(\vGamma)=\sum_{n\ge1}\Tr\left(\sum_{i=1}^n-\Delta_{x_i}\right)\Gamma^{(n)},
\qquad
\Cc(\vGamma)=\sum_{n\ge2}\Tr\left(\sum_{1\le i<j\le n}\frac1{|x_i-x_j|}\right)\Gamma^{(n)},
\]
and define its indirect energy
\[
\EC(\vGamma)=\TC(\vGamma)+\Cc(\vGamma)-D(\rho_{\vGamma}).
\]
By~\eqref{eq:UEG_E}, $E(\rho)=\inf_{\rho_{\vGamma}=\rho}\EC(\vGamma)$. We collect the general facts used in this section and in Appendix~\ref{sec:LN_proof}. An $N$-particle state, that is, a density matrix $\Gamma$ on $\Hil^N$, is identified with the grand-canonical state whose only nonzero sector is $\Gamma^{(N)}=\Gamma$. We keep non-bold letters for such canonical states.

We use the Lieb--Oxford inequality to bound the energy of localized densities near the boundary from below.

\begin{theorem}[{Lieb--Oxford~\cite{lieb1981improved}, \cite[Sec.~3.1]{lewin2020local}}]\label{thm:basic_lieb_oxford}
There is a universal constant $c_{\rm LO}>0$ such that every grand-canonical state $\vGamma$ with $\rho_{\vGamma}\in L^1\cap L^{4/3}(\RR^3)$ satisfies
\[
\Cc(\vGamma)-D(\rho_{\vGamma})\ge-c_{\rm LO}\int_{\RR^3}\rho_{\vGamma}^{4/3}.
\]
\end{theorem}

The following lemma gives the gradient bounds needed to localize densities and control the kinetic energy
of a smeared indicator.

\begin{lemma}\label{lem:basic_elementary}
The following hold true.
\begin{enumerate}[(i)]
\item Let $\Omega\subset\RR^3$ be measurable, and let $\eta\ge0$ satisfy $\int_{\RR^3}\eta=1$ and $\sqrt\eta\in H^1(\RR^3)$. Then $0\le\iden_\Omega*\eta\le1$ and
\begin{equation}\label{eq:basic_sqrt_cutoff}
|\nabla\sqrt{\iden_\Omega*\eta}|^2
\le\iden_\Omega*|\nabla\sqrt\eta|^2
\le\int_{\RR^3}|\nabla\sqrt\eta|^2
\quad\hbox{a.e.}
\end{equation}
In particular, $\sqrt{\iden_\Omega*\eta}\in W^{1,\infty}(\RR^3)$. If $\rho\ge0$ and $\sqrt\rho\in H^1(\RR^3)$, then $\sqrt{(\iden_\Omega*\eta)\rho}\in H^1(\RR^3)$ and
\[
\int_{\RR^3}|\nabla\sqrt{(\iden_\Omega*\eta)\rho}|^2
\le2\int_{\RR^3}|\nabla\sqrt\rho|^2
+2\|\nabla\sqrt{\iden_\Omega*\eta}\|_\infty^2\int_{\RR^3}\rho.
\]
\item Let $\Omega\subset\RR^3$ be a bounded Lipschitz domain, and let $\eta$ satisfy the assumptions in~(i), with $\supp\eta\subset B_1$. Set $\eta_\delta(x)=\delta^{-3}\eta(x/\delta)$. Then
\[
\int_{\RR^3}|\nabla\sqrt{\iden_\Omega*\eta_\delta}|^2
\le C\frac{|\partial\Omega|}{\delta}
\]
for all sufficiently small $\delta>0$, where $|\partial\Omega|=\mathcal H^2(\partial\Omega)$ is the surface area and $C$ depends only on $\eta$ and the Lipschitz geometry of $\Omega$.
\end{enumerate}
\end{lemma}

We use Newton's theorem to compute radial Coulomb potentials and show that neutral radial charges do not interact with charges outside their supporting balls.

\begin{theorem}[{Newton's theorem~\cite[Thm.~5.2]{lieb2010stability}}]\label{thm:basic_newton}
Let $f\in L^1(\RR^3)$ be radial about $a\in\RR^3$ and supported in $\overline{B_r(a)}$, and write $f(s)$ for its value at distance $s$ from $a$. Then, for $x\ne a$,
\begin{equation}\label{eq:basic_newton}
\frac{1}{4\pi}\int_{\RR^3}\frac{f(y)}{|x-y|}\,\mathrm dy=\frac1{|x-a|}\int_0^{|x-a|}f(s)s^2\,\mathrm ds+\int_{|x-a|}^\infty f(s)s\,\mathrm ds.
\end{equation}
Consequently:
\begin{enumerate}[(i)]
\item For $|x-a|\ge r$,
\[
\int_{\RR^3}\frac{f(y)}{|x-y|}\,\mathrm dy=\frac1{|x-a|}\int_{\RR^3}f.
\]
If $\int_{\RR^3}f=0$, the left side vanishes outside $B_r(a)$, and $D(f,g)=0$ for every $g\in L^1(\RR^3)$ supported outside $B_r(a)$.
\item If $f\ge0$, then $\int_{\RR^3}f(y)|x-y|^{-1}\,\mathrm dy\le|x-a|^{-1}\int_{\RR^3}f$ for all $x\ne a$.
\item If $g\in L^1(\RR^3)$ is radial about $b\in\RR^3$ and supported in a closed ball $\overline{B_{r'}(b)}$ disjoint from $\overline{B_r(a)}$, then
\begin{equation}\label{eq:basic_radial_cross}
2D(f,g)=\frac1{|a-b|}\int_{\RR^3}f\int_{\RR^3}g.
\end{equation}
\item If $f$ is bounded, then
\[
D(f)=\frac12\int_0^\infty\frac1{s^2}\Big(\int_{B_s(a)}f\Big)^2\,\mathrm ds.
\]
\end{enumerate}
\end{theorem}

The following form of Onsager's inequality retains the positive Coulomb term needed to control the excess charge in Appendix~\ref{sec:canonical_completion}.

\begin{lemma}[{Onsager~\cite{onsager1939electrostatic}; cf.~\cite[Lemma~6.1]{lieb2010stability}}]\label{lem:basic_onsager}
Let $\rho\in L^1(\RR^3)$ satisfy $0\le\rho\le\bar\rho$, and let $x_1,\dots,x_N\in\RR^3$. For $a>0$, let $\mu_j:=|B_a|^{-1}\iden_{B_a(x_j)}$. Then
\begin{multline*}
\sum_{1\le j<k\le N}\frac1{|x_j-x_k|}-\sum_{j=1}^N\int_{\RR^3}\frac{\rho(y)}{|x_j-y|}\,\mathrm dy+D(\rho)\\
\ge D\Big(\sum_{j=1}^N\mu_j-\rho\Big)-N\Big(\frac3{5a}+\frac{2\pi}5\bar\rho a^2\Big).
\end{multline*}
In particular, the left side is at least $-N\big(\frac3{5a}+\frac{2\pi}5\bar\rho a^2\big)$.
\end{lemma}

\begin{proof}
For $\nu=\sum_j\mu_j-\rho$,
\[
D(\nu)=\sum_jD(\mu_j)+2\sum_{j<k}D(\mu_j,\mu_k)-2\sum_jD(\mu_j,\rho)+D(\rho).
\]
The uniform probability measure on a ball of radius $a$ has $D(\mu_j)=\frac3{5a}$. By Theorem~\ref{thm:basic_newton}(ii), the potential of $\mu_j$ is at most $|x-x_j|^{-1}$, so $2D(\mu_j,\mu_k)\le|x_j-x_k|^{-1}$. For the background term, let $\mu_a:=|B_a|^{-1}\iden_{B_a(0)}$. The function $|z|^{-1}-\int_{\RR^3}\mu_a(y)|z-y|^{-1}\,\mathrm dy$ is nonnegative by Theorem~\ref{thm:basic_newton}(ii) and vanishes for $|z|\ge a$ by Theorem~\ref{thm:basic_newton}(i). Hence
\begin{align*}
\int_{\RR^3}\frac{\rho(y)}{|x_j-y|}\,\mathrm dy-2D(\mu_j,\rho)
&=\int_{\RR^3}\rho(x_j-z)\Big(\frac1{|z|}-\int_{\RR^3}\frac{\mu_a(y)}{|z-y|}\,\mathrm dy\Big)\,\mathrm dz\\
&\le\bar\rho\int_{B_a}\Big(\frac1{|z|}-\int_{\RR^3}\frac{\mu_a(y)}{|z-y|}\,\mathrm dy\Big)\,\mathrm dz
=\frac{2\pi}5\bar\rho a^2.
\end{align*}
Inserting these three facts into the expansion of $D(\nu)$ gives the claim, and $D(\nu)\ge0$ gives the last assertion.
\end{proof}

The proofs of the lower bound (Theorem~\ref{thm:lower}) and the upper bound (Theorem~\ref{thm:upper}) both use grand-canonical states. Let $\zeta\ge0$ be a background with $\zeta\in L^1\cap L^\infty(\RR^3)$. Consider a grand-canonical state $\vGamma=(\Gamma^{(n)})$ as in Section~\ref{sec:notation}, with finite kinetic energy and $\int_{\RR^3}\rho_{\vGamma}<\infty$. Its Jellium energy is
\begin{multline*}
\mathcal{J}_\zeta(\vGamma):=\sum_{n\ge1}\Tr_{\Hil^n}\bigg[\sum_{j=1}^n\Big(-\Delta_{x_j}-\int_{\RR^3}\frac{\zeta(y)}{|x_j-y|}\,\mathrm dy\Big)+\sum_{1\le j<k\le n}\frac{1}{|x_j-x_k|}\bigg]\Gamma^{(n)}\\
+D(\zeta).
\end{multline*}
Here, the second term is simply $-2D(\rho_{\vGamma},\zeta)$, so
\begin{equation}\label{eq:conj_identity}
\mathcal{J}_\zeta(\vGamma)=\EC(\vGamma)+D(\rho_{\vGamma})-2D(\rho_{\vGamma},\zeta)+D(\zeta)=\EC(\vGamma)+D(\rho_{\vGamma}-\zeta).
\end{equation}
For a fixed background $\zeta$, the left side of~\eqref{eq:conj_identity} is affine in $\vGamma$. On the right, only the sum is affine: the indirect energy $\EC(\vGamma)$ is concave in $\vGamma$, because $D(\rho_{\vGamma})$ is convex, while the direct term $D(\rho_{\vGamma}-\zeta)$ is convex. Averaging trial states with a common background therefore gives an exact formula for the Jellium energy of the average, and dropping the direct term in~\eqref{eq:conj_identity} turns it into an upper bound on the indirect energy of the average. This will determine how we average the final trial states.

Let $A_1,\ldots,A_J\subset\RR^3$ be pairwise disjoint measurable sets. Put $\mathfrak h_j=L^2(A_j;\CC^q)$ and let $\Fock$ denote fermionic Fock space. Since the spaces $\mathfrak h_j$ are mutually orthogonal, there is a canonical unitary
\[
\mathcal U:\bigotimes_{1\le j\le J}\Fock(\mathfrak h_j)\to\Fock\Bigl(\bigoplus_{1\le j\le J}\mathfrak h_j\Bigr),
\qquad
\mathcal U(\Psi_1\otimes\cdots\otimes\Psi_J)=\Psi_1\wedge\cdots\wedge\Psi_J,
\]
where $\Psi_j\in\bigwedge^{n_j}\mathfrak h_j$ and $\Psi_1\wedge\cdots\wedge\Psi_J$ is the antisymmetrized product, normalized to have norm one when every $\Psi_j$ does. For number-conserving fermionic states $\vGamma_j$ supported in $A_j$, define the antisymmetrized product state
\[
\vGamma=\vGamma_1\wedge\cdots\wedge\vGamma_J=\mathcal U(\vGamma_1\otimes\cdots\otimes\vGamma_J)\,\mathcal U^*,
\]
with the vacuum on the orthogonal complement of $\bigoplus_{1\le j\le J}\mathfrak h_j$ in $L^2(\RR^3;\CC^q)$. For pure states $\vGamma_j=|\Psi_j\rangle\langle\Psi_j|$, it is $|\Psi_1\wedge\cdots\wedge\Psi_J\rangle\langle\Psi_1\wedge\cdots\wedge\Psi_J|$.

\begin{lemma}[Antisymmetrized product of disjoint fermionic states]\label{lem:basic_product}
Let the states $\vGamma_1,\ldots,\vGamma_J$ be number-conserving, fermionic and disjointly supported, with $\vGamma_j$ supported in $A_j$, finite mean particle number and finite kinetic and interaction energies after extension to $\RR^3$. Then $\vGamma=\vGamma_1\wedge\cdots\wedge\vGamma_J$ is a number-conserving fermionic state and
\begin{align}
\rho_{\vGamma}&=\sum_{1\le j\le J}\rho_{\vGamma_j},\label{eq:basic_product_density}\\
\TC(\vGamma)&=\sum_{1\le j\le J}\TC(\vGamma_j),\label{eq:basic_product_kinetic}\\
\Cc(\vGamma)&=\sum_{1\le j\le J}\Cc(\vGamma_j)+2\sum_{1\le j<k\le J}D(\rho_{\vGamma_j},\rho_{\vGamma_k}).
\label{eq:basic_product_interaction}
\end{align}
\end{lemma}
Consequently, for every decomposition $\zeta=\sum_{1\le j\le J}\zeta_j$ into $\zeta_j\ge0$ in $L^1\cap L^\infty(\RR^3)$,
\begin{equation}\label{eq:basic_product_J}
\mathcal J_\zeta(\vGamma)
=\sum_{1\le j\le J}\mathcal J_{\zeta_j}(\vGamma_j)
+2\sum_{1\le j<k\le J}D(\rho_{\vGamma_j}-\zeta_j,\rho_{\vGamma_k}-\zeta_k).
\end{equation}
In fact, by definition $\mathcal J_\zeta(\vGamma)=\TC(\vGamma)+\Cc(\vGamma)-2D(\rho_{\vGamma},\zeta)+D(\zeta)$. Inserting~\eqref{eq:basic_product_density}--\eqref{eq:basic_product_interaction} and $\zeta=\sum_{1\le j\le J}\zeta_j$, the terms with a single index $j$ add up to $\mathcal J_{\zeta_j}(\vGamma_j)$, while for $j<k$ the cross terms $2D(\rho_{\vGamma_j},\rho_{\vGamma_k})-2D(\rho_{\vGamma_j},\zeta_k)-2D(\rho_{\vGamma_k},\zeta_j)+2D(\zeta_j,\zeta_k)$ combine to $2D(\rho_{\vGamma_j}-\zeta_j,\rho_{\vGamma_k}-\zeta_k)$.

We recall the following upper bound of Lewin, Lieb and Seiringer \cite[Thm.~8, Rem.~9 and Cor.~11]{lewin2020local}.

\begin{theorem}\label{thm:basic_fluid}
Let $\rho\ge0$ be such that $\sqrt\rho\in H^1(\RR^3)$, and let $0<\varepsilon\le1$. There is a number-conserving grand-canonical quasi-free fermionic state $\vGamma$ with $\rho_\vGamma=\rho$ such that
\begin{equation}\label{eq:basic_fluid_bound}
\EC(\vGamma)\le\TC(\vGamma)
\le\Big(1+\frac\varepsilon{15}\Big)c_\TF\int_{\RR^3}\rho^{5/3}
+\frac{19}{\varepsilon}
\int_{\RR^3}|\nabla\sqrt\rho|^2.
\end{equation}
\end{theorem}

We also recall the localization lower bound of~\cite{lewin2020local}. Choose a nonnegative radial function $\eta\in C_c^\infty(\RR^3)$ with $\int_{\RR^3}\eta=1$, $\supp\eta\subset B_1$, and $\sqrt\eta\in C_c^\infty(\RR^3)$. Such a function is obtained by normalizing the square of a nonnegative radial test function. For a smearing scale $\delta>0$ put $\eta_\delta=\delta^{-3}\eta(\cdot/\delta)$. Then $\eta_\delta$ is radial, $\int_{\RR^3}\eta_\delta=1$, $\supp\eta_\delta\subset B_\delta$.

We now describe the tetrahedral tiling of Graf and Schenker~\cite{graf1995molecular}, in the notation of~\cite[Sec.~4.2]{lewin2020local}. The unit cube $C_1$ is the union, up to null sets, of 24 disjoint congruent tetrahedra $\mu_j\tetra$, where $j=1,\ldots,24$. Here $\tetra$ is a reference tetrahedron centered at the origin and $|\tetra|=1/24$, and $\mu_j$ is an isometry $\mu_jx=R_jx-z_j$ for some $R_j\in\mathrm{SO}(3)$ and $z_j\in C_1$. Fix a scale $h>0$. For a rotation $Q\in\mathrm{SO}(3)$ and a translation $a\in C_h$, the tetrahedra
\[
Q^{-1}(h\mu_j\tetra+h z+a),
\qquad
z\in\ZZ^3,\ 1\le j\le24,
\]
tile $\RR^3$ up to null sets. These tetrahedra are all congruent to $h\tetra$, so they have volume $|h\tetra|=h^3/24$ and diameter $d_0h$, where $d_0:=\operatorname{diam}\tetra$. Convolution of the tiling identity with $\eta_\delta$ gives
\begin{equation}\label{eq:basic_partition}
\sum_{z\in\ZZ^3}\sum_{j=1}^{24}\iden_{Q^{-1}(h\mu_j\tetra+h z+a)}*\eta_\delta=1
\quad\hbox{a.e.\ in }\RR^3.
\end{equation}
We use normalized Haar measure $\mathrm dQ$ on $\mathrm{SO}(3)$. Using the IMS localization formula and the smeared Graf--Schenker inequality~\cite[Lemma~6]{graf1995molecular} one obtains
\begin{theorem}[Localization on a tetrahedral tiling {\cite[Prop.~16]{lewin2020local}}]\label{thm:basic_localization}
There is a constant $C>0$ such that for every $\rho\ge0$ with $\sqrt\rho\in H^1(\RR^3)$, every $\delta>0$ and every $h>C\delta$,
\begin{multline}\label{eq:basic_localization}
E(\rho)\ge
\left(1-\frac{C\delta}h\right)
\int_{\mathrm{SO}(3)}\frac1{h^3}\int_{C_h}
\sum_{z\in\ZZ^3}\sum_{j=1}^{24}
E\big((\iden_{Q^{-1}(h\mu_j\tetra+h z+a)}*\eta_\delta)\rho\big)
\,\mathrm da\,\mathrm dQ
\\
-\frac{C}{h}\int_{\RR^3}(\rho+\delta^3\rho^2)
-\frac{C}{h\delta}\int_{\RR^3}\rho.
\end{multline}
\end{theorem}

\subsection{The lower bound}\label{sec:lower}

The following lemma allows us to complete a confined state to neutrality by adding particles and background charge in a shell of bounded thickness, with an energy increase bounded by a constant times the surface area.
\begin{lemma}[Neutralization in a slightly larger ball]\label{lem:lower_completion}
Fix $\rho_0>0$ and $C_0>0$. There exist constants $h,C_1>0$, depending only on $\rho_0$ and $C_0$, such that
whenever  $\ell\ge1$ and a normalized $N$-particle state $\Gamma_N$ supported in $B_\ell$ with homogeneous Dirichlet boundary conditions verifies the bound
\[
\mathcal J_{\rho_0\iden_{B_\ell}}(\Gamma_N)\le C_0\ell^3,
\]
there exist a radius $\ell'\in[\ell,\ell+h]$ with $M=\rho_0|B_{\ell'}|\in\NN$ (neutrality) and a normalized $M$-particle state $\Gamma_M$ supported in $B_{\ell'}$ satisfying homogeneous Dirichlet boundary conditions such that
\begin{equation}\label{eq:lower_completion}
\mathcal J_{\rho_0\iden_{B_{\ell'}}}(\Gamma_M)
\le\mathcal J_{\rho_0\iden_{B_\ell}}(\Gamma_N)+C_1\ell^2.
\end{equation}
\end{lemma}
The proof can be found in Appendix~\ref{sec:canonical_completion}.

\begin{theorem}\label{thm:lower}
$e_\Jel(\rho_0)\le e_\UEG(\rho_0)$
\end{theorem}

\begin{proof}
Choose a radial smoothing function $\chi\ge0$ with $\int_{\RR^3}\chi=1$, $\supp\chi\subset B_1$ and $\sqrt\chi\in H^1(\RR^3)$. For $\ell\to\infty$, put
\[
\delta_\ell=\ell^{-1/2},\qquad r_\ell=\ell-2\delta_\ell.
\]
The conditions $\delta_\ell/r_\ell\to0$ and $\delta_\ell r_\ell\to\infty$ allow us to apply Theorem~\ref{thm:UEG}(ii). Since $|B_{r_\ell}|/|B_\ell|\to1$, we obtain
\[
E\big(\rho_0\iden_{B_{r_\ell}}*\chi_{\delta_\ell}\big)=e_\UEG(\rho_0)|B_\ell|+o(\ell^3).
\]
By~\cite[Thm.~4]{lewin2019universal}, the infimum defining $F_{\rm LL}(\rho_0\iden_{B_{r_\ell}}*\chi_{\delta_\ell})$ is attained. Let $\vGamma_\ell$ be a minimizer, so that $\rho_{\vGamma_\ell}=\rho_0\iden_{B_{r_\ell}}*\chi_{\delta_\ell}$ and $\EC(\vGamma_\ell)=E(\rho_{\vGamma_\ell})$ by~\eqref{eq:UEG_E}. Clearly, $\supp\rho_{\vGamma_\ell}\subset B_{\ell-\delta_\ell}$.

The density satisfies $0\le\rho_{\vGamma_\ell}\le\rho_0\iden_{B_\ell}$ and $\int_{\RR^3}\rho_{\vGamma_\ell}=\rho_0|B_{r_\ell}|$. Consequently, for $f_\ell=\rho_{\vGamma_\ell}-\rho_0\iden_{B_\ell}$,
\begin{align*}
\|f_\ell\|_1&=\rho_0(|B_\ell|-|B_{r_\ell}|)\le C\ell^2\delta_\ell,
\qquad \|f_\ell\|_\infty\le\rho_0,\\
D(f_\ell)&\le C(\ell^2\delta_\ell)^{5/3}=C\ell^{5/2}=o(\ell^3),
\end{align*}
where we used $D(f)\le D(|f|)\le C\|f\|_1^{5/3}\|f\|_\infty^{1/3}$ for real $f\in L^1\cap L^\infty(\RR^3)$. Combining this with the identity~\eqref{eq:conj_identity} gives
\begin{equation}\label{eq:lower_trial}
\mathcal J_{\rho_0\iden_{B_\ell}}(\vGamma_\ell)
=e_\UEG(\rho_0)|B_\ell|+o(\ell^3).
\end{equation}
Set $p_n=\Tr\Gamma_\ell^{(n)}$. Since $\vGamma\mapsto \mathcal J_\zeta(\vGamma)$ is affine for fixed $\zeta$ we obtain
\[
\mathcal J_{\rho_0\iden_{B_\ell}}(\vGamma_\ell)=\sum_{n:\,p_n>0}p_n\,\mathcal J_{\rho_0\iden_{B_\ell}}\big(\Gamma_\ell^{(n)}/p_n\big),\qquad\sum_n p_n=1,
\]
so the left side of~\eqref{eq:lower_trial} is an average of the Jellium energies of the normalized sectors.
Hence there is $N_\ell$ with $p_{N_\ell}>0$ such that the $N_\ell$-particle state $\Gamma_{N_\ell}=\Gamma_\ell^{(N_\ell)}/p_{N_\ell}$ satisfies $\mathcal J_{\rho_0\iden_{B_\ell}}(\Gamma_{N_\ell})\le\mathcal J_{\rho_0\iden_{B_\ell}}(\vGamma_\ell)$.

Next, we check the hypotheses of Lemma~\ref{lem:lower_completion}. Since $\vGamma_\ell$ has finite kinetic energy and its density vanishes outside $B_{\ell-\delta_\ell}$, every occupied sector has finite kinetic energy and is supported in this smaller ball. In particular, $\Gamma_{N_\ell}$ satisfies homogeneous Dirichlet boundary conditions in $B_\ell$. Fix $C_0=1+|B_1||e_\UEG(\rho_0)|$.
By~\eqref{eq:lower_trial} and the choice of $N_\ell$, for all sufficiently large $\ell$ we have
\[
\mathcal J_{\rho_0\iden_{B_\ell}}(\Gamma_{N_\ell})
\le\mathcal J_{\rho_0\iden_{B_\ell}}(\vGamma_\ell)
\le C_0\ell^3.
\]
Thus Lemma~\ref{lem:lower_completion} applies with the same $\rho_0$ and $C_0$ for all these $\ell$. It provides constants $h,C_1>0$ independent of $\ell$ and, for each $\ell$, a radius $\ell'\in[\ell,\ell+h]$ with $M_\ell=\rho_0|B_{\ell'}|\in\NN$ and a normalized $M_\ell$-particle state $\Gamma_{M_\ell}$ supported in $B_{\ell'}$ satisfying homogeneous Dirichlet boundary conditions. This is a trial state for the canonical neutral Jellium problem in $B_{\ell'}$, so the variational principle and~\eqref{eq:lower_completion} give
\begin{align*}
E_\Jel(\rho_0,B_{\ell'})
&\le\mathcal J_{\rho_0\iden_{B_{\ell'}}}(\Gamma_{M_\ell})\\
&\le\mathcal J_{\rho_0\iden_{B_\ell}}(\Gamma_{N_\ell})+C_1\ell^2\\
&\le\mathcal J_{\rho_0\iden_{B_\ell}}(\vGamma_\ell)+C_1\ell^2\\
&=e_\UEG(\rho_0)|B_\ell|+o(\ell^3).
\end{align*}
Here the last line uses~\eqref{eq:lower_trial} and $C_1\ell^2=o(\ell^3)$, since $C_1$ is independent of $\ell$. Moreover, $0\le\ell'-\ell\le h$ implies $\ell'/\ell\to1$. Dividing by $|B_{\ell'}|$ therefore gives
\[
\frac{E_\Jel(\rho_0,B_{\ell'})}{|B_{\ell'}|}
\le e_\UEG(\rho_0)\left(\frac\ell{\ell'}\right)^3+o(1).
\]
The left side tends to $e_\Jel(\rho_0)$ by Theorem~\ref{thm:LN_main}(ii) applied to the neutral balls $B_{\ell'}$ (Remark~\ref{rem:LN_balls}), while the right side tends to $e_\UEG(\rho_0)$. This proves the claim.
\end{proof}

\subsection{The upper bound}\label{sec:upper}

The purpose of this section is to prove

\begin{theorem}\label{thm:upper}
$e_\UEG(\rho_0)\le e_\Jel(\rho_0)$
\end{theorem}

Its proof is given after Lemma~\ref{lem:upper_packing} below. We first control densities with a bounded transition layer.

\begin{lemma}[A bounded transition layer]\label{lem:upper_bounded_layer}
Fix $\rho_0>0$ and $\Lambda\ge1$. For all sufficiently large $\ell>0$, let $\delta'=\delta'(\ell)\ge0$ be the width of the transition layer, and let $\rho_\ell\ge 0$ satisfy $\sqrt{\rho_\ell}\in H^1(\RR^3)$ and
\[
0\le\rho_\ell\le \Lambda\rho_0\iden_{C_\ell+B_{\delta'}},
\qquad
\rho_\ell=\rho_0\quad\hbox{a.e.\ on }C_\ell\setminus(\partial C_\ell+B_{\delta'}).
\]
If $\delta'/\ell\to0$ as $\ell\to\infty$, then
\begin{equation}\label{eq:upper_bounded_layer}
e_\UEG(\rho_0)\le \liminf_{\ell\to\infty}\frac{E(\rho_\ell)}{\ell^3}.
\end{equation}
\end{lemma}

\begin{proof}
All constants in this proof are independent of $\ell$, of $\delta\in(0,1]$, and of the translation and rotation of the tiling. Dependence on $\Lambda$ and $\rho_0$ is indicated where it enters.

We use the tiling and the partition~\eqref{eq:basic_partition} introduced before Theorem~\ref{thm:basic_localization}. By Lemma~\ref{lem:basic_elementary}(i), every localized density $(\iden_{Q^{-1}(h\mu_j\tetra+h z+a)}*\eta_\delta)\rho_\ell$ is admissible for $E$. Fix $\delta\in(0,1]$ and $h>C\delta$, where $C$ is the constant of Theorem~\ref{thm:basic_localization}, so that $1-C\delta/h>0$. Since $\sqrt{\rho_\ell}\in H^1(\RR^3)$, Theorem~\ref{thm:basic_localization} gives
\begin{multline}\label{eq:upper_layer_localization}
E(\rho_\ell)\ge
\left(1-\frac{C\delta}h\right)
\int_{\mathrm{SO}(3)}\frac1{h^3}\int_{C_h}
\sum_{z\in\ZZ^3}\sum_{j=1}^{24}
E\big((\iden_{Q^{-1}(h\mu_j\tetra+h z+a)}*\eta_\delta)\rho_\ell\big)
\,\mathrm da\,\mathrm dQ
\\
-\frac{C}{h}\int_{\RR^3}(\rho_\ell+\delta^3\rho_\ell^2)
-\frac{C}{h\delta}\int_{\RR^3}\rho_\ell.
\end{multline}

Put
\[
h'=\delta'+d_0h+2\delta+1.
\]
Since $\delta'=o(\ell)$, for all sufficiently large $\ell$ we have $h'<\ell/4$. Comparing the cubes with side lengths $\ell+2h'$ and $\ell-2h'$ yields
\begin{equation}\label{eq:upper_boundary_volume}
|\partial C_\ell+B_{h'}|\le C \ell^2h'.
\end{equation}
For the same $\ell$, the set $C_\ell+B_{\delta'}$ lies in the cube $C_{\ell+2h'}$, whose volume is at most $C\ell^3$. The support and pointwise bounds on $\rho_\ell$ therefore control the integral terms in~\eqref{eq:upper_layer_localization}:
\begin{equation}\label{eq:upper_layer_integrals}
\int_{\RR^3}\rho_\ell\le C\Lambda\rho_0\ell^3,
\qquad
\int_{\RR^3}(\rho_\ell+\delta^3\rho_\ell^2)
\le C(\Lambda\rho_0+\delta^3\Lambda^2\rho_0^2)\ell^3.
\end{equation}

To treat the first term of \eqref{eq:upper_layer_localization}, we introduce the index set of tetrahedra that meet the enlarged cube $C_\ell+B_{\delta'+\delta}$,
\[
J=\big\{(z,j)\in\ZZ^3\times\{1,\dots,24\}:Q^{-1}(h\mu_j\tetra+h z+a)\cap(C_\ell+B_{\delta'+\delta})\ne\emptyset\big\}
\]
 If $(z,j)\notin J$, the localized density $(\iden_{Q^{-1}(h\mu_j\tetra+h z+a)}*\eta_\delta)\rho_\ell$ vanishes,
 because the smeared indicator of the tetrahedron is supported in its $\delta$-enlargement and $\rho_\ell$ vanishes outside
 $C_\ell+B_{\delta'}$. Let further
\[
J_0=\big\{(z,j)\in\ZZ^3\times\{1,\dots,24\}: Q^{-1}(h\mu_j\tetra+h z+a)\subset C_\ell\setminus(\partial C_\ell+B_{\delta'+\delta})\big\}
\]
be the set of interior pairs. Clearly $J_0\subset J$. Both sets depend on $\ell$, $Q$ and $a$; see Figure~\ref{fig:bounded_layer}.

\begin{figure}[ht]
\centering
\begin{tikzpicture}[>=stealth]
  \def\l{0.6}
  \begin{scope}
    \clip (-3.7,-3.7) rectangle (3.7,3.7);
    \begin{scope}[rotate=18.0,shift={(0.31,0.17)}]
      \foreach \x/\y/\j in {-4/-2/0,-4/-2/1,-4/-2/2,-4/-2/3,-4/-1/0,-4/-1/3,-4/0/0,-3/-3/1,-3/-2/0,-3/-2/1,-3/-2/2,-3/-2/3,-3/-1/0,-3/-1/1,-3/-1/2,-3/-1/3,-3/0/0,-3/0/1,-3/0/2,-3/0/3,-3/1/0,-3/1/1,-3/1/2,-3/1/3,-3/2/0,-3/2/3,-3/3/0,-2/-3/0,-2/-3/1,-2/-2/0,-2/-2/1,-2/-2/2,-2/-2/3,-2/-1/0,-2/-1/1,-2/-1/2,-2/-1/3,-2/0/0,-2/0/1,-2/0/2,-2/0/3,-2/1/0,-2/1/1,-2/1/2,-2/1/3,-2/2/0,-2/2/1,-2/2/2,-2/2/3,-2/3/0,-2/3/1,-2/3/2,-2/3/3,-1/-3/0,-1/-3/1,-1/-3/2,-1/-3/3,-1/-2/0,-1/-2/1,-1/-2/2,-1/-2/3,-1/-1/0,-1/-1/1,-1/-1/2,-1/-1/3,-1/0/0,-1/0/1,-1/0/2,-1/0/3,-1/1/0,-1/1/1,-1/1/2,-1/1/3,-1/2/0,-1/2/1,-1/2/2,-1/2/3,-1/3/2,-1/3/3,0/-4/1,0/-3/0,0/-3/1,0/-3/2,0/-3/3,0/-2/0,0/-2/1,0/-2/2,0/-2/3,0/-1/0,0/-1/1,0/-1/2,0/-1/3,0/0/0,0/0/1,0/0/2,0/0/3,0/1/0,0/1/1,0/1/2,0/1/3,0/2/0,0/2/1,0/2/2,0/2/3,0/3/3,1/-4/0,1/-4/1,1/-3/0,1/-3/1,1/-3/2,1/-3/3,1/-2/0,1/-2/1,1/-2/2,1/-2/3,1/-1/0,1/-1/1,1/-1/2,1/-1/3,1/0/0,1/0/1,1/0/2,1/0/3,1/1/0,1/1/1,1/1/2,1/1/3,1/2/0,1/2/1,1/2/2,1/2/3,2/-3/2,2/-2/1,2/-2/2,2/-1/0,2/-1/1,2/-1/2,2/-1/3,2/0/0,2/0/1,2/0/2,2/0/3,2/1/0,2/1/1,2/1/2,2/1/3,2/2/2,2/2/3,3/0/2,3/1/1,3/1/2,3/2/3}
        \fill[blue!30] (\x*\l,\y*\l) -- ++({90*\j-45}:{\l*0.70711}) -- ++({90*\j+90}:\l) -- cycle;
      \foreach \x/\y/\j in {-7/-4/0,-7/-4/1,-7/-3/0,-7/-3/1,-7/-3/3,-7/-2/0,-7/-2/3,-6/-4/0,-6/-4/1,-6/-4/2,-6/-3/0,-6/-3/1,-6/-3/2,-6/-3/3,-6/-2/0,-6/-2/1,-6/-2/2,-6/-2/3,-6/-1/0,-6/-1/1,-6/-1/2,-6/-1/3,-6/0/0,-6/0/1,-6/0/3,-6/1/0,-6/1/3,-5/-4/0,-5/-4/1,-5/-4/2,-5/-4/3,-5/-3/0,-5/-3/1,-5/-3/2,-5/-3/3,-5/-2/0,-5/-2/1,-5/-2/2,-5/-2/3,-5/-1/0,-5/-1/1,-5/-1/2,-5/-1/3,-5/0/0,-5/0/1,-5/0/2,-5/0/3,-5/1/0,-5/1/1,-5/1/2,-5/1/3,-5/2/0,-5/2/1,-5/2/2,-5/2/3,-5/3/0,-5/3/1,-5/3/3,-5/4/0,-5/4/3,-4/-5/0,-4/-5/1,-4/-4/0,-4/-4/1,-4/-4/2,-4/-4/3,-4/-3/0,-4/-3/1,-4/-3/2,-4/-3/3,-4/-1/1,-4/-1/2,-4/0/1,-4/0/2,-4/0/3,-4/1/0,-4/1/1,-4/1/2,-4/1/3,-4/2/0,-4/2/1,-4/2/2,-4/2/3,-4/3/0,-4/3/1,-4/3/2,-4/3/3,-4/4/0,-4/4/1,-4/4/2,-4/4/3,-4/5/0,-4/5/1,-4/5/2,-4/5/3,-4/6/0,-4/6/3,-3/-5/0,-3/-5/1,-3/-5/2,-3/-4/0,-3/-4/1,-3/-4/2,-3/-4/3,-3/-3/0,-3/-3/2,-3/-3/3,-3/2/1,-3/2/2,-3/3/1,-3/3/2,-3/3/3,-3/4/0,-3/4/1,-3/4/2,-3/4/3,-3/5/0,-3/5/1,-3/5/2,-3/5/3,-3/6/0,-3/6/2,-3/6/3,-2/-5/0,-2/-5/1,-2/-5/2,-2/-5/3,-2/-4/0,-2/-4/1,-2/-4/2,-2/-4/3,-2/-3/2,-2/-3/3,-2/4/0,-2/4/1,-2/4/2,-2/4/3,-2/5/0,-2/5/1,-2/5/2,-2/5/3,-2/6/2,-2/6/3,-1/-6/0,-1/-6/1,-1/-5/0,-1/-5/1,-1/-5/2,-1/-5/3,-1/-4/0,-1/-4/1,-1/-4/2,-1/-4/3,-1/3/0,-1/3/1,-1/4/0,-1/4/1,-1/4/2,-1/4/3,-1/5/0,-1/5/1,-1/5/2,-1/5/3,0/-6/0,0/-6/1,0/-6/2,0/-5/0,0/-5/1,0/-5/2,0/-5/3,0/-4/0,0/-4/2,0/-4/3,0/3/0,0/3/1,0/3/2,0/4/0,0/4/1,0/4/2,0/4/3,0/5/0,0/5/2,0/5/3,1/-6/0,1/-6/1,1/-6/2,1/-6/3,1/-5/0,1/-5/1,1/-5/2,1/-5/3,1/-4/2,1/-4/3,1/3/0,1/3/1,1/3/2,1/3/3,1/4/0,1/4/1,1/4/2,1/4/3,1/5/2,1/5/3,2/-7/0,2/-7/1,2/-6/0,2/-6/1,2/-6/2,2/-6/3,2/-5/0,2/-5/1,2/-5/2,2/-5/3,2/-4/0,2/-4/1,2/-4/2,2/-4/3,2/-3/0,2/-3/1,2/-3/3,2/-2/0,2/-2/3,2/2/0,2/2/1,2/3/0,2/3/1,2/3/2,2/3/3,2/4/0,2/4/1,2/4/2,2/4/3,3/-7/1,3/-7/2,3/-6/0,3/-6/1,3/-6/2,3/-6/3,3/-5/0,3/-5/1,3/-5/2,3/-5/3,3/-4/0,3/-4/1,3/-4/2,3/-4/3,3/-3/0,3/-3/1,3/-3/2,3/-3/3,3/-2/0,3/-2/1,3/-2/2,3/-2/3,3/-1/0,3/-1/1,3/-1/2,3/-1/3,3/0/0,3/0/1,3/0/3,3/1/0,3/1/3,3/2/0,3/2/1,3/2/2,3/3/0,3/3/1,3/3/2,3/3/3,3/4/0,3/4/1,3/4/2,3/4/3,4/-5/1,4/-5/2,4/-4/1,4/-4/2,4/-4/3,4/-3/0,4/-3/1,4/-3/2,4/-3/3,4/-2/0,4/-2/1,4/-2/2,4/-2/3,4/-1/0,4/-1/1,4/-1/2,4/-1/3,4/0/0,4/0/1,4/0/2,4/0/3,4/1/0,4/1/1,4/1/2,4/1/3,4/2/0,4/2/1,4/2/2,4/2/3,4/3/0,4/3/1,4/3/2,4/3/3,4/4/0,4/4/2,4/4/3,5/-2/1,5/-2/2,5/-1/1,5/-1/2,5/-1/3,5/0/0,5/0/1,5/0/2,5/0/3,5/1/0,5/1/1,5/1/2,5/1/3,5/2/0,5/2/1,5/2/2,5/2/3,5/3/0,5/3/1,5/3/2,5/3/3,5/4/2,5/4/3,6/2/1,6/2/2,6/3/2,6/3/3}
        \fill[red!25] (\x*\l,\y*\l) -- ++({90*\j-45}:{\l*0.70711}) -- ++({90*\j+90}:\l) -- cycle;
      \foreach \k in {-18,...,18} {
        \draw[gray!70,very thin] ({(\k+0.5)*\l},-12) -- ({(\k+0.5)*\l},12);
        \draw[gray!70,very thin] (-12,{(\k+0.5)*\l}) -- (12,{(\k+0.5)*\l});
        \draw[gray!70,very thin] ({\k*\l-12},-12) -- ({\k*\l+12},12);
        \draw[gray!70,very thin] ({\k*\l-12},12) -- ({\k*\l+12},-12);
      }
    \end{scope}
  \end{scope}
  \fill[pattern=north east lines,pattern color=black!45,even odd rule]
    [rounded corners=0.3cm] (-2.90,-2.90) rectangle (2.90,2.90)
    [sharp corners] (-2.30,-2.30) rectangle (2.30,2.30);
  \draw[densely dashed,rounded corners=0.3cm] (-2.90,-2.90) rectangle (2.90,2.90);
  \draw[densely dashed] (-2.30,-2.30) rectangle (2.30,2.30);
  \draw[densely dotted,semithick,rounded corners=0.50cm] (-3.10,-3.10) rectangle (3.10,3.10);
  \draw[densely dotted,semithick] (-2.10,-2.10) rectangle (2.10,2.10);
  \draw[thick] (-2.6,-2.6) rectangle (2.6,2.6);
  \draw[->] (4.3,3.0) node[right] {$C_\ell$} -- (2.6,2.2);
  \draw[->] (4.3,1.6) node[right] {$\partial C_\ell+B_{\delta'}$} -- (2.77,1.25);
  \draw[->] (4.3,0.0) node[right] {$J_0$} -- (1.69,-0.32);
  \draw[->] (4.3,-1.6) node[right] {$J\setminus J_0$} -- (3.39,-1.66);
  \draw[->] (4.3,-3.0) node[right] {$\partial C_\ell+B_{\delta'+\delta}$} -- (3.10,-2.4);
\end{tikzpicture}
\caption{Two-dimensional sketch of the proof of Lemma~\ref{lem:upper_bounded_layer}.}
\label{fig:bounded_layer}
\end{figure}
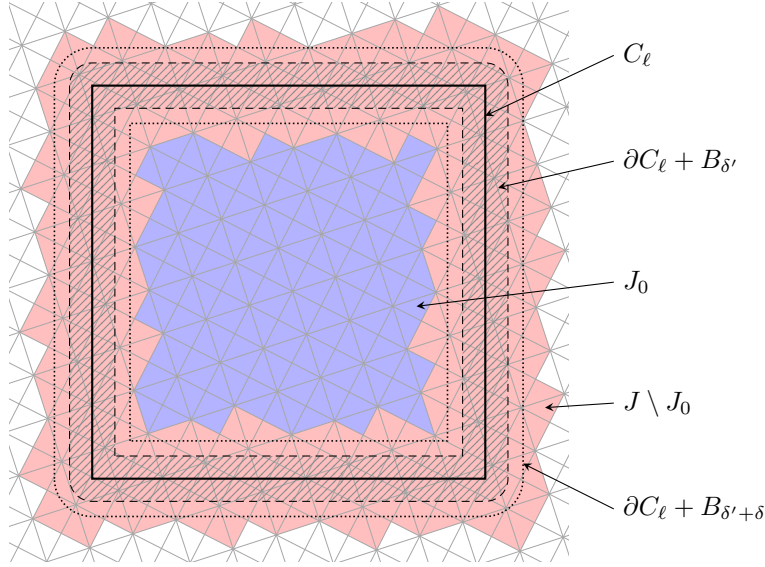

The corresponding tetrahedra are disjoint subsets of $C_\ell$, so
$|J_0|\,|h\tetra|\le \ell^3$.
Conversely, if a point $x\in C_\ell$ is at a distance more than $h'$ from $\partial C_\ell$, then the index of the tetrahedron containing $x$ belongs to $J_0$. Indeed, every point of that tetrahedron lies within $d_0h$ of $x$. The union of the tetrahedra with indices in $J_0$ consequently contains $C_{\ell-2h'}$ up to null sets. Thus
\begin{equation}\label{eq:upper_interior_count}
0\le \ell^3-|J_0|\,|h\tetra|
\le \ell^3-(\ell-2h')^3
\le C \ell^2h'.
\end{equation}
For $(z,j)\in J_0$, the $\delta$-enlargement of the tetrahedron lies in $C_\ell\setminus(\partial C_\ell+B_{\delta'})$, where $\rho_\ell=\rho_0$. Hence the localized density equals $\rho_0\,\iden_{Q^{-1}(h\mu_j\tetra+h z+a)}*\eta_\delta$. Translation invariance, orthogonal invariance, and radiality of the smoothing function show that its energy is the same finite number $E(\rho_0\,\iden_{h\tetra}*\eta_\delta)$.

For $(z,j)\in J\setminus J_0$, the localized density $(\iden_{Q^{-1}(h\mu_j\tetra+h z+a)}*\eta_\delta)\rho_\ell$ vanishes outside $\partial C_\ell+B_{h'}$. Indeed, a point outside $C_\ell$ at which $\rho_\ell$ is nonzero is within distance $\delta'$ of its boundary. If instead $x\in C_\ell$ has distance at least $h'$ from the boundary and $(\iden_{Q^{-1}(h\mu_j\tetra+h z+a)}*\eta_\delta)(x)>0$, then $x$ lies in the $\delta$-enlargement of the tetrahedron. Every point of the tetrahedron is then within distance $d_0h+\delta$ of $x$, which would give $(z,j)\in J_0$. This proves the assertion.

The Lieb--Oxford inequality (Theorem~\ref{thm:basic_lieb_oxford}) and nonnegativity of the kinetic energy give $E(\rho)\ge-c_{\mathrm{LO}}\int_{\RR^3}\rho^{4/3}$ for each of the localized densities $\rho$. For the sum over $J\setminus J_0$ we use that each smeared indicator is at most $1$, so that its $4/3$ power is at most the smeared indicator itself, together with the partition identity~\eqref{eq:basic_partition}. This avoids an estimate that grows with the overlap multiplicity or the number of boundary tetrahedra. We obtain
\begin{align*}
&\sum_{(z,j)\in J\setminus J_0}E\big((\iden_{Q^{-1}(h\mu_j\tetra+h z+a)}*\eta_\delta)\rho_\ell\big)\\
&\qquad\ge-c_{\mathrm{LO}}\int_{\partial C_\ell+B_{h'}}\rho_\ell^{4/3}
\sum_{(z,j)\in J\setminus J_0}(\iden_{Q^{-1}(h\mu_j\tetra+h z+a)}*\eta_\delta)^{4/3}\\
&\qquad\ge-c_{\mathrm{LO}}(\Lambda\rho_0)^{4/3}|\partial C_\ell+B_{h'}|
\ge-C(\Lambda\rho_0)^{4/3}\ell^2h'.
\end{align*}
For the interior contribution,~\eqref{eq:upper_interior_count} implies the sign-independent estimate
\[
|J_0|\,E(\rho_0\,\iden_{h\tetra}*\eta_\delta)
\ge\frac{\ell^3}{|h\tetra|}E(\rho_0\,\iden_{h\tetra}*\eta_\delta)
-C\frac{|E(\rho_0\,\iden_{h\tetra}*\eta_\delta)|}{|h\tetra|}\ell^2h'.
\]
Since the terms with $(z,j)\notin J$ vanish, combining the last two estimates, uniformly in $(Q,a)$, gives
\begin{multline}\label{eq:upper_local_sum_bound}
\sum_{z\in\ZZ^3}\sum_{j=1}^{24}E\big((\iden_{Q^{-1}(h\mu_j\tetra+h z+a)}*\eta_\delta)\rho_\ell\big)
\\
\ge \ell^3\frac{E(\rho_0\,\iden_{h\tetra}*\eta_\delta)}{|h\tetra|}
-C \ell^2h'\left(\frac{|E(\rho_0\,\iden_{h\tetra}*\eta_\delta)|}{|h\tetra|}+(\Lambda\rho_0)^{4/3}\right).
\end{multline}
This form is valid whether $E(\rho_0\,\iden_{h\tetra}*\eta_\delta)$ is positive, negative, or zero.

Insert~\eqref{eq:upper_local_sum_bound} and~\eqref{eq:upper_layer_integrals} into~\eqref{eq:upper_layer_localization}, noting that the normalized average of a uniform bound is the same bound and that $0<1-C\delta/h\le1$. After dividing by $\ell^3$, we find
\begin{align*}
\frac{E(\rho_\ell)}{\ell^3}
&\ge\left(1-\frac{C\delta}h\right)\frac{E(\rho_0\,\iden_{h\tetra}*\eta_\delta)}{|h\tetra|}\\
&\quad-C\frac{h'}{\ell}
\left(\frac{|E(\rho_0\,\iden_{h\tetra}*\eta_\delta)|}{|h\tetra|}+(\Lambda\rho_0)^{4/3}\right)
-\frac{C(\Lambda\rho_0+\delta^3\Lambda^2\rho_0^2)}h
-\frac{C\Lambda\rho_0}{h\delta}.
\end{align*}
At fixed $h$, the energy $E(\rho_0\,\iden_{h\tetra}*\eta_\delta)$ is finite and $h'/\ell\to0$. Therefore
\[
\liminf_{\ell\to\infty}\frac{E(\rho_\ell)}{\ell^3}
\ge\left(1-\frac{C\delta}h\right)\frac{E(\rho_0\,\iden_{h\tetra}*\eta_\delta)}{|h\tetra|}
-\frac{C_{\Lambda,\rho_0}}{h\delta}.
\]
 The hypotheses of Theorem~\ref{thm:UEG}(ii) hold as $h\to\infty$, since $\delta/h\to0$ and $\delta h\to\infty$. Thus
\[
\lim_{h\to\infty}\frac{E(\rho_0\,\iden_{h\tetra}*\eta_\delta)}{|h\tetra|}=e_\UEG(\rho_0).
\]
Letting $h\to\infty$ at fixed $\delta$ in the previous lower bound proves~\eqref{eq:upper_bounded_layer}.
\end{proof}

\begin{lemma}\label{lem:upper_collar}
Fix $\delta>0$ and $r\ge2\delta$. There exists a bounded radial function $g_{r,\delta}$ supported in $\overline{B_{r+2\delta}\setminus B_r}$ such that
\begin{equation}\label{eq:upper_collar_charge}
\int_{\RR^3}g_{r,\delta}=0,
\qquad
0\le\rho_0+g_{r,\delta}\le C\rho_0
\quad\hbox{on }B_{r+2\delta}\setminus B_r,
\end{equation}
where $C$ is an absolute constant. Define
\[
\rho_{r,\delta}(x)=
\begin{cases}
0,&|x|\le r,\\
\rho_0+g_{r,\delta}(x),&r<|x|<r+2\delta,\\
\rho_0,&|x|\ge r+2\delta.
\end{cases}
\]
Then $\sqrt{\rho_{r,\delta}}$ is smooth and
\begin{equation}\label{eq:upper_collar_costs}
\int_{B_{r+2\delta}\setminus B_r}|\nabla\sqrt{\rho_{r,\delta}}|^2
\le C\rho_0r^2/\delta,
\qquad
D(g_{r,\delta})\le C\rho_0^2r^2\delta^3.
\end{equation}
The potential $\int_{\RR^3}g_{r,\delta}(y)|x-y|^{-1}\,\mathrm dy$ vanishes for $|x|\ge r+2\delta$.
\end{lemma}

\begin{proof}
Choose fixed functions $\sigma,\vartheta\in C^\infty(\RR)$ such that
\[
0\le \sigma\le1,
\qquad \sigma(t)=0\ \hbox{for }t\le\tfrac14,
\qquad \sigma(t)=1\ \hbox{for }t\ge\tfrac34,
\]
and
\[
\vartheta\ge0,
\qquad \supp\vartheta\subset(1,2),
\qquad I_\vartheta=\int_1^2\vartheta(t)\,\mathrm dt>0.
\]
For the ratio $\delta/r\le1/2$, set
\begin{equation}\label{eq:upper_collar_coefficient}
 h(t)=\sigma(t)^2+c\vartheta(t),
 \qquad
 c=
 \frac{\displaystyle\int_0^2(1-\sigma(t)^2)(r+\delta t)^2\,\mathrm dt}
 {\displaystyle\int_1^2\vartheta(t)(r+\delta t)^2\,\mathrm dt}.
\end{equation}
The numerator is at least $r^2/4$, because $\sigma=0$ on $[0,1/4]$. It is at most $(r+\delta)^2\le9r^2/4$, because $1-\sigma^2$ vanishes for $t\ge3/4$. The denominator lies between $r^2I_\vartheta$ and $(r+2\delta)^2I_\vartheta\le4r^2I_\vartheta$. Consequently
\begin{equation}\label{eq:upper_collar_uniform_coefficient}
\frac1{16I_\vartheta}\le c\le\frac9{4I_\vartheta}.
\end{equation}
In particular, $0\le h\le C$ for an absolute constant $C$, and
\[
\sqrt{h(t)}=
\begin{cases}
\sigma(t),&0\le t\le1,\\
\sqrt{1+c\vartheta(t)},&1\le t\le2.
\end{cases}
\]
Both expressions are smooth, the second since $1+c\vartheta\ge1$. As $\supp\vartheta$ is a compact subset of $(1,2)$, $\vartheta$ vanishes near $t=1$ and near $t=2$. Since also $\sigma=1$ on $[3/4,1]$, both expressions equal one near $t=1$, so $\sqrt h$ is smooth on $[0,2]$. It vanishes for $t\le1/4$ and equals one near $t=2$. Moreover,
\begin{equation}\label{eq:upper_collar_profile_gradient}
\int_0^2\big|(\sqrt h)'(t)\big|^2\,\mathrm dt
\le\int_0^1|\sigma'(t)|^2\,\mathrm dt
+\frac14\left(\frac9{4I_\vartheta}\right)^2
\int_1^2|\vartheta'(t)|^2\,\mathrm dt
\le C.
\end{equation}

Define
\[
g_{r,\delta}(x)=
\begin{cases}
\rho_0\left(h\left(\dfrac{|x|-r}{\delta}\right)-1\right),&r<|x|<r+2\delta,\\
0,&\hbox{otherwise}.
\end{cases}
\]
The choice of $c$ implies
\[
\int_{\RR^3}g_{r,\delta}
=4\pi\rho_0\delta\int_0^2(h(t)-1)(r+\delta t)^2\,\mathrm dt=0.
\]
The bound on the density follows from $0\le h\le C$. We have, by~\eqref{eq:upper_collar_profile_gradient},
\[
\int_{B_{r+2\delta}\setminus B_r}|\nabla\sqrt{\rho_{r,\delta}}|^2
=\frac{4\pi\rho_0}{\delta}
\int_0^2\big|(\sqrt h)'(t)\big|^2(r+\delta t)^2\,\mathrm dt
\le C\rho_0r^2/\delta.
\]

To estimate the direct energy of $g_{r,\delta}$, let $Q(t)$ be the charge inside the sphere of radius $t$,
\[
Q(t)=4\pi\int_0^t s^2g_{r,\delta}(s)\,\mathrm ds.
\]
Its support as a function of $t$ is contained in $[r,r+2\delta]$, and
\[
|Q(t)|\le\int_{\RR^3}|g_{r,\delta}|
\le C\rho_0r^2\delta.
\]
Since the total charge is zero, the potential $\int_{\RR^3}g_{r,\delta}(y)|x-y|^{-1}\,\mathrm dy$ vanishes for $|x|\ge r+2\delta$ by Theorem~\ref{thm:basic_newton}(i). Theorem~\ref{thm:basic_newton}(iv) gives
\[
D(g_{r,\delta})
=\frac12\int_0^\infty\frac{Q(t)^2}{t^2}\,\mathrm dt
=\frac12\int_r^{r+2\delta}\frac{Q(t)^2}{t^2}\,\mathrm dt
\le\delta\frac{\sup_tQ(t)^2}{r^2}
\le C\rho_0^2r^2\delta^3.
\]
This proves the lemma.
\end{proof}

\begin{lemma}[A finite packing with room for shells]\label{lem:upper_packing}
Given $\varepsilon\in(0,1)$, $\delta>0$ and $r_*>0$, there exist $L>0$ and finitely many balls $B_{r_m}(c_m)$, $1\le m\le M$, such that
\[
r_m\ge r_*\,,
\qquad N_m=\rho_0|B_{r_m}|\in\NN,
\qquad
\frac1{L^3}\sum_{m=1}^M|B_{r_m}|\ge1-\varepsilon,
\]
and the enlarged closed balls $\overline{B_{r_m+2\delta}(c_m)}$ are pairwise disjoint and contained in $C_L$.
\end{lemma}

\begin{proof}
Let $\BC$ be the family of all closed balls $\overline{B_s(z)}$ with $s>0$ that are contained in the open unit cube $C_1$. It is a Vitali cover of $C_1$. Indeed, if $x\in C_1$ and $0<s<d(x,\partial C_1)$, then $\overline{B_s(x)}\in\BC$, so every point of $C_1$ lies in balls of $\BC$ of arbitrarily small radius. Since $\BC$ consists of nondegenerate closed balls of diameter at most $1$, the Vitali covering theorem for Lebesgue measure~\cite[Sec.~1.5]{evans2015measure} applies. We obtain countably many pairwise disjoint balls $\overline{B_{s_k}(z_k)}\in\BC$, $k\in\NN$, such that $C_1\setminus\bigcup_k\overline{B_{s_k}(z_k)}$ is a null set. Since these balls are disjoint subsets of $C_1$,
\[
\sum_{k=1}^\infty|B_{s_k}|=|C_1|=1.
\]
Choose $M\ge2$ such that $\sum_{k=1}^M|B_{s_k}|>1-\varepsilon/2$. The balls $\overline{B_{s_m}(z_m)}$, $1\le m\le M$, are finitely many pairwise disjoint compact subsets of the open cube $C_1$. Hence
\[
d=\min\Big\{\min_{1\le m<k\le M}\big(|z_m-z_k|-s_m-s_k\big),\ \min_{1\le m\le M}\big(d(z_m,\partial C_1)-s_m\big)\Big\}>0.
\]

Next we dilate the cube $C_1$ and the balls $\overline{B_{s_m}(z_m)}$, $1\le m\le M$, by a large factor $L>0$. This gives the cube $C_L=LC_1$ and the balls $\overline{B_{Ls_m}(Lz_m)}$.
 We then round the radii $Ls_m$ down to neutral ones. For $1\le m\le M$ put
\[
N_m=\big\lfloor\rho_0|B_{Ls_m}|\big\rfloor,
\qquad
c_m=Lz_m,
\]
and let $r_m\ge0$ be the radius with $\rho_0|B_{r_m}|=N_m$. Then $r_m\le Ls_m$. Since $N_m>\rho_0|B_{Ls_m}|-1$, we have
\[
|B_{r_m}|>|B_{Ls_m}|-\frac1{\rho_0}.
\]
In particular, $N_m\ge1$ and $r_m\ge r_*$ for all sufficiently large $L$. Summing over $m$ gives
\[
\frac1{L^3}\sum_{m=1}^M|B_{r_m}|
>\sum_{m=1}^M|B_{s_m}|-\frac{M}{\rho_0L^3}
>1-\frac\varepsilon2-\frac{M}{\rho_0L^3},
\]
which is at least $1-\varepsilon$ when $L^3\ge2M/(\rho_0\varepsilon)$. Finally, for $m\ne k$, the definition of $d$ and $r_m\le Ls_m$ give
\[
|c_m-c_k|=L|z_m-z_k|\ge L(s_m+s_k+d)\ge r_m+r_k+Ld,
\]
so $\overline{B_{r_m+2\delta}(c_m)}$ and $\overline{B_{r_k+2\delta}(c_k)}$ are disjoint when $Ld>4\delta$. Similarly,
\[
d(c_m,\partial C_L)=L\,d(z_m,\partial C_1)\ge L(s_m+d)\ge r_m+Ld,
\]
so $\overline{B_{r_m+2\delta}(c_m)}\subset C_L$ when $Ld>2\delta$. See Figure~\ref{fig:packing}. Choosing $L$ large enough for all these conditions completes the proof.
\end{proof}

\begin{figure}[ht]
\centering
\begin{tikzpicture}[>=stealth]
  \fill[gray!20] (-3.0,-3.0) rectangle (3.0,3.0);
  \foreach \x/\y/\r in {-1.930/-1.570/0.800,1.890/0.250/0.800,-1.890/0.850/0.800,0.990/2.050/0.800,0.610/-1.030/0.800,-0.050/0.650/0.800,2.110/-1.990/0.800,-0.790/2.170/0.748,-0.450/-2.250/0.668,-2.290/2.330/0.588,-0.810/-0.610/0.529,2.390/1.630/0.520,0.790/-2.470/0.448,-2.510/-0.350/0.408,2.150/2.590/0.328,2.590/-0.810/0.328,-2.610/-2.610/0.301,-1.690/-0.330/0.252,-1.430/-2.650/0.248,1.750/-0.890/0.207,-1.990/-2.730/0.176,0.110/2.730/0.170,-2.750/1.550/0.167,2.750/2.390/0.163,1.030/0.950/0.159,2.770/0.910/0.148,-1.030/0.190/0.141,0.830/0.030/0.134,-0.870/-1.410/0.129,-2.790/0.270/0.128,-0.430/-1.310/0.125,-2.790/-0.950/0.112,1.410/-2.810/0.108,2.810/-0.270/0.108,2.750/2.810/0.108,-1.570/2.810/0.108,2.810/-2.790/0.108,1.710/1.290/0.099,-2.810/-2.110/0.090,-0.930/1.210/0.080,1.650/2.830/0.080,-1.310/-0.050/0.078,0.230/-2.830/0.076,0.050/1.670/0.072,-2.850/1.170/0.068,2.850/-1.290/0.068,2.850/0.550/0.064,1.370/1.110/0.063} {
    \fill[red!25] (\x,\y) circle ({\r+0.06});
    \fill[blue!30] (\x,\y) circle (\r);
    \draw[densely dashed] (\x,\y) circle ({\r+0.06});
    \draw (\x,\y) circle (\r);
  }
  \draw[thick] (-3.0,-3.0) rectangle (3.0,3.0);
  \draw[->] (2.110,-1.990) -- ++(150:0.800) node[pos=0.5,above] {$r_m$};
  \fill (2.110,-1.990) circle (0.03) node[below=2pt] {$c_m$};
  \draw[->] (4.0,2.6) node[right] {$C_L$} -- (3.0,2.0);
  \draw[->] (4.0,-0.790) node[right] {$B_{r_m}(c_m)$} -- (2.447,-1.707);
  \draw[->] (4.0,-2.190) node[right] {$\overline{B_{r_m+2\delta}(c_m)}\setminus B_{r_m}(c_m)$} -- (2.927,-2.134);
\end{tikzpicture}
\caption{Two-dimensional sketch of the packing of Lemma~\ref{lem:upper_packing}, not to scale. The balls $B_{r_m}(c_m)$ (blue) and their shells of width $2\delta$ (red, bounded by the dashed circles) lie in $C_L$, and the enlarged balls are pairwise disjoint.}
\label{fig:packing}
\end{figure}
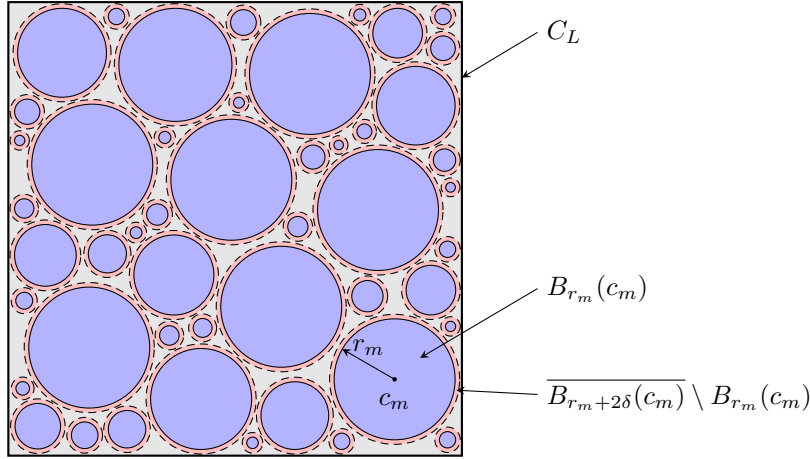

We are now in a position to complete the proof of the main result of this section.

\begin{proof}[Proof of Theorem~\ref{thm:upper}]
We shall apply Lemma~\ref{lem:upper_bounded_layer}. Then, it remains to bound the resulting $\liminf$ from above by $e_\Jel(\rho_0)$, up to an error that can be made arbitrarily small.

Fix a shell width $\delta>0$ and a number $\varepsilon\in(0,1)$. Constants depending on $\rho_0,\delta$ will not depend on the finite packing constructed below.

By the thermodynamic limit on neutral balls (Theorem~\ref{thm:LN_main}(ii) and Remark~\ref{rem:LN_balls}), there is a radius $r_*>2\delta$ such that every $r\ge r_*$ with $N=\rho_0|B_r|\in\NN$ satisfies
\begin{equation}\label{eq:upper_core_energy}
E_\Jel(\rho_0,B_r)
\le(e_\Jel(\rho_0)+\varepsilon)|B_r|.
\end{equation}
We may increase $r_*$ further whenever needed.

For each such $r$, let $\Gamma_r$ be a normalized rotation-invariant ground-state density matrix of $H^{\rho_0}_{B_r}$. A ground state exists because the Dirichlet Hamiltonian on the ball has compact resolvent. Averaging a ground state over spatial rotations produces such a $\Gamma_r$, since the Hamiltonian commutes with rotations. This state may be mixed, and it satisfies
\[
\mathcal J_{\rho_0\iden_{B_r}}(\Gamma_r)
=E_\Jel(\rho_0,B_r).
\]
The density $\rho_{\Gamma_r}$ is radial, supported in $B_r$, and has integral $N=\rho_0|B_r|$. Hence $\rho_{\Gamma_r}-\rho_0\iden_{B_r}$ is a radial neutral charge distribution, and Newton's theorem~\eqref{eq:basic_newton} gives
\begin{equation}\label{eq:upper_core_zero_potential}
\int_{\RR^3}\frac{\rho_{\Gamma_r}(y)-\rho_0\iden_{B_r}(y)}{|x-y|}\,\mathrm dy=0\quad\hbox{for }x\in\RR^3\setminus B_r.
\end{equation}
The density $\rho_{\Gamma_r}$ has finite direct energy. Indeed, $\Gamma_r$ has finite kinetic energy, so $\sqrt{\rho_{\Gamma_r}}\in H^1$ by the Hoffmann--Ostenhof inequality~\cite{hoffmann1977schrodinger}. Hence $\rho_{\Gamma_r}\in L^1\cap L^3$, and its direct energy is finite.

Next, we apply Lemma~\ref{lem:upper_packing} with $\varepsilon$, $\delta$ and the chosen $r_*$. It gives a side length $L>0$ and finitely many balls $B_{r_m}(c_m)$, $1\le m\le M$, with $r_m\ge r_*$ and $\rho_0|B_{r_m}|\in\NN$, whose closed enlargements $\overline{B_{r_m+2\delta}(c_m)}$ are pairwise disjoint and contained in $C_L$, and which cover the fraction $\theta$ of $C_L$, where
\[
\theta=\frac1{L^3}\sum_{m=1}^M|B_{r_m}|\ge1-\varepsilon.
\]
Since each $r_m$ is admissible in~\eqref{eq:upper_core_energy}, we may use the rotation-invariant ground-state density matrix $\Gamma_{r_m}$ centered at the origin for $m=1,\dots,M$.
For a fixed shift $a\in C_L$ and an index $(z,m)\in\ZZ^3\times\{1,\dots,M\}$, the translated ball is $B_{r_m}(c_m+Lz+a)$ and the translated enlarged ball is $B_{r_m+2\delta}(c_m+Lz+a)$. Let $\Gamma_{z,m,a}$ be the translate of $\Gamma_{r_m}$ by $c_m+Lz+a$. The closed enlarged balls are disjoint, including between different periodic cells.

Choose a radial mollifier $\chi$, that is, a radial function with $\chi\ge0$, $\int_{\RR^3}\chi=1$, $\supp\chi\subset B_1$ and $\sqrt\chi\in C_c^\infty(\RR^3)$, and put $\chi_\delta(x)=\delta^{-3}\chi(x/\delta)$. We use the fixed outer background $\rho_0\iden_{C_\ell}*\chi_\delta$, which does not depend on $a$. It satisfies $0\le\rho_0\iden_{C_\ell}*\chi_\delta\le\rho_0$, $\int_{\RR^3}\rho_0\iden_{C_\ell}*\chi_\delta=\rho_0\ell^3$, and $\rho_0\iden_{C_\ell}*\chi_\delta=\rho_0$ on the cube $C_{\ell-2\delta}$ of points of $C_\ell$ at distance more than $\delta$ from $\partial C_\ell$. Its support lies in $C_\ell+B_\delta$.

Hence, we retain only the balls indexed by
\begin{equation}\label{eq:upper_selected_indices}
J_0(a,\ell)=\big\{(z,m)\in\ZZ^3\times\{1,\dots,M\}:\ \overline{B_{r_m+2\delta}(c_m+Lz+a)}\subset C_{\ell-2\delta}\big\}.
\end{equation}
See Figure~\ref{fig:upper_construction} (top).
Define the total volume and the total surface area of the selected balls by
\[
V(a,\ell)=\sum_{(z,m)\in J_0(a,\ell)}|B_{r_m}|,
\qquad
A(a,\ell)=\sum_{(z,m)\in J_0(a,\ell)}|\partial B_{r_m}|.
\]
\begin{claim}
Uniformly for $a\in C_L$,
\begin{equation}\label{eq:upper_explicit_counting}
V(a,\ell)=\theta \ell^3+\mathcal O_{L,\delta}(\ell^2),
\qquad
A(a,\ell)\le\frac3{r_*}V(a,\ell)\le\frac3{r_*}\ell^3.
\end{equation}
\end{claim}

\begin{proof}[Proof of the claim]
For the first assertion, consider the cells $a+Lz+C_L$, $z\in\ZZ^3$, which tile $\RR^3$ up to null sets. The cell with index $z$ contains the balls $B_{r_m}(c_m+Lz+a)$, $1\le m\le M$, of total volume $\theta L^3$. Let
\[
Z_0=\{z\in\ZZ^3:\ a+Lz+\overline{C_L}\subset C_{\ell-2\delta}\},
\qquad
Z=\{z\in\ZZ^3:\ (a+Lz+C_L)\cap C_\ell\ne\emptyset\}.
\]
If $z\in Z_0$, then all enlarged balls of that cell lie in $C_{\ell-2\delta}$, so $(z,m)\in J_0(a,\ell)$ for every $m$. If $(z,m)\in J_0(a,\ell)$, then $B_{r_m}(c_m+Lz+a)\subset C_\ell$, so $z\in Z$. Hence
\[
\theta L^3|Z_0|\le V(a,\ell)\le\theta L^3|Z|.
\]
Two points of the closure of a cell differ by at most $L$ in each coordinate. Therefore every cell meeting $C_{\ell-2\delta-2L}$ has its index in $Z_0$, and every cell with index in $Z$ is contained in $C_{\ell+2L}$. Comparing volumes gives
\[
(\ell-2\delta-2L)^3\le L^3|Z_0|,
\qquad
L^3|Z|\le(\ell+2L)^3.
\]
Let $\ell\ge2(\delta+L)$. Using $(\ell-t)^3\ge \ell^3-3t\ell^2$ for $0\le t\le \ell$, $(\ell+2L)^3\le \ell^3+26L\ell^2$ for $\ell\ge L$, and $\theta\le1$, we obtain
\[
|V(a,\ell)-\theta \ell^3|\le26(\delta+L)\ell^2,
\]
uniformly in $a\in C_L$. This proves the first assertion. Note that $\ell$ need not be a multiple of $L$, and the cells need not fit exactly into $C_\ell$. For the second, $|\partial B_r|=3|B_r|/r$ and $r_m\ge r_*$ give
\[
A(a,\ell)=\sum_{(z,m)\in J_0(a,\ell)}\frac{3|B_{r_m}|}{r_m}
\le\frac3{r_*}V(a,\ell),
\]
and $V(a,\ell)\le \ell^3$ because the selected balls are pairwise disjoint subsets of $C_\ell$.
\end{proof}

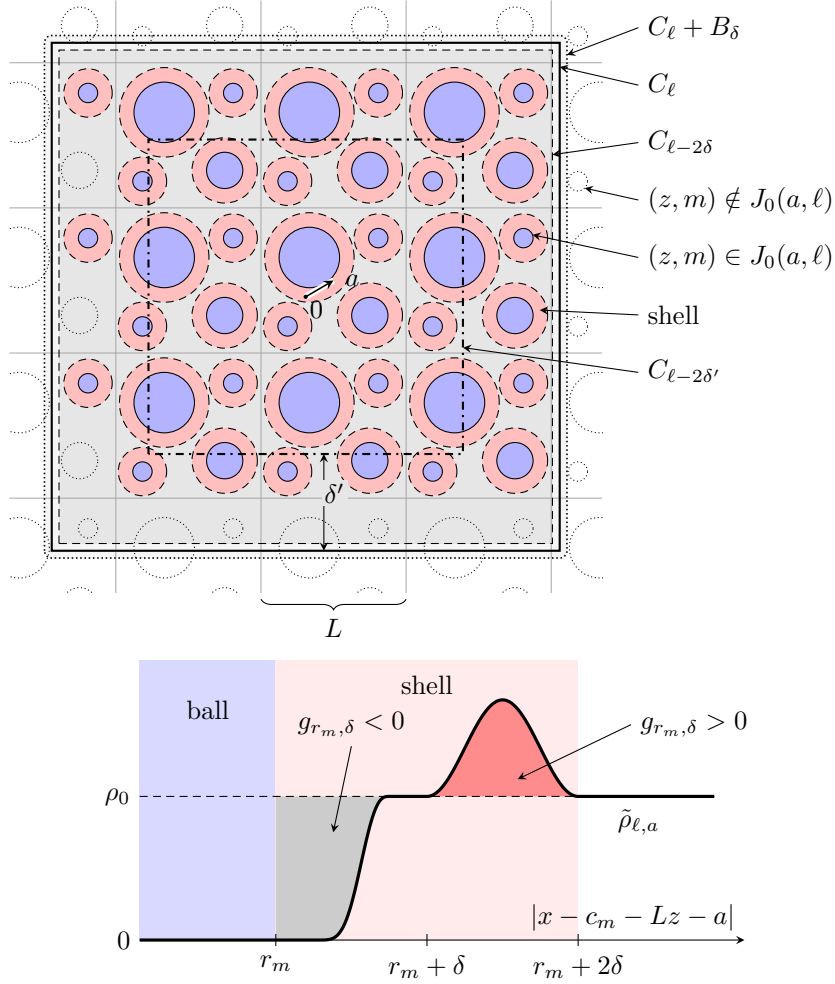
\begin{figure}[htp]
\centering
\begin{tikzpicture}[>=stealth,scale=0.8]
  \begin{scope}
    \clip (-4.9,-4.9) rectangle (4.9,4.9);
    \fill[gray!8,rounded corners=0.096cm] (-4.32,-4.32) rectangle (4.32,4.32);
    \fill[gray!20] (-4.08,-4.08) rectangle (4.08,4.08);
    \foreach \t in {-3.14,-0.74,1.66,4.06} \draw[gray!70,very thin] (\t,-4.9) -- (\t,4.9);
    \foreach \t in {-3.33,-0.93,1.47,3.87} \draw[gray!70,very thin] (-4.9,\t) -- (4.9,\t);
    \foreach \x/\y/\r in {-4.740/-4.150/0.500,-3.740/-5.110/0.300,-3.600/-3.830/0.160,-4.740/-1.750/0.500,-3.740/-2.710/0.300,-4.740/0.650/0.500,-3.740/-0.310/0.300,-4.740/3.050/0.500,-3.740/2.090/0.300,-3.740/4.490/0.300,-2.340/-4.150/0.500,-1.340/-5.110/0.300,-1.200/-3.830/0.160,-1.340/4.490/0.300,-2.700/4.310/0.160,0.060/-4.150/0.500,1.060/-5.110/0.300,1.200/-3.830/0.160,1.060/4.490/0.300,-0.300/4.310/0.160,2.460/-4.150/0.500,3.460/-5.110/0.300,3.600/-3.830/0.160,3.460/4.490/0.300,2.100/4.310/0.160,4.860/-4.150/0.500,4.860/-1.750/0.500,4.500/-2.890/0.160,4.860/0.650/0.500,4.500/-0.490/0.160,4.860/3.050/0.500,4.500/1.910/0.160,4.500/4.310/0.160}
      \draw[densely dotted] (\x,\y) circle (\r);
    \foreach \x/\y/\r in {-3.600/-1.430/0.160,-3.600/0.970/0.160,-3.600/3.370/0.160,-2.340/-1.750/0.500,-1.340/-2.710/0.300,-1.200/-1.430/0.160,-2.700/-2.890/0.160,-2.340/0.650/0.500,-1.340/-0.310/0.300,-1.200/0.970/0.160,-2.700/-0.490/0.160,-2.340/3.050/0.500,-1.340/2.090/0.300,-1.200/3.370/0.160,-2.700/1.910/0.160,0.060/-1.750/0.500,1.060/-2.710/0.300,1.200/-1.430/0.160,-0.300/-2.890/0.160,0.060/0.650/0.500,1.060/-0.310/0.300,1.200/0.970/0.160,-0.300/-0.490/0.160,0.060/3.050/0.500,1.060/2.090/0.300,1.200/3.370/0.160,-0.300/1.910/0.160,2.460/-1.750/0.500,3.460/-2.710/0.300,3.600/-1.430/0.160,2.100/-2.890/0.160,2.460/0.650/0.500,3.460/-0.310/0.300,3.600/0.970/0.160,2.100/-0.490/0.160,2.460/3.050/0.500,3.460/2.090/0.300,3.600/3.370/0.160,2.100/1.910/0.160} {
      \fill[red!25] (\x,\y) circle ({\r+0.24});
      \fill[blue!30] (\x,\y) circle (\r);
      \draw[densely dashed] (\x,\y) circle ({\r+0.24});
      \draw (\x,\y) circle (\r);
    }
  \end{scope}
  \draw[densely dotted,semithick,rounded corners=0.096cm] (-4.32,-4.32) rectangle (4.32,4.32);
  \draw[thick] (-4.2,-4.2) rectangle (4.2,4.2);
  \draw[densely dashed] (-4.08,-4.08) rectangle (4.08,4.08);
  \draw[thick,dash dot] (-2.6,-2.6) rectangle (2.6,2.6);
  \draw[white,line width=2.4pt] (0,0) -- (0.46,0.27);
  \draw[->,semithick] (0,0) -- (0.46,0.27);
  \fill (0,0) circle (0.035);
  \node[below right=-2pt] at (0,0) {$0$};
  \node[right=1pt] at (0.46,0.27) {$a$};
  \draw[<->] (0.3,-2.6) -- (0.3,-4.2);
  \node[right=-1pt,fill=gray!20,inner sep=1pt] at (0.3,-3.25) {$\delta'$};
  \draw[decorate,decoration={brace,mirror,amplitude=4pt}] (-0.74,-5.0) -- (1.66,-5.0) node[midway,below=4pt] {$L$};
  \draw[->] (5.5,4.45) node[right] {$C_\ell+B_\delta$} -- (4.32,4.0);
  \draw[->] (5.5,3.5) node[right] {$C_\ell$} -- (4.2,3.8);
  \draw[->] (5.5,2.55) node[right] {$C_{\ell-2\delta}$} -- (4.08,2.35);
  \draw[->] (5.5,1.6) node[right] {$(z,m)\notin J_0(a,\ell)$} -- (4.613,1.797);
  \draw[->] (5.5,0.65) node[right] {$(z,m)\in J_0(a,\ell)$} -- (3.75,1.025);
  \draw[->] (5.5,-0.3) node[right] {shell} -- (3.87,-0.22);
  \draw[->] (5.5,-1.25) node[right] {$C_{\ell-2\delta'}$} -- (2.6,-0.84);
\end{tikzpicture}
\par\medskip
\begin{tikzpicture}[>=stealth]
  \fill[blue!15] (-1.800,0) rectangle (0,3.705);
  \fill[red!8] (0,0) rectangle (4.000,3.705);
  \fill[gray!40] (0.000,1.900) -- (0.000,0.000) -- (0.005,0.000) -- (0.024,0.000) -- (0.043,0.000) -- (0.062,0.000) -- (0.081,0.000) -- (0.100,0.000) -- (0.119,0.000) -- (0.138,0.000) -- (0.157,0.000) -- (0.176,0.000) -- (0.195,0.000) -- (0.214,0.000) -- (0.233,0.000) -- (0.252,0.000) -- (0.271,0.000) -- (0.290,0.000) -- (0.309,0.000) -- (0.328,0.000) -- (0.347,0.000) -- (0.366,0.000) -- (0.385,0.000) -- (0.404,0.000) -- (0.423,0.000) -- (0.442,0.000) -- (0.461,0.000) -- (0.480,0.000) -- (0.499,0.000) -- (0.500,0.000) -- (0.518,0.000) -- (0.537,0.000) -- (0.556,0.000) -- (0.575,0.000) -- (0.594,0.000) -- (0.613,0.000) -- (0.632,0.001) -- (0.651,0.001) -- (0.670,0.003) -- (0.689,0.005) -- (0.708,0.008) -- (0.727,0.012) -- (0.746,0.019) -- (0.765,0.027) -- (0.784,0.039) -- (0.803,0.053) -- (0.822,0.071) -- (0.841,0.093) -- (0.860,0.120) -- (0.879,0.151) -- (0.898,0.187) -- (0.917,0.229) -- (0.936,0.276) -- (0.955,0.329) -- (0.974,0.387) -- (0.993,0.450) -- (1.012,0.519) -- (1.031,0.592) -- (1.050,0.668) -- (1.069,0.749) -- (1.088,0.832) -- (1.107,0.917) -- (1.126,1.003) -- (1.145,1.089) -- (1.164,1.174) -- (1.183,1.258) -- (1.202,1.339) -- (1.221,1.417) -- (1.240,1.490) -- (1.259,1.559) -- (1.278,1.621) -- (1.297,1.678) -- (1.316,1.728) -- (1.335,1.771) -- (1.354,1.807) -- (1.373,1.837) -- (1.392,1.860) -- (1.411,1.877) -- (1.430,1.888) -- (1.449,1.895) -- (1.468,1.899) -- (1.487,1.900) -- (1.500,1.900) -- cycle;
  \fill[red!45] (2.000,1.900) -- (2.000,1.900) -- (2.019,1.901) -- (2.038,1.905) -- (2.057,1.910) -- (2.076,1.918) -- (2.095,1.928) -- (2.114,1.941) -- (2.133,1.955) -- (2.152,1.971) -- (2.171,1.990) -- (2.190,2.010) -- (2.209,2.033) -- (2.228,2.057) -- (2.247,2.083) -- (2.266,2.110) -- (2.285,2.139) -- (2.304,2.170) -- (2.323,2.202) -- (2.342,2.234) -- (2.361,2.269) -- (2.380,2.304) -- (2.399,2.339) -- (2.418,2.376) -- (2.437,2.413) -- (2.456,2.451) -- (2.475,2.488) -- (2.494,2.527) -- (2.513,2.565) -- (2.532,2.603) -- (2.551,2.640) -- (2.570,2.678) -- (2.589,2.715) -- (2.608,2.751) -- (2.627,2.787) -- (2.646,2.821) -- (2.665,2.855) -- (2.684,2.888) -- (2.703,2.919) -- (2.722,2.949) -- (2.741,2.977) -- (2.760,3.004) -- (2.779,3.029) -- (2.798,3.053) -- (2.817,3.075) -- (2.836,3.094) -- (2.855,3.112) -- (2.874,3.128) -- (2.893,3.141) -- (2.912,3.153) -- (2.931,3.162) -- (2.950,3.169) -- (2.969,3.174) -- (2.988,3.177) -- (3.007,3.177) -- (3.026,3.175) -- (3.045,3.171) -- (3.064,3.164) -- (3.083,3.156) -- (3.102,3.145) -- (3.121,3.132) -- (3.140,3.116) -- (3.159,3.099) -- (3.178,3.080) -- (3.197,3.059) -- (3.216,3.036) -- (3.235,3.011) -- (3.254,2.984) -- (3.273,2.956) -- (3.292,2.927) -- (3.311,2.896) -- (3.330,2.864) -- (3.349,2.830) -- (3.368,2.796) -- (3.387,2.761) -- (3.406,2.724) -- (3.425,2.688) -- (3.444,2.650) -- (3.463,2.613) -- (3.482,2.575) -- (3.501,2.537) -- (3.520,2.498) -- (3.539,2.461) -- (3.558,2.423) -- (3.577,2.386) -- (3.596,2.349) -- (3.615,2.313) -- (3.634,2.278) -- (3.653,2.243) -- (3.672,2.210) -- (3.691,2.178) -- (3.710,2.147) -- (3.729,2.118) -- (3.748,2.090) -- (3.767,2.064) -- (3.786,2.039) -- (3.805,2.016) -- (3.824,1.995) -- (3.843,1.976) -- (3.862,1.959) -- (3.881,1.944) -- (3.900,1.931) -- (3.919,1.921) -- (3.938,1.912) -- (3.957,1.906) -- (3.976,1.902) -- (3.995,1.900) -- (4.000,1.900) -- cycle;
  \draw[->] (-1.800,0) -- (6.200,0) node[above left] {$|x-c_m-Lz-a|$};
  \draw[densely dashed] (-1.800,1.900) -- (5.800,1.900);
  \node[left] at (-1.800,1.900) {$\rho_0$};
  \node[left] at (-1.800,0) {$0$};
  \draw[very thick] (-1.800,0) -- (0,0) -- (0.000,0.000) -- (0.005,0.000) -- (0.024,0.000) -- (0.043,0.000) -- (0.062,0.000) -- (0.081,0.000) -- (0.100,0.000) -- (0.119,0.000) -- (0.138,0.000) -- (0.157,0.000) -- (0.176,0.000) -- (0.195,0.000) -- (0.214,0.000) -- (0.233,0.000) -- (0.252,0.000) -- (0.271,0.000) -- (0.290,0.000) -- (0.309,0.000) -- (0.328,0.000) -- (0.347,0.000) -- (0.366,0.000) -- (0.385,0.000) -- (0.404,0.000) -- (0.423,0.000) -- (0.442,0.000) -- (0.461,0.000) -- (0.480,0.000) -- (0.499,0.000) -- (0.500,0.000) -- (0.518,0.000) -- (0.537,0.000) -- (0.556,0.000) -- (0.575,0.000) -- (0.594,0.000) -- (0.613,0.000) -- (0.632,0.001) -- (0.651,0.001) -- (0.670,0.003) -- (0.689,0.005) -- (0.708,0.008) -- (0.727,0.012) -- (0.746,0.019) -- (0.765,0.027) -- (0.784,0.039) -- (0.803,0.053) -- (0.822,0.071) -- (0.841,0.093) -- (0.860,0.120) -- (0.879,0.151) -- (0.898,0.187) -- (0.917,0.229) -- (0.936,0.276) -- (0.955,0.329) -- (0.974,0.387) -- (0.993,0.450) -- (1.012,0.519) -- (1.031,0.592) -- (1.050,0.668) -- (1.069,0.749) -- (1.088,0.832) -- (1.107,0.917) -- (1.126,1.003) -- (1.145,1.089) -- (1.164,1.174) -- (1.183,1.258) -- (1.202,1.339) -- (1.221,1.417) -- (1.240,1.490) -- (1.259,1.559) -- (1.278,1.621) -- (1.297,1.678) -- (1.316,1.728) -- (1.335,1.771) -- (1.354,1.807) -- (1.373,1.837) -- (1.392,1.860) -- (1.411,1.877) -- (1.430,1.888) -- (1.449,1.895) -- (1.468,1.899) -- (1.487,1.900) -- (1.500,1.900) -- (1.506,1.900) -- (1.525,1.900) -- (1.544,1.900) -- (1.563,1.900) -- (1.582,1.900) -- (1.601,1.900) -- (1.620,1.900) -- (1.639,1.900) -- (1.658,1.900) -- (1.677,1.900) -- (1.696,1.900) -- (1.715,1.900) -- (1.734,1.900) -- (1.753,1.900) -- (1.772,1.900) -- (1.791,1.900) -- (1.810,1.900) -- (1.829,1.900) -- (1.848,1.900) -- (1.867,1.900) -- (1.886,1.900) -- (1.905,1.900) -- (1.924,1.900) -- (1.943,1.900) -- (1.962,1.900) -- (1.981,1.900) -- (2.000,1.900) -- (2.000,1.900) -- (2.019,1.901) -- (2.038,1.905) -- (2.057,1.910) -- (2.076,1.918) -- (2.095,1.928) -- (2.114,1.941) -- (2.133,1.955) -- (2.152,1.971) -- (2.171,1.990) -- (2.190,2.010) -- (2.209,2.033) -- (2.228,2.057) -- (2.247,2.083) -- (2.266,2.110) -- (2.285,2.139) -- (2.304,2.170) -- (2.323,2.202) -- (2.342,2.234) -- (2.361,2.269) -- (2.380,2.304) -- (2.399,2.339) -- (2.418,2.376) -- (2.437,2.413) -- (2.456,2.451) -- (2.475,2.488) -- (2.494,2.527) -- (2.513,2.565) -- (2.532,2.603) -- (2.551,2.640) -- (2.570,2.678) -- (2.589,2.715) -- (2.608,2.751) -- (2.627,2.787) -- (2.646,2.821) -- (2.665,2.855) -- (2.684,2.888) -- (2.703,2.919) -- (2.722,2.949) -- (2.741,2.977) -- (2.760,3.004) -- (2.779,3.029) -- (2.798,3.053) -- (2.817,3.075) -- (2.836,3.094) -- (2.855,3.112) -- (2.874,3.128) -- (2.893,3.141) -- (2.912,3.153) -- (2.931,3.162) -- (2.950,3.169) -- (2.969,3.174) -- (2.988,3.177) -- (3.007,3.177) -- (3.026,3.175) -- (3.045,3.171) -- (3.064,3.164) -- (3.083,3.156) -- (3.102,3.145) -- (3.121,3.132) -- (3.140,3.116) -- (3.159,3.099) -- (3.178,3.080) -- (3.197,3.059) -- (3.216,3.036) -- (3.235,3.011) -- (3.254,2.984) -- (3.273,2.956) -- (3.292,2.927) -- (3.311,2.896) -- (3.330,2.864) -- (3.349,2.830) -- (3.368,2.796) -- (3.387,2.761) -- (3.406,2.724) -- (3.425,2.688) -- (3.444,2.650) -- (3.463,2.613) -- (3.482,2.575) -- (3.501,2.537) -- (3.520,2.498) -- (3.539,2.461) -- (3.558,2.423) -- (3.577,2.386) -- (3.596,2.349) -- (3.615,2.313) -- (3.634,2.278) -- (3.653,2.243) -- (3.672,2.210) -- (3.691,2.178) -- (3.710,2.147) -- (3.729,2.118) -- (3.748,2.090) -- (3.767,2.064) -- (3.786,2.039) -- (3.805,2.016) -- (3.824,1.995) -- (3.843,1.976) -- (3.862,1.959) -- (3.881,1.944) -- (3.900,1.931) -- (3.919,1.921) -- (3.938,1.912) -- (3.957,1.906) -- (3.976,1.902) -- (3.995,1.900) -- (4.000,1.900) -- (4.014,1.900) -- (4.033,1.900) -- (4.052,1.900) -- (4.071,1.900) -- (4.090,1.900) -- (4.109,1.900) -- (4.128,1.900) -- (4.147,1.900) -- (4.166,1.900) -- (4.185,1.900) -- (4.204,1.900) -- (4.223,1.900) -- (4.242,1.900) -- (4.261,1.900) -- (4.280,1.900) -- (4.299,1.900) -- (4.318,1.900) -- (4.337,1.900) -- (4.356,1.900) -- (4.375,1.900) -- (4.394,1.900) -- (4.413,1.900) -- (4.432,1.900) -- (4.451,1.900) -- (4.470,1.900) -- (4.489,1.900) -- (4.508,1.900) -- (4.527,1.900) -- (4.546,1.900) -- (4.565,1.900) -- (4.584,1.900) -- (4.603,1.900) -- (4.622,1.900) -- (4.641,1.900) -- (4.660,1.900) -- (4.679,1.900) -- (4.698,1.900) -- (4.717,1.900) -- (4.736,1.900) -- (4.755,1.900) -- (4.774,1.900) -- (4.793,1.900) -- (4.812,1.900) -- (4.831,1.900) -- (4.850,1.900) -- (4.869,1.900) -- (4.888,1.900) -- (4.907,1.900) -- (4.926,1.900) -- (4.945,1.900) -- (4.964,1.900) -- (4.983,1.900) -- (5.002,1.900) -- (5.021,1.900) -- (5.040,1.900) -- (5.059,1.900) -- (5.078,1.900) -- (5.097,1.900) -- (5.116,1.900) -- (5.135,1.900) -- (5.154,1.900) -- (5.173,1.900) -- (5.192,1.900) -- (5.211,1.900) -- (5.230,1.900) -- (5.249,1.900) -- (5.268,1.900) -- (5.287,1.900) -- (5.306,1.900) -- (5.325,1.900) -- (5.344,1.900) -- (5.363,1.900) -- (5.382,1.900) -- (5.401,1.900) -- (5.420,1.900) -- (5.439,1.900) -- (5.458,1.900) -- (5.477,1.900) -- (5.496,1.900) -- (5.515,1.900) -- (5.534,1.900) -- (5.553,1.900) -- (5.572,1.900) -- (5.591,1.900) -- (5.610,1.900) -- (5.629,1.900) -- (5.648,1.900) -- (5.667,1.900) -- (5.686,1.900) -- (5.705,1.900) -- (5.724,1.900) -- (5.743,1.900) -- (5.762,1.900) -- (5.781,1.900) -- (5.800,1.900);
  \foreach \t/\lab in {0/{r_m},2.0/{r_m+\delta},4.0/{r_m+2\delta}}
    \draw (\t,0.08) -- (\t,-0.08) node[below] {$\lab$};
  \node at (-0.900,3.040) {ball};
  \node at (2.000,3.382) {shell};
  \draw[->] (1.000,2.565) node[above] {$g_{r_m,\delta}<0$} -- (0.760,1.425);
  \draw[->] (4.700,2.850) node[right] {$g_{r_m,\delta}>0$} -- (3.200,2.128);
  \node[right] at (4.400,1.558) {$\tilde\rho_{\ell,a}$};
\end{tikzpicture}
\caption{Construction for the upper bound in Theorem~\ref{thm:upper}, shown schematically.
\emph{Top.} For a fixed shift $a$ (arrow), thin lines mark the cells $a+Lz+C_L$.
Selected balls (blue) carry $\Gamma_{z,m,a}$ and have their closed enlarged balls inside $C_{\ell-2\delta}$ (dashed).
Their shells of width $2\delta$ (red) carry $g_{r_m,\delta}$, while dotted balls are discarded.
Outside the selected enlarged balls, $\tilde\rho_{\ell,a}$ equals the smoothed background, supported in $C_\ell+B_\delta$ (light gray) and equal to $\rho_0$ on $C_{\ell-2\delta}$ (dark gray).
After averaging over $a\in C_L$, the density equals $\rho_0$ on $C_{\ell-2\delta'}$ (dash-dotted).
\emph{Bottom.} Radial profile of $\tilde\rho_{\ell,a}$ near a selected ball, with $\delta/r_m=1/4$.
It vanishes in the ball and in a gap of width $\delta/4$, then follows the shell profile $\rho_{r_m,\delta}$.
The charge removed near the ball (gray) is returned in the outer shell (red), so that $g_{r_m,\delta}$ has total charge zero.}
\label{fig:upper_construction}
\end{figure}

For each $m$, apply Lemma~\ref{lem:upper_collar} with $r=r_m$, which is allowed because $r_m\ge r_*>2\delta$. This gives the profile $g_{r_m,\delta}$, supported in $\overline{B_{r_m+2\delta}\setminus B_{r_m}}$ and of integral zero, and the density $\rho_{r_m,\delta}$. Define the density
\begin{equation}\label{eq:upper_expanded_fluid_density}
\begin{aligned}
\tilde\rho_{\ell,a}&=\rho_0\iden_{C_\ell}*\chi_\delta-
\rho_0\sum_{(z,m)\in J_0(a,\ell)}\iden_{B_{r_m}(c_m+Lz+a)}
+G_{\ell,a},\\
G_{\ell,a}&=\sum_{(z,m)\in J_0(a,\ell)}g_{r_m,\delta}(\cdot-c_m-Lz-a).
\end{aligned}
\end{equation}
Let $(z,m)\in J_0(a,\ell)$. Since $\rho_0\iden_{C_\ell}*\chi_\delta=\rho_0$ on $C_{\ell-2\delta}$ and the enlarged balls are disjoint, $\tilde\rho_{\ell,a}=0$ in the ball $B_{r_m}(c_m+Lz+a)$, and
\[
\tilde\rho_{\ell,a}=\rho_0+g_{r_m,\delta}(\cdot-c_m-Lz-a)=\rho_{r_m,\delta}(\cdot-c_m-Lz-a)
\]
in its shell $B_{r_m+2\delta}(c_m+Lz+a)\setminus B_{r_m}(c_m+Lz+a)$; see Figure~\ref{fig:upper_construction} (bottom). In fact, this profile is zero up to distance $\delta/4$ from the boundary of the ball, leaving a positive separation between the support of $\tilde\rho_{\ell,a}$ and the selected balls. Outside all enlarged balls, it equals $\rho_0\iden_{C_\ell}*\chi_\delta$. The shells are disjoint and lie in the plateau of $\rho_0\iden_{C_\ell}*\chi_\delta$. Thus these descriptions match smoothly at their internal interfaces and
\begin{equation}\label{eq:upper_fluid_mass_bound}
0\le \tilde\rho_{\ell,a}\le C\rho_0,
\qquad
\int_{\RR^3}\tilde\rho_{\ell,a}=\rho_0(\ell^3-V(a,\ell)).
\end{equation}

Applying Lemma~\ref{lem:basic_elementary}(ii) to the outer smoothing and adding the disjoint shell gradient integrals yields
\begin{equation}\label{eq:upper_fluid_gradient_bound}
\int_{\RR^3}|\nabla\sqrt{\tilde\rho_{\ell,a}}|^2
\le C\rho_0\ell^2/\delta+C\rho_0A(a,\ell)/\delta.
\end{equation}

Next, we construct a background state in between the ball states. Theorem~\ref{thm:basic_fluid} gives a grand-canonical state $\tilde\vGamma_{\ell,a}$ with density $\tilde\rho_{\ell,a}$ that satisfies~\eqref{eq:basic_fluid_bound} with $\varepsilon=1$. Since $\tilde\rho_{\ell,a}$ vanishes on the selected balls, $\tilde\vGamma_{\ell,a}$ is supported outside them. Using $\int_{\RR^3}\tilde\rho_{\ell,a}^{5/3}\le(C\rho_0)^{2/3}\int_{\RR^3}\tilde\rho_{\ell,a}$, \eqref{eq:upper_fluid_mass_bound} and~\eqref{eq:upper_fluid_gradient_bound}, its indirect energy satisfies
\begin{equation}\label{eq:upper_fluid_energy_expanded}
\begin{aligned}
\EC(\tilde\vGamma_{\ell,a})
&\le C\int_{\RR^3}\tilde\rho_{\ell,a}^{5/3}
+C\int_{\RR^3}|\nabla\sqrt{\tilde\rho_{\ell,a}}|^2\\
&\le C\rho_0^{5/3}(\ell^3-V(a,\ell))
+C\rho_0\delta^{-1}\big(\ell^2+A(a,\ell)\big).
\end{aligned}
\end{equation}
The constants $C$ here do not depend on the packing.

For $(z,m)\in J_0(a,\ell)$, the state $\Gamma_{z,m,a}$ is supported in $B_{r_m}(c_m+Lz+a)$, and $\tilde\vGamma_{\ell,a}$ is supported outside all selected balls. By taking the antisymmetric product of these states (Lemma~\ref{lem:basic_product}), we obtain
\[
\vGamma_{\ell,a}=\tilde\vGamma_{\ell,a}\wedge\bigwedge_{(z,m)\in J_0(a,\ell)}\Gamma_{z,m,a}.
\]
Its mean particle number is exactly $\rho_0\ell^3$. Indeed, by~\eqref{eq:basic_product_density}, \eqref{eq:upper_fluid_mass_bound}, the fact that $\Gamma_{z,m,a}$ has $\rho_0|B_{r_m}|$ particles, and by the definition of $V(a,\ell)$,
\begin{align*}
\int_{\RR^3}\rho_{\vGamma_{\ell,a}}
&=\int_{\RR^3}\tilde\rho_{\ell,a}+\sum_{(z,m)\in J_0(a,\ell)}\int_{\RR^3}\rho_{\Gamma_{z,m,a}}\\
&=\rho_0\big(\ell^3-V(a,\ell)\big)+\rho_0\sum_{(z,m)\in J_0(a,\ell)}|B_{r_m}|
=\rho_0\ell^3.
\end{align*}

Next we compute the Jellium energy of our state $\vGamma_{\ell,a}$. To apply~\eqref{eq:basic_product_J}, we split the background into one part for each factor of $\vGamma_{\ell,a}$,
\begin{align*}
\rho_0\iden_{C_\ell}*\chi_\delta
&=\Big(\rho_0\iden_{C_\ell}*\chi_\delta-\sum_{(z,m)\in J_0(a,\ell)}\rho_0\iden_{B_{r_m}(c_m+Lz+a)}\Big)\\
&\quad+\sum_{(z,m)\in J_0(a,\ell)}\rho_0\iden_{B_{r_m}(c_m+Lz+a)}.
\end{align*}
The first part is the background belonging to $\tilde\vGamma_{\ell,a}$, and $\rho_0\iden_{B_{r_m}(c_m+Lz+a)}$ is the background belonging to $\Gamma_{z,m,a}$. The first part is nonnegative, because the selected balls lie in $C_{\ell-2\delta}$, where $\rho_0\iden_{C_\ell}*\chi_\delta=\rho_0$. By~\eqref{eq:upper_expanded_fluid_density}, the density $\tilde\rho_{\ell,a}$ of $\tilde\vGamma_{\ell,a}$ differs from it by exactly $G_{\ell,a}$. Then \eqref{eq:basic_product_J} gives the exact equality
\begin{align*}
&\mathcal J_{\rho_0\iden_{C_\ell}*\chi_\delta}(\vGamma_{\ell,a})\\
&=\sum_{(z,m)\in J_0(a,\ell)}
\mathcal J_{\rho_0\iden_{B_{r_m}(c_m+Lz+a)}}(\Gamma_{z,m,a})
+\mathcal J_{\tilde\rho_{\ell,a}-G_{\ell,a}}(\tilde\vGamma_{\ell,a})\\
&+\sum_{\substack{(z,m),(z',m')\in J_0(a,\ell)\\(z,m)\ne(z',m')}}
D\big(\rho_{\Gamma_{z,m,a}}-\rho_0\iden_{B_{r_m}(c_m+Lz+a)},\rho_{\Gamma_{z',m',a}}-\rho_0\iden_{B_{r_{m'}}(c_{m'}+Lz'+a)}\big)\\
&\quad+2\sum_{(z,m)\in J_0(a,\ell)}D\big(\rho_{\Gamma_{z,m,a}}-\rho_0\iden_{B_{r_m}(c_m+Lz+a)},G_{\ell,a}\big)\\
&=\sum_{(z,m)\in J_0(a,\ell)}
\mathcal J_{\rho_0\iden_{B_{r_m}(c_m+Lz+a)}}(\Gamma_{z,m,a})
+\EC(\tilde\vGamma_{\ell,a})+D(G_{\ell,a}).
\end{align*}
For the second equality, \eqref{eq:conj_identity} gives $\mathcal J_{\tilde\rho_{\ell,a}-G_{\ell,a}}(\tilde\vGamma_{\ell,a})=\EC(\tilde\vGamma_{\ell,a})+D(G_{\ell,a})$, since $\tilde\rho_{\ell,a}$ is the density of $\tilde\vGamma_{\ell,a}$. Moreover, for $(z,m)\in J_0(a,\ell)$ the net charge $\rho_{\Gamma_{z,m,a}}-\rho_0\iden_{B_{r_m}(c_m+Lz+a)}$ of the selected ball has zero potential outside $B_{r_m}(c_m+Lz+a)$ by~\eqref{eq:upper_core_zero_potential}. Every other selected ball and every shell lies outside that ball, including the shell surrounding the same ball, so both sums of cross terms vanish. Similarly, each term $g_{r_m,\delta}(\cdot-c_m-Lz-a)$ of $G_{\ell,a}$ has zero potential outside $B_{r_m+2\delta}(c_m+Lz+a)$, and the enlarged balls are disjoint. Thus the cross terms in $D(G_{\ell,a})$ vanish, and by translation invariance
\[
D(G_{\ell,a})=\sum_{(z,m)\in J_0(a,\ell)}D\big(g_{r_m,\delta}(\cdot-c_m-Lz-a)\big)=\sum_{(z,m)\in J_0(a,\ell)}D(g_{r_m,\delta}).
\]
By translation invariance, each state $\Gamma_{z,m,a}$ has Jellium energy $E_\Jel(\rho_0,B_{r_m})$. We obtain
\begin{equation}\label{eq:upper_exact_identity}
\mathcal J_{\rho_0\iden_{C_\ell}*\chi_\delta}(\vGamma_{\ell,a})
=\sum_{(z,m)\in J_0(a,\ell)}
E_\Jel(\rho_0,B_{r_m})
+\EC(\tilde\vGamma_{\ell,a})
+\sum_{(z,m)\in J_0(a,\ell)}D(g_{r_m,\delta}).
\end{equation}
Note that the background $\rho_0\iden_{C_\ell}*\chi_\delta$ in this identity is the same for every $a$.
Using \eqref{eq:upper_core_energy}, \eqref{eq:upper_collar_costs} and \eqref{eq:upper_fluid_energy_expanded} gives
\begin{equation}\label{eq:upper_prelimit_estimate}
\begin{aligned}
\mathcal J_{\rho_0\iden_{C_\ell}*\chi_\delta}(\vGamma_{\ell,a})
&\le(e_\Jel(\rho_0)+\varepsilon)V(a,\ell)
+C\rho_0^{5/3}(\ell^3-V(a,\ell))\\
&\quad+C\left(\rho_0\delta^{-1}+\rho_0^2\delta^3\right)A(a,\ell)+C\rho_0\delta^{-1}\ell^2,
\end{aligned}
\end{equation}
where the constants $C$ do not depend on the packing. We now let $\ell\to\infty$ with the packing held fixed, as the error term in~\eqref{eq:upper_explicit_counting} depends on $L$. By~\eqref{eq:upper_explicit_counting}, the terms of order $\ell^2$ disappear after division by $\ell^3$. We arrive at
\begin{equation}\label{eq:upper_fixed_packing}
\limsup_{\ell\to\infty}\sup_{a\in C_L}
\frac{\mathcal J_{\rho_0\iden_{C_\ell}*\chi_\delta}(\vGamma_{\ell,a})}{\ell^3}
\le\theta(e_\Jel(\rho_0)+\varepsilon)+C\rho_0^{5/3}(1-\theta)
+\frac{C}{r_*}\left(\rho_0\delta^{-1}+\rho_0^2\delta^3\right).
\end{equation}

Define
\[
\overline{\vGamma}_\ell=\frac1{L^3}\int_{C_L}\vGamma_{\ell,a}\,\mathrm da.
\]
The average is well defined. In fact, for fixed $\ell$ and $L$, only finitely many pairs $(z,m)\in\ZZ^3\times\{1,\dots,M\}$ belong to $J_0(a,\ell)$ for some $a\in C_L$, since the corresponding balls lie in $C_\ell$ and $a$ ranges over the bounded set $C_L$. For each such pair, the set of $a\in C_L$ with $(z,m)\in J_0(a,\ell)$ is the intersection of $C_L$ with an open cube. Hence $C_L$ splits into finitely many measurable sets and on each $J_0(a,\ell)$ is constant. On each of these sets, translation is continuous for the density matrices $\Gamma_{z,m,a}$ in trace norm and for $\tilde\rho_{\ell,a}$ in $L^1$. We choose the states $\tilde\vGamma_{\ell,a}$ using the explicit quasi-free construction in~\cite[proof of Thm.~8 and Sec.~3.3]{lewin2020local}, which depends measurably on $\tilde\rho_{\ell,a}$. Thus the full family is measurable. Their kinetic energies are uniformly bounded for fixed $\ell,L$. Thus the averaged state has trace one and finite kinetic energy.

Before averaging, \eqref{eq:basic_product_density} and~\eqref{eq:upper_expanded_fluid_density} give its density as
\begin{equation}\label{eq:upper_density_before_average}
\begin{aligned}
\rho_{\vGamma_{\ell,a}}
&=\tilde\rho_{\ell,a}+\sum_{(z,m)\in J_0(a,\ell)}\rho_{\Gamma_{r_m}}(\cdot-c_m-Lz-a)\\
&=\rho_0\iden_{C_\ell}*\chi_\delta+
\sum_{(z,m)\in J_0(a,\ell)}(\rho_{\Gamma_{r_m}}-\rho_0\iden_{B_{r_m}}+g_{r_m,\delta})(\cdot-c_m-Lz-a).
\end{aligned}
\end{equation}
Set $\delta'=5\delta+2\max_{1\le m\le M}r_m$.
Fix $x\in C_\ell$ with $d(x,\partial C_\ell)>\delta'$. For every $a\in C_L$, the translated terms $(\rho_{\Gamma_{r_m}}-\rho_0\iden_{B_{r_m}}+g_{r_m,\delta})(\cdot-c_m-Lz-a)$ with $(z,m)\notin J_0(a,\ell)$ vanish at $x$: such a term is supported in its closed enlarged ball, which has diameter $2r_m+4\delta\le\delta'-\delta$ and would therefore lie in $C_{\ell-2\delta}$ if it contained $x$. Hence, at $x$, the second sum in~\eqref{eq:upper_density_before_average} may be taken over all $(z,m)\in\ZZ^3\times\{1,\dots,M\}$.

Now fix $m=1,\ldots,M$. As $a$ ranges over $C_L$ and $z$ over $\ZZ^3$, the point $Lz+a$ covers $\RR^3$ exactly once up to a null set, because the cells $C_L+Lz$ tile $\RR^3$. The change of variables $y=x-c_m-Lz-a$ therefore gives
\begin{align*}
&\frac1{L^3}\int_{C_L}\sum_{z\in\ZZ^3}
(\rho_{\Gamma_{r_m}}-\rho_0\iden_{B_{r_m}}+g_{r_m,\delta})(x-c_m-Lz-a)\,\mathrm da\\
&\qquad=\frac1{L^3}\int_{\RR^3}(\rho_{\Gamma_{r_m}}-\rho_0\iden_{B_{r_m}}+g_{r_m,\delta})(y)\,\mathrm dy
=0,
\end{align*}
because the charge distribution $\rho_{\Gamma_{r_m}}-\rho_0\iden_{B_{r_m}}$ of the ball and the shell profile $g_{r_m,\delta}$ both have integral zero. The density of $\overline{\vGamma}_\ell$ is the average over $a\in C_L$ of~\eqref{eq:upper_density_before_average}, and $(\rho_0\iden_{C_\ell}*\chi_\delta)(x)=\rho_0$ because $x\in C_{\ell-2\delta}$. We conclude that
\begin{equation}\label{eq:upper_density_bulk_exact}
\rho_{\overline{\vGamma}_\ell}(x)=\rho_0
\quad\hbox{whenever }x\in C_\ell\hbox{ and }d(x,\partial C_\ell)>\delta'.
\end{equation}

There is also a pointwise bound on $\rho_{\overline{\vGamma}_\ell}$ that is independent of the packing. By the first line of~\eqref{eq:upper_density_before_average}, $\rho_{\overline{\vGamma}_\ell}$ is the average over $a\in C_L$ of $\tilde\rho_{\ell,a}$ plus the sum of the densities $\rho_{\Gamma_{z,m,a}}$. The density $\tilde\rho_{\ell,a}$ is at most $C\rho_0$ by~\eqref{eq:upper_fluid_mass_bound}. In the average of this sum we may extend the sum to all $(z,m)\in\ZZ^3\times\{1,\dots,M\}$, since the added terms are nonnegative. The same change of variables as above then gives, for every $x\in\RR^3$,
\begin{align*}
&\frac1{L^3}\int_{C_L}\sum_{(z,m)\in J_0(a,\ell)}
\rho_{\Gamma_{r_m}}(x-c_m-Lz-a)\,\mathrm da\\
&\qquad\le\frac1{L^3}\sum_{m=1}^M\int_{C_L}\sum_{z\in\ZZ^3}\rho_{\Gamma_{r_m}}(x-c_m-Lz-a)\,\mathrm da\\
&\qquad=\frac1{L^3}\sum_{m=1}^M\int_{\RR^3}\rho_{\Gamma_{r_m}}
=\frac{\rho_0}{L^3}\sum_{m=1}^M|B_{r_m}|
=\theta\rho_0,
\end{align*}
because $\Gamma_{r_m}$ has $\rho_0|B_{r_m}|$ particles. Since $\theta\le1$, we conclude that
\begin{equation}\label{eq:upper_density_uniform_bound}
0\le \rho_{\overline{\vGamma}_\ell}\le(C+1)\rho_0\iden_{C_\ell+B_\delta}.
\end{equation}
The averaged density $\rho_{\overline{\vGamma}_\ell}$ vanishes outside $C_\ell+B_\delta$. Indeed, for every $a$, all selected enlarged balls lie in $C_{\ell-2\delta}$, and outside these balls the density $\rho_{\vGamma_{\ell,a}}$ equals $\rho_0\iden_{C_\ell}*\chi_\delta$, which vanishes outside $C_\ell+B_\delta$. Averaging over $a$ preserves this property. Finite kinetic energy of $\overline{\vGamma}_\ell$ gives $\sqrt{\rho_{\overline{\vGamma}_\ell}}\in H^1(\RR^3)$ by the Hoffmann--Ostenhof inequality~\cite{hoffmann1977schrodinger}.

All hypotheses of Lemma~\ref{lem:upper_bounded_layer} hold for $\rho_{\overline{\vGamma}_\ell}$, with the layer width $\delta'$, which does not depend on $\ell$, and with the packing-independent constant $\Lambda=C+1$ as the density bound. Therefore
\[
e_\UEG(\rho_0)
\le\liminf_{\ell\to\infty}\frac{E(\rho_{\overline{\vGamma}_\ell})}{\ell^3}.
\]
On the other hand, the averaged state $\overline{\vGamma}_\ell$ is a competitor for $E(\rho_{\overline{\vGamma}_\ell})$, and~\eqref{eq:conj_identity} with the fixed background $\zeta=\rho_0\iden_{C_\ell}*\chi_\delta$ gives
\begin{align*}
E(\rho_{\overline{\vGamma}_\ell})
&\le\EC(\overline{\vGamma}_\ell)\\
&=\mathcal J_{\rho_0\iden_{C_\ell}*\chi_\delta}(\overline{\vGamma}_\ell)
-D(\rho_{\overline{\vGamma}_\ell}-\rho_0\iden_{C_\ell}*\chi_\delta)\\
&\le\frac1{L^3}\int_{C_L}\mathcal J_{\rho_0\iden_{C_\ell}*\chi_\delta}(\vGamma_{\ell,a})\,\mathrm da.
\end{align*}
Combining the bound from Lemma~\ref{lem:upper_bounded_layer} with the preceding estimate and~\eqref{eq:upper_fixed_packing} yields
\begin{equation}\label{eq:upper_quantitative_error}
\begin{aligned}
e_\UEG(\rho_0)
&\le\theta(e_\Jel(\rho_0)+\varepsilon)+C\rho_0^{5/3}(1-\theta)+\frac{C}{r_*}\left(\rho_0\delta^{-1}+\rho_0^2\delta^3\right)\\
&\le e_\Jel(\rho_0)+
\varepsilon\bigl(1+e_\Jel(\rho_0)_-+C\rho_0^{5/3}\bigr)
+\frac{C}{r_*}\left(\rho_0\delta^{-1}+\rho_0^2\delta^3\right).
\end{aligned}
\end{equation}
For the second inequality, write $e_\Jel(\rho_0)=e_\Jel(\rho_0)_+-e_\Jel(\rho_0)_-$ with $e_\Jel(\rho_0)_\pm=\max\{\pm e_\Jel(\rho_0),0\}$, and use $1-\varepsilon\le\theta\le1$, which gives $\theta e_\Jel(\rho_0)_+\le e_\Jel(\rho_0)_+$ and $-\theta e_\Jel(\rho_0)_-\le-(1-\varepsilon)e_\Jel(\rho_0)_-$.

To finish, keep $\delta$ fixed. Estimate~\eqref{eq:upper_quantitative_error} holds for every $\varepsilon\in(0,1)$ and every $r_*$ large enough for~\eqref{eq:upper_core_energy}. Choosing $r_*$ also so large that $C(\rho_0\delta^{-1}+\rho_0^2\delta^3)/r_*\le\varepsilon$ and letting $\varepsilon\downarrow0$ proves the theorem.
\end{proof}

\appendix

\section{Proof sketch for the thermodynamic limit of Jellium}\label{sec:LN_proof}

We sketch the proof of Theorem~\ref{thm:LN_main}(i) at $T=0$. The arguments of~\cite{lieb1975thermodynamic} are given at positive temperature, and they carry over to the zero temperature case when Gibbs states are replaced by ground states, with one extra observation. Newton's theorem requires the states in the packed balls to have radial densities. A ground state $\Psi$ in a ball, which may be degenerate, can be replaced by its rotation average $\int_{\SO(3)}R|\Psi\rangle\langle\Psi|R^*\,\mathrm{d}R$, which has the same energy and a radial density. The method is the one Lieb and Lebowitz~\cite{lieb1972constitution} introduced for real matter: pack large balls with smaller neutral balls and use Newton's theorem to remove the interactions between them. The new difficulty for Jellium is that the background is rigid. The part of a large ball not covered by the packing cannot be left empty, since its uncompensated background charge would cost an electrostatic energy which is not $o(|\mathcal B_K|)$~\cite[Sec.~3.1]{lieb1975thermodynamic}. One therefore has to place the right number of particles there and control their energy.

\medskip

\noindent\textbf{Step 0: Newton's theorem.} By Theorem~\ref{thm:basic_newton}(i), if $\tau$ is a neutral radial charge distribution supported in a closed ball $B$, then $\int_{\RR^3}\tau(y)|x-y|^{-1}\,\mathrm dy=0$ for $x\notin B$, so $D(\tau,\tau')=0$ for every $\tau'$ supported outside $B$. In particular, if $\Gamma$ is a state of $\rho_0|B|$ particles in $B$ with radial density $\rho_\Gamma$, the total charge $\rho_\Gamma-\rho_0\iden_B$ does not interact with anything outside $B$.

\medskip

\noindent\textbf{Step 1: Swiss-cheese packing.} See Figure~\ref{fig:LN_packing}.

\begin{lemma}[{\cite[Secs.~III--IV]{lieb1972constitution}}, {\cite[Secs.~3.1 and~5.1]{lieb1975thermodynamic}}]\label{lem:LN_packing}
Let $p=26$ and $\gamma:=p/(1+p)=26/27$, and let $\mathcal B_k$, $k\ge0$, be the balls of radii $R_k=R_0\,27^k$ with $\rho_0|\mathcal B_0|=1$, so that $N_k:=\rho_0|\mathcal B_k|=27^{3k}$. For each $K\ge1$ the ball $\mathcal B_K$ contains pairwise disjoint balls, namely $n_{K-j}:=p^{K-j-1}(1+p)^{2(K-j)}$ translates of $\mathcal B_j$ for each $0\le j\le K-1$. The uncovered set $\mathcal D_K$ satisfies $|\mathcal D_K|=\gamma^K|\mathcal B_K|$, so that $M_K:=\rho_0|\mathcal D_K|=26^K\,27^{2K}\in\NN$.

For the quantum problem we use the radii $R_k=R_0'(27^k-\frac12)$, still with $\rho_0|\mathcal B_0|=1$, so that $N_k=(2\cdot27^k-1)^3\in\NN$. With the same multiplicities, the packing can then be chosen so that each packed ball is at distance at least $4h$ from $\partial \mathcal B_K$ and at least $8h$ from the other packed balls, where $h:=R_0'(1-\frac1{27})/8$. For this packing,
\[
M_K:=\rho_0|\mathcal D_K|=N_K-\sum_{j=0}^{K-1}n_{K-j}N_j\in\NN,
\]
and $|\mathcal D_K|/|\mathcal B_K|=\mathcal O(\gamma^K)$.
\end{lemma}

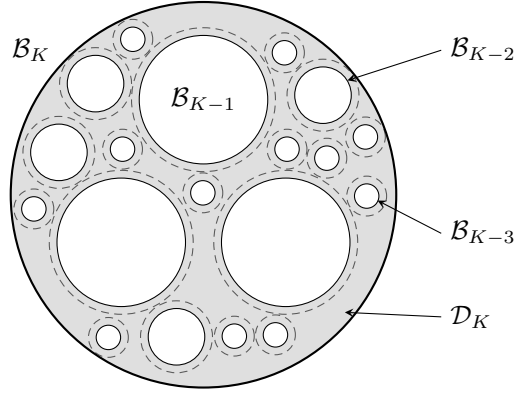
\begin{figure}[ht]
\centering
\begin{tikzpicture}[scale=0.85,>=stealth]
  \def\h{0.12}
  \fill[gray!25] (0,0) circle (3);
  \foreach \r/\centers in {
      1.0/{0.00/1.48,-1.28/-0.74,1.28/-0.74},
      0.44/{-1.68/1.73,-2.25/0.66,-0.42/-2.21,1.86/1.55},
      0.19/{2.54/-0.02,2.52/0.90,1.12/-2.18,1.30/0.71,-1.26/0.71,-2.64/-0.22,-1.48/-2.22,1.26/2.21,0.48/-2.20,-0.01/0.03,1.92/0.57,-1.10/2.42}}
    \foreach \x/\y in \centers {
      \draw[densely dashed,gray!70!black,thin] (\x,\y) circle ({\r+\h});
      \fill[white] (\x,\y) circle (\r);
      \draw (\x,\y) circle (\r);
    }
  \draw[thick] (0,0) circle (3);
  \node at (0,1.48) {$\mathcal B_{K-1}$};
  \draw[->] (3.7,2.25) node[right] {$\mathcal B_{K-2}$} -- (2.24,1.78);
  \draw[->] (3.7,-0.6) node[right] {$\mathcal B_{K-3}$} -- (2.71,-0.09);
  \draw[->] (3.7,-1.9) node[right] {$\mathcal D_K$} -- (2.2,-1.83);
  \node at (140:3.5) {$\mathcal B_K$};
\end{tikzpicture}
\caption{Schematic cross-section of the Swiss-cheese packing of Lemma~\ref{lem:LN_packing}, not to scale. The ball $\mathcal B_K$ contains pairwise disjoint translates of $\mathcal B_{K-1},\mathcal B_{K-2},\dots,\mathcal B_0$, of which three generations are shown. The uncovered set $\mathcal D_K$ is shaded. In the quantum version, the dashed shells of width $4h$ around the packed balls are pairwise disjoint and contained in $\mathcal B_K$. Equivalently, each packed ball is at distance at least $4h$ from $\partial \mathcal B_K$ and at least $8h$ from the other packed balls.}
\label{fig:LN_packing}
\end{figure}

\noindent\textbf{Step 2: The basic inequality for balls.} Let $\Gamma_j$ be rotation-averaged ground states in $\mathcal B_j$ and let $\Psi_K$ be the trial state on $\mathcal D_K$ from Lemma~\ref{lem:LN_trial} below. The antisymmetrized tensor product of copies of the $\Gamma_j$, placed in the packed balls, with $\Psi_K$ is an admissible state of $N_K$ particles in $\mathcal B_K$, since the supports are disjoint and lie in $\mathcal B_K$. By Step~0 all the cross interactions vanish, because the packed balls are neutral with radial charge. Hence
\[
E_\Jel(\rho_0,\mathcal B_K)\le\sum_{j=0}^{K-1}n_{K-j}\,E_\Jel(\rho_0,\mathcal B_j)+\big\langle\Psi_K,H^{\rho_0}_{\mathcal D_K}\Psi_K\big\rangle.
\]
With $e_k:=E_\Jel(\rho_0,\mathcal B_k)/|\mathcal B_k|$ and $w_{K,j}:=n_{K-j}|\mathcal B_j|/|\mathcal B_K|$ (which equals $p^{-1}\gamma^{K-j}$ when $R_k=R_0\,27^k$), this reads
\begin{equation}\label{eq:LN_recursion}
e_K\le\sum_{j=0}^{K-1}w_{K,j}\,e_j+\frac{\langle\Psi_K,H^{\rho_0}_{\mathcal D_K}\Psi_K\rangle}{|\mathcal B_K|},\qquad\sum_{j=0}^{K-1}w_{K,j}=1-\frac{|\mathcal D_K|}{|\mathcal B_K|}.
\end{equation}

\noindent\textbf{Step 3: The leftover region.} Classically, let $M=\rho_0|\mathcal D_K|$. Writing $Y=(y_1,\ldots,y_M)$ and $\mathrm dY=\mathrm dy_1\cdots\mathrm dy_M$, integration against the uniform probability measure on $\mathcal D_K^M$ gives
\[
\frac1{|\mathcal D_K|^M}\int_{\mathcal D_K^M}U^{\rho_0}_{\mathcal D_K}(Y)\,\mathrm dY
=-\frac{\rho_0}{|\mathcal D_K|}\,D(\iden_{\mathcal D_K})\le0
\]
\cite[Eq.~(11)]{lieb1975thermodynamic}. Consequently, there exists a configuration $Y\in\mathcal D_K^M$ such that $U^{\rho_0}_{\mathcal D_K}(Y)\le0$. A substitute in the quantum case is provided by
\begin{lemma}[Trial state on the leftover region {\cite[Sec.~5.1]{lieb1975thermodynamic}}]\label{lem:LN_trial}
For the quantum packing of Lemma~\ref{lem:LN_packing} there exist $C=C(\rho_0)$, independent of $K$, and normalized states $\Psi_K\in\bigwedge^{M_K}H^1_0(\mathcal D_K;\CC^q)$ with
\[
\big\langle\Psi_K,H^{\rho_0}_{\mathcal D_K}\Psi_K\big\rangle\le C\,|\mathcal D_K|.
\]
\end{lemma}

\begin{proof}[Sketch]

(a) \emph{Penalization.} For $Y=(y_1,\ldots,y_{M_K})\in\mathcal D_K^{M_K}$, define
\[
U'(Y):=U^{\rho_0}_{\mathcal D_K}(Y)
+8\pi^2\sum_{1\le i<j\le M_K}\frac{\iden_{\{|y_i-y_j|\le1\}}}{|y_i-y_j|^2}.
\]
The added term penalizes close pairs and is integrable on $\mathcal D_K^{M_K}$.

(b) \emph{Keeping away from $\partial\mathcal D_K$.} Around each packed ball $\mathcal B^i$ (radius $R^i$) take the shells $S^i$ and $T^i$ of radii $(R^i,R^i+h)$ and $(R^i+h,R^i+2h)$, and inside $\partial \mathcal B_K$ the shells $S_K$ and $T_K$ of radii $(R_K-h,R_K)$ and $(R_K-2h,R_K-h)$. By the gaps in Lemma~\ref{lem:LN_packing}, all of them are disjoint and contained in $\mathcal D_K$. Let $f:=0$ on the $S$-shells, $f:=1+|S^i|/|T^i|$ on $T^i$, $f:=1+|S_K|/|T_K|$ on $T_K$ and $f:=1$ on the rest of $\mathcal D_K$, and extend $f$ by zero outside $\mathcal D_K$. Then $\int_{\mathcal D_K}f=|\mathcal D_K|$ and $f\le3$.

(c) \emph{Averaging.} Integrate $U'$ over $\mathcal D_K^{M_K}$ against the product of the probability densities $f(y_i)/|\mathcal D_K|$. A direct computation gives
\[
\frac1{|\mathcal D_K|^{M_K}}\int_{\mathcal D_K^{M_K}}U^{\rho_0}_{\mathcal D_K}(Y)\prod_{i=1}^{M_K} f(y_i)\,\mathrm dY
=D\big(\rho_0(\iden_{\mathcal D_K}-f)\big)-\frac{\rho_0}{|\mathcal D_K|}D(f)
\]
\cite[Eq.~(49)]{lieb1975thermodynamic}. Here $\rho_0(\iden_{\mathcal D_K}-f)$ is a sum of neutral radial double layers $\tau^i$ (on $S^i\cup T^i$) and $\tau_K$ (on $S_K\cup T_K$). By Step~0 its self-energy is $\sum_iD(\tau^i)+D(\tau_K)\le C\rho_0^2h^2|\mathcal D_K|$, because each double layer has energy of order $\rho_0^2h^2$ times its volume and all the shells lie in $\mathcal D_K$. Moreover,
\[
\frac{8\pi^2}{|\mathcal D_K|^{M_K}}\int_{\mathcal D_K^{M_K}}
\Bigg(\sum_{1\le i<j\le M_K}\frac{\iden_{\{|y_i-y_j|\le1\}}}{|y_i-y_j|^2}\Bigg)
\prod_{i=1}^{M_K} f(y_i)\,\mathrm dY
\le C\rho_0M_K.
\]
Hence some $Y$ with all $y_i$ at distance $\ge h$ from $\partial\mathcal D_K$ satisfies $U'(Y)\le C|\mathcal D_K|$.

(d) \emph{Orbitals.} Let
\begin{align*}
d_i&:=\frac12\min\{2h,1,\min_{j\ne i}|y_i-y_j|\},\\
\varphi_i(x)&=(2\pi d_i)^{-1/2}|x-y_i|^{-1}\sin(\pi|x-y_i|/d_i)\,\iden_{B_{d_i}(y_i)}(x),
\end{align*}
where $\varphi_i$ is the Dirichlet ground state of the ball $B_{d_i}(y_i)\subset\mathcal D_K$. These balls are pairwise disjoint, so the Slater determinant $\Psi_K$ of the $\varphi_i$ (with any fixed spins) has the same energy as the Hartree product.

(e) \emph{Energy.} The kinetic energy satisfies
\begin{align*}
\sum_{i=1}^{M_K}\Big(\frac{\pi}{d_i}\Big)^2
&\le8\pi^2\sum_{1\le i<j\le M_K}\frac{\iden_{\{|y_i-y_j|\le1\}}}{|y_i-y_j|^2}\\
&\quad+4\pi^2\max\{(2h)^{-2},1\}\,M_K.
\end{align*}
By Newton's theorem the pair interaction equals that of the point charges $y_i$. The background attraction decreases, since $V_{\mathcal D_K}$ is superharmonic, but per particle by at most
\[
\rho_0\int_{\RR^3}\Big(\frac{1}{|z|}-\int_{\RR^3}\frac{|\varphi_i(y_i+y)|^2}{|z-y|}\,\mathrm{d}y\Big)\,\mathrm{d}z=\xi\rho_0d_i^2,
\]
for a universal constant $\xi$. Altogether $\langle\Psi_K,H^{\rho_0}_{\mathcal D_K}\Psi_K\rangle\le U'(Y)+C'M_K\le C|\mathcal D_K|$.
\end{proof}

\noindent\textbf{Step 4: Convergence for balls.} We use the following elementary lemma on sequences (cf.~\cite[Sec.~IV.D]{lieb1972constitution}).

\begin{lemma}\label{lem:LN_sequences}
Let $0<\gamma<1$ and let $(a_K)_{K\ge0}$ be a real sequence bounded below. Suppose that, for some $C\ge0$ and every $K\ge1$,
\[
a_K\le(1-\gamma)\sum_{j=0}^{K-1}\gamma^{K-1-j}a_j+C\gamma^K.
\]
Then $(a_K)$ converges to a finite limit.
\end{lemma}

\begin{proof}
Set
\[
b_K:=(1-\gamma)\sum_{j=0}^{K-1}\gamma^{K-1-j}a_j+C\gamma^K.
\]
For $K\ge1$, the hypothesis gives
\[
b_{K+1}=\gamma b_K+(1-\gamma)a_K\le b_K.
\]
Moreover, $b_K$ is bounded below, since it is a convex combination of $C,a_0,\ldots,a_{K-1}$. Thus $b_K\to b$ for some $b\in\RR$, and
$a_K=\frac{b_{K+1}-\gamma b_K}{1-\gamma}\to b$.
\end{proof}

By Lemma~\ref{lem:LN_trial} and Lemma~\ref{lem:LN_packing}, the last term in~\eqref{eq:LN_recursion} is at most $C\gamma^K$, while $e_k\ge-c_{\rm LN}\rho_0^{4/3}$ by Theorem~\ref{thm:LN_stability}. For the radii $R_k=R_0\,27^k$, $w_{K,j}=p^{-1}\gamma^{K-j}=(1-\gamma)\gamma^{K-1-j}$, so Lemma~\ref{lem:LN_sequences} applies directly to $a_K=e_K$.

For the quantum radii $R_k=R_0'(27^k-\frac12)$, put
\[
\alpha_K:=\left(1-\frac{1}{2\cdot27^K}\right)^3,
\qquad a_K:=\alpha_K e_K.
\]
The volume formula and the packing multiplicities give
\[
\alpha_K w_{K,j}=(1-\gamma)\gamma^{K-1-j}\alpha_j.
\]
Multiplying~\eqref{eq:LN_recursion} by $\alpha_K\le1$ therefore gives the inequality in Lemma~\ref{lem:LN_sequences} for $(a_K)$. Since $1/8\le\alpha_K\le1$, this sequence is bounded below and hence converges. Finally, $\alpha_K\to1$, so $e_K=a_K/\alpha_K$ has the same finite limit. This proves Theorem~\ref{thm:LN_main}(i).

\medskip

\begin{remark}[Neutral balls]\label{rem:LN_balls}
For every integer $N\ge1$, define $R_N$ by $\rho_0|B_{R_N}|=N$. The neutral balls $\Omega_N=B_{R_N}$ form an LN-admissible sequence. Conditions (A) and (B) are clear, and (i) holds because $|\Sigma^{+1}_{N,h}|\le|\Sigma^{+2}_{N,h}|$. For each fixed $h>0$, $R_N\ge4h$ holds for all sufficiently large $N$, so $|\Sigma^{-2}_{N,h}|>0$ on this tail. Moreover, $R_N\ge4h$ holds for every $N\ge1$ once $4h<R_1$, as required for the uniform estimate in~(ii). Under this condition the double layers are radial about the origin, neutral, bounded by $4$, and supported in $\{R_N\le|x|\le R_N+2h\}$ and $\{R_N-2h\le|x|\le R_N\}$, respectively. For $0<|x|\le R_N+2h$ put $m=\max\{|x|,R_N\}$. Newton's formula~\eqref{eq:basic_newton} and neutrality give
\[
\varphi_{N,h}(x)=4\pi\int_m^{R_N+2h}\sigma^+_{N,h}(s)\,s\Big(1-\frac sm\Big)\,\mathrm ds.
\]
Since $0\le s/m-1\le2h/R_N$ in the integral, $|\varphi_{N,h}(x)|\le16\pi h^2(1+2h/R_N)\le24\pi h^2$, which is (ii). Finally, the charge of either double layer inside the sphere of radius $t$ vanishes for $t$ outside an interval of length $2h$ contained in $[R_N/2,\infty)$, and it is at most $CR_N^2h$ in absolute value. Hence Theorem~\ref{thm:basic_newton}(iv) gives $D(\sigma^+_{N,h}),D(\sigma^-_{N,h})\le CR_N^2h^3=o(|\Omega_N|)$, which is (iii) and (iv).
\end{remark}

\section{Neutralization in a shell}\label{sec:canonical_completion}

We prove Lemma~\ref{lem:lower_completion}. The construction adapts the sampling and localized orbitals of~\cite[Sec.~5.1]{lieb1975thermodynamic}, recalled in Lemma~\ref{lem:LN_trial}, while retaining the interaction with the possibly charged ball $B_\ell$. Constants below may depend on $\rho_0$ and $C_0$, but not on $\ell$, $N$ or the state $\Gamma_N$.

We proceed in five steps.
Step~1 uses the Onsager bound and the hypothesis $\mathcal J_{\rho_0\iden_{B_\ell}}(\Gamma_N)\le C_0\ell^3$ to show that the excess background charge $Q=\rho_0|B_\ell|-N$ of the ball $B_\ell$ satisfies $|Q|\le K_0\ell^2$. Step~2 chooses $\ell'$ and places in the annulus $B_{\ell'}\setminus\overline{B_\ell}$ a background of density $\rho_0$ and $M-N$ particles with uniform density in a slightly thinner annulus. By Newton's theorem, the resulting radial shell charge costs at most $-Q^2/(2\ell')+C\ell^2$, including its interaction with the ball $B_\ell$. Step~3 averages over the positions of $M-N$ centers in this thinner annulus and selects a configuration whose classical energy, plus a short-range pair potential, obeys the same bound. Step~4 places disjointly supported orbitals around these centers, at an additional cost $\mathcal O(\ell^2)$. Step~5 takes the fermionic product with $\Gamma_N$.

\medskip

\noindent\textbf{Step 1: Control of the excess charge.}
Fix a configuration $X=(x_1,\dots,x_N)\in B_\ell^N$, let $U^{\rho_0}_{B_\ell}(X)$ be the classical Jellium energy~\eqref{eq:LN_U},
 and let
\[
\sigma_X=\sum_{j=1}^N\frac{\iden_{B_1(x_j)}}{|B_1|}-\rho_0\iden_{B_\ell}
\]
be the total charge of the configuration, with each particle smeared uniformly over a unit ball and the background kept as it is. Keeping the positive Coulomb term, Lemma~\ref{lem:basic_onsager} with $\rho=\rho_0\iden_{B_\ell}$ and $a=1$ gives
\begin{equation}\label{eq:completion_onsager}
U^{\rho_0}_{B_\ell}(X)\ge D(\sigma_X)-cN,
\qquad c=\frac35+\frac{2\pi}5\rho_0.
\end{equation}
The charge distribution $\sigma_X$ is supported in $\overline{B_{\ell+1}}$ and has total charge $-Q$, where $Q=\rho_0|B_\ell|-N$ is the excess background charge.
Every signed charge distribution $\sigma$ of finite Coulomb energy supported in $\overline{B_r}$ satisfies
\[
D(\sigma)\ge\frac{(\int_{\RR^3}\sigma)^2}{2r}.
\]
Indeed, if $\omega_r$ is the uniform probability measure on $\partial B_r$, its potential is $1/r$ throughout $B_r$ and $D(\omega_r)=1/(2r)$. Expanding $D(\sigma-(\int_{\RR^3}\sigma)\omega_r)\ge0$ proves the estimate. Then \eqref{eq:completion_onsager} becomes
\[
U^{\rho_0}_{B_\ell}(X)\ge\frac{Q^2}{2(\ell+1)}-cN
\qquad\text{for all }X\in B_\ell^N.
\]
The energy $\mathcal J_{\rho_0\iden_{B_\ell}}(\Gamma_N)$ is the expectation of the nonnegative kinetic energy plus $U^{\rho_0}_{B_\ell}$ with respect to $\Gamma_N$. Hence
\begin{equation}\label{eq:completion_charge_coercivity}
\mathcal J_{\rho_0\iden_{B_\ell}}(\Gamma_N)
\ge\frac{Q^2}{2(\ell+1)}-cN.
\end{equation}
Since the left side is at most $C_0\ell^3$ by hypothesis, $N=\rho_0|B_\ell|-Q\le\frac{4\pi}3\rho_0\ell^3+|Q|$, and $2(\ell+1)\le4\ell$ for $\ell\ge1$, this implies
\[
Q^2\le4\Big(C_0+\frac{4\pi}3c\rho_0\Big)\ell^4+4c\ell|Q|.
\]
Solving this quadratic inequality and using $\ell\le\ell^2$ gives
\begin{equation}\label{eq:completion_charge_bound}
|Q|\le K_0\ell^2,\qquad K_0:=4c+2\sqrt{C_0+\frac{4\pi}3c\rho_0}.
\end{equation}
Note that the constant $K_0$ depends only on $\rho_0$ and $C_0$.

\begin{figure}[ht]
\centering
\begin{tikzpicture}[scale=0.9,>=stealth]
  \fill[gray!12,even odd rule] (0,0) circle (3.3) (0,0) circle (2.0);
  \fill[gray!35,even odd rule] (0,0) circle (3.04) (0,0) circle (2.26);
  \fill[pattern=north east lines,pattern color=gray!70] (0,0) circle (2.0);
  \draw[thick] (0,0) circle (2.0);
  \draw[thick] (0,0) circle (3.3);
  \draw[densely dashed] (0,0) circle (2.26);
  \draw[densely dashed] (0,0) circle (3.04);
  \foreach \x/\y in {-2.00/-1.95,0.82/-2.80,-0.18/-2.90,1.96/-1.41,-2.45/0.48,-2.59/-0.73,-0.53/2.85,0.24/-2.43,0.71/-2.33,0.62/2.40,-1.58/1.88,1.46/1.89,-1.06/2.42,1.23/2.65,2.18/1.76,1.70/2.03}
    \fill (\x,\y) circle (0.045);
  \draw[decorate,decoration={brace,amplitude=4pt}] (2.26,0.07) -- (3.04,0.07) node[midway,above=4pt] {$\mathcal A'$};
  \draw[decorate,decoration={brace,amplitude=4pt,mirror}] (2.0,-0.07) -- (3.3,-0.07) node[midway,below=4pt] {$\mathcal A$};
  \draw[->] (0,0) -- (200:2.0) node[pos=0.55,fill=white,inner sep=1pt] {$\ell$};
  \draw[->] (0,0) -- (245:3.3) node[pos=0.8,fill=white,inner sep=1pt] {$\ell'$};
  \node[fill=white,inner sep=2pt] at (0,0.75) {$\operatorname{supp}\nu\subset\overline{B_\ell}$};
  \draw[->] (4.2,2.6) node[right] {$y_j\in\mathcal A'$} -- (50:2.72);
  \draw[->] (4.2,-1.6) node[right] {$\operatorname{supp}\mu\subset\overline{\mathcal A}$} -- (-28:3.17);
\end{tikzpicture}
\caption{The geometrical setup of the proof of Lemma~\ref{lem:lower_completion}.}\label{fig:completion_shell}
\end{figure}
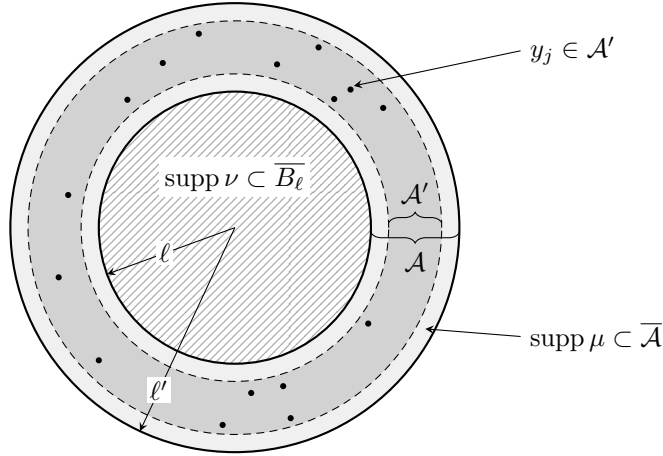

\medskip

\noindent\textbf{Step 2: The annulus.}
Choose $h'>4$ such that $4\pi\rho_0h'>K_0+1$, and set
\[
M=\left\lceil\rho_0|B_{\ell+h'}|\right\rceil,
\qquad
\ell'=\left(\frac{3M}{4\pi\rho_0}\right)^{1/3}.
\]
Then $\ell'\ge\ell+h'$, and $\ell'\le\ell+h$ with $h:=h'+1$. Indeed, $M<\rho_0|B_{\ell+h'}|+1$, so it suffices that $(s+1)^3-s^3\ge3/(4\pi\rho_0)$ for $s=\ell+h'$. This holds because $s>\max\{4,1/(4\pi\rho_0)\}$ implies $s^2\ge1/(4\pi\rho_0)$. Moreover, $M-N$ is an integer, and~\eqref{eq:completion_charge_bound} gives
\begin{align*}
M-N&\ge\rho_0\big(|B_{\ell+h'}|-|B_\ell|\big)+Q\ge(4\pi\rho_0h'-K_0)\ell^2>\ell^2,\\
M-N&\le\rho_0\big(|B_{\ell+h'}|-|B_\ell|\big)+Q+1\le C\ell^2,
\end{align*}
where $C$ depends only on $\rho_0$ and $C_0$. Write
\[
\mathcal A=B_{\ell'}\setminus\overline{B_\ell},\qquad
\mathcal A'=B_{\ell'-1}\setminus\overline{B_{\ell+1}};
\]
see Figure~\ref{fig:completion_shell}. Since $h'>4$, we have $|\mathcal A'|\ge8\pi\ell^2$, and hence $(M-N)/|\mathcal A'|\le C$. Put
\[
\nu=\rho_{\Gamma_N}-\rho_0\iden_{B_\ell},\qquad
\mu=\frac{M-N}{|\mathcal A'|}\iden_{\mathcal A'}-\rho_0\iden_{\mathcal A},\qquad
m(t)=\int_{\ell<|x|<t}\mu(x)\,\mathrm dx.
\]
The charge distribution $\nu$ need not be radial. We only use its support in $\overline{B_\ell}$ and its total charge $-Q$. The shell charge $\mu$ is radial and has total charge $Q$, so $m(\ell)=0$ and $m(\ell')=Q$. Let $0<|x|\le\ell$. Since $\mu$ vanishes outside $\{\ell\le|y|\le\ell'\}$, the first term on the right of Newton's formula~\eqref{eq:basic_newton} vanishes, and the second reduces to an integral over $[\ell,\ell']$. Using $\mathrm dm(t)=4\pi t^2\mu(t)\,\mathrm dt$ and integrating by parts, we get
\[
\int_{\RR^3}\frac{\mu(y)}{|x-y|}\,\mathrm dy
=4\pi\int_\ell^{\ell'}\mu(t)\,t\,\mathrm dt
=\int_\ell^{\ell'}\frac{1}{t}\,\mathrm dm(t)
=\frac Q{\ell'}+\int_\ell^{\ell'}\frac{m(t)}{t^2}\,\mathrm dt,
\]
which does not depend on $x$. Since $\nu$ is supported in $\overline{B_\ell}$ with total charge $-Q$, this gives
\[
2D(\nu,\mu)=-\frac{Q^2}{\ell'}-Q\int_\ell^{\ell'}\frac{m(t)}{t^2}\,\mathrm dt.
\]
Then Theorem~\ref{thm:basic_newton}(iv), together with $m(t)=Q$ for $t\ge \ell'$, gives
\[
D(\mu)=\frac12\int_\ell^\infty\frac{m(t)^2}{t^2}\,\mathrm dt
=\frac{Q^2}{2\ell'}+\frac12\int_\ell^{\ell'}\frac{m(t)^2}{t^2}\,\mathrm dt.
\]
Adding these yields the exact identity
\begin{equation}\label{eq:completion_shell_identity}
D(\mu)+2D(\nu,\mu)
=-\frac{Q^2}{2\ell'}
+\frac12\int_\ell^{\ell'}\frac{m(t)^2-2Qm(t)}{t^2}\,\mathrm dt.
\end{equation}
Since $\sup|m|\le\|\mu\|_1\le C\ell^2$ and $\ell'-\ell=\mathcal O(1)$, we obtain
\begin{equation}\label{eq:completion_shell_bound}
D(\mu)+2D(\nu,\mu)\le-\frac{Q^2}{2\ell'}+C\ell^2.
\end{equation}

\medskip

\noindent\textbf{Step 3: A classical shell configuration.}
For $M-N$ particles in the annulus $\mathcal A'$, at positions $Y=(y_1,\dots,y_{M-N})$, the classical energy added to the ball $B_\ell$ by these particles and the annular background $\rho_0\iden_{\mathcal A}$ is
\[
\mathcal W(Y)=U^{\rho_0}_{\mathcal A}(Y)
+\sum_{j=1}^{M-N}v_\nu(y_j)-2D(\nu,\rho_0\iden_{\mathcal A}),
\]
where $U^{\rho_0}_{\mathcal A}$ is the classical Jellium energy~\eqref{eq:LN_U} of the annulus and $v_\nu$ is the potential of $\nu$. We average over configurations with respect to the uniform probability measure on $(\mathcal A')^{M-N}$, writing $\mathrm dY=\mathrm dy_1\cdots\mathrm dy_{M-N}$. Integrating term by term, then dropping the negative last term and applying~\eqref{eq:completion_shell_bound}, gives
\begin{equation}\label{eq:completion_sampling}
\begin{split}
\frac1{|\mathcal A'|^{M-N}}\int_{(\mathcal A')^{M-N}}\mathcal W(Y)\,\mathrm dY
&=D(\mu)+2D(\nu,\mu)-\frac{M-N}{|\mathcal A'|^2}D(\iden_{\mathcal A'})\\
&\le-\frac{Q^2}{2\ell'}+C\ell^2.
\end{split}
\end{equation}
The negative term accounts for the absence of self-interaction of the $M-N$ particles. Introduce the integrable pair potential
\[
w(y)=8\pi^2\frac{\iden_{\{|y|\le1\}}}{|y|^2}.
\]
Its average is
\begin{multline*}
\frac1{|\mathcal A'|^{M-N}}\int_{(\mathcal A')^{M-N}}\sum_{1\le j<k\le M-N}w(y_j-y_k)\,\mathrm dY\\
=\binom{M-N}2\frac1{|\mathcal A'|^2}\int_{\mathcal A'}\int_{\mathcal A'}w(x-y)\,\mathrm dx\,\mathrm dy
\le\frac{(M-N)^2}{2|\mathcal A'|}\|w\|_1\le C\ell^2.
\end{multline*}
The last step uses $\|w\|_1=32\pi^3$, together with $M-N=\mathcal O(\ell^2)$ and $(M-N)/|\mathcal A'|\le C$ from Step~2. Since configurations with coinciding particles form a null set, there is consequently a configuration $\widetilde Y=(\tilde y_1,\dots,\tilde y_{M-N})$ of distinct particles in $\mathcal A'$ such that
\begin{equation}\label{eq:completion_penalized_configuration}
\mathcal W(\widetilde Y)+\sum_{1\le j<k\le M-N}w(\tilde y_j-\tilde y_k)
\le-\frac{Q^2}{2\ell'}+C\ell^2.
\end{equation}

\medskip

\noindent\textbf{Step 4: A canonical shell state.}
For the configuration $\widetilde Y$ of Step~3, set
\[
d_j=\frac12\min\left\{1,\min_{k\ne j}|\tilde y_j-\tilde y_k|\right\},
\]
where the minimum over the empty set is $+\infty$. The balls $B_{d_j}(\tilde y_j)$ are pairwise disjoint and lie in $\mathcal A$. Let $\varphi_j$ be their normalized radial Dirichlet ground states,
\[
\varphi_j(x)=(2\pi d_j)^{-1/2}\,\frac{\sin(\pi|x-\tilde y_j|/d_j)}{|x-\tilde y_j|}\,\iden_{B_{d_j}(\tilde y_j)}(x),
\]
times any fixed unit spin vector in $\CC^q$, as in step~(d) of the proof of Lemma~\ref{lem:LN_trial}, and take their Slater determinant
\[
\Psi=\varphi_1\wedge\cdots\wedge\varphi_{M-N}\in{\textstyle\bigwedge^{M-N}}H^1_0(\mathcal A;\CC^q).
\]
Its kinetic energy satisfies
\begin{equation}\label{eq:completion_kinetic}
\TC\big(|\Psi\rangle\langle\Psi|\big)
=\sum_{j=1}^{M-N}\left(\frac\pi{d_j}\right)^2
\le\sum_{1\le j<k\le M-N}w(\tilde y_j-\tilde y_k)+4\pi^2(M-N).
\end{equation}
The inequality follows by summing $d_j^{-2}\le4\big(1+\sum_{k\ne j}|\tilde y_j-\tilde y_k|^{-2}\iden_{\{|\tilde y_j-\tilde y_k|\le1\}}\big)$ over $j$.
Disjoint supports remove exchange terms between these orbitals, and Newton's theorem makes their mutual Coulomb interaction exactly the interaction of point particles at the $\tilde y_j$:
\begin{equation}\label{eq:completion_coulomb}
\begin{split}
\Cc\big(|\Psi\rangle\langle\Psi|\big)
&=\sum_{1\le j<k\le M-N}\int_{\RR^3}\int_{\RR^3}\frac{|\varphi_j(x)|^2|\varphi_k(y)|^2}{|x-y|}\,\mathrm dx\,\mathrm dy\\
&=\sum_{1\le j<k\le M-N}\frac1{|\tilde y_j-\tilde y_k|}.
\end{split}
\end{equation}
Moreover, $v_\nu$ is harmonic throughout each orbital ball, so radial averaging gives
\[
\int_{\RR^3}|\varphi_j(x)|^2v_\nu(x)\,\mathrm dx=v_\nu(\tilde y_j).
\]
Since $2D(\nu,f)=\int_{\RR^3}f\,v_\nu$ and $\rho_\Psi=\sum_j|\varphi_j|^2$, the ball--shell interaction is therefore
\begin{equation}\label{eq:completion_core_shell}
2D(\nu,\rho_\Psi-\rho_0\iden_{\mathcal A})=\sum_{j=1}^{M-N}v_\nu(\tilde y_j)-2D(\nu,\rho_0\iden_{\mathcal A}),
\end{equation}
exactly the ball--shell part of $\mathcal W(\widetilde Y)$, even when the density in $B_\ell$ is not radial.

We claim that replacing the point charges at $\tilde y_j$ by the orbital densities $|\varphi_j|^2$ in the attraction to the background $\rho_0\iden_{\mathcal A}$ increases the energy by at most $C(M-N)$, that is,
\begin{equation}\label{eq:completion_claim}
-2D(\rho_\Psi,\rho_0\iden_{\mathcal A})
\le-\rho_0\sum_{j=1}^{M-N}V_{\mathcal A}(\tilde y_j)+C(M-N),
\end{equation}
where $V_{\mathcal A}(x)=\int_{\mathcal A}|x-y|^{-1}\,\mathrm dy$, so that $\rho_0V_{\mathcal A}$ is the potential of the background, as in~\eqref{eq:LN_U}. To prove~\eqref{eq:completion_claim}, fix $j$ and consider the kernel difference
\[
\kappa_j(z)=\frac1{|z|}-\int_{\RR^3}\frac{|\varphi_j(\tilde y_j+u)|^2}{|z-u|}\,\mathrm du,\qquad z\ne0.
\]
The function $u\mapsto|\varphi_j(\tilde y_j+u)|^2$ is a radial probability density supported in $\overline{B_{d_j}}$. Hence $\kappa_j\ge0$ by Theorem~\ref{thm:basic_newton}(ii), and $\kappa_j(z)=0$ for $|z|\ge d_j$ by Theorem~\ref{thm:basic_newton}(i). Since the integral in the definition of $\kappa_j$ is nonnegative, we also have
\[
\int_{\RR^3}\kappa_j\le\int_{B_{d_j}}\frac{\mathrm dz}{|z|}=2\pi d_j^2.
\]
Now $2D(\rho_\Psi,\rho_0\iden_{\mathcal A})=\rho_0\int_{\RR^3}\rho_\Psi V_{\mathcal A}$ and $\rho_\Psi=\sum_j|\varphi_j|^2$. Exchanging the order of integration and substituting $x=\tilde y_j+u$ gives
\[
V_{\mathcal A}(\tilde y_j)-\int_{\RR^3}|\varphi_j(x)|^2V_{\mathcal A}(x)\,\mathrm dx
=\int_{\mathcal A}\kappa_j(y-\tilde y_j)\,\mathrm dy
\le2\pi d_j^2.
\]
Multiplying by $\rho_0$, summing over $j$ and using $d_j\le\frac12$ proves~\eqref{eq:completion_claim} with $C=\frac\pi2\rho_0$.

\medskip

\noindent\textbf{Step 5: Assembly.}
Let $\Gamma_M=\Gamma_N\wedge|\Psi\rangle\langle\Psi|$ be the antisymmetrized product given by Lemma~\ref{lem:basic_product}, with $A_1=B_\ell$ and $A_2=\mathcal A$. It is a normalized state of exactly $N+(M-N)=M$ particles supported in $B_{\ell'}$ satisfying homogeneous Dirichlet boundary conditions, and its background charge is $\rho_0|B_{\ell'}|=M$. Since $\rho_0\iden_{B_{\ell'}}=\rho_0\iden_{B_\ell}+\rho_0\iden_{\mathcal A}$ almost everywhere, \eqref{eq:basic_product_J} gives
\begin{equation}\label{eq:completion_product}
\mathcal J_{\rho_0\iden_{B_{\ell'}}}(\Gamma_M)
=\mathcal J_{\rho_0\iden_{B_\ell}}(\Gamma_N)
+\mathcal J_{\rho_0\iden_{\mathcal A}}\big(|\Psi\rangle\langle\Psi|\big)
+2D(\nu,\rho_\Psi-\rho_0\iden_{\mathcal A}).
\end{equation}
By the definition of $\mathcal J_\zeta(\vGamma)$ in Section~\ref{sec:basic},
\[
\mathcal J_{\rho_0\iden_{\mathcal A}}\big(|\Psi\rangle\langle\Psi|\big)
=\TC\big(|\Psi\rangle\langle\Psi|\big)+\Cc\big(|\Psi\rangle\langle\Psi|\big)
-2D(\rho_\Psi,\rho_0\iden_{\mathcal A})+D(\rho_0\iden_{\mathcal A}).
\]
Inserting~\eqref{eq:completion_kinetic}, \eqref{eq:completion_coulomb} and~\eqref{eq:completion_claim}, and comparing with~\eqref{eq:LN_U}, we obtain
\[
\mathcal J_{\rho_0\iden_{\mathcal A}}\big(|\Psi\rangle\langle\Psi|\big)
\le U^{\rho_0}_{\mathcal A}(\widetilde Y)+\sum_{1\le j<k\le M-N}w(\tilde y_j-\tilde y_k)+C(M-N).
\]
Adding~\eqref{eq:completion_core_shell} and recalling the definition of $\mathcal W(Y)$, the last two terms in~\eqref{eq:completion_product} are therefore at most
\[
\mathcal W(\widetilde Y)+\sum_{1\le j<k\le M-N}w(\tilde y_j-\tilde y_k)+C(M-N)
\le-\frac{Q^2}{2\ell'}+C\ell^2,
\]
where we used~\eqref{eq:completion_penalized_configuration} and $M-N=\mathcal O(\ell^2)$. Inserting this into~\eqref{eq:completion_product} and dropping the nonpositive term $-Q^2/(2\ell')$ gives
\[
\mathcal J_{\rho_0\iden_{B_{\ell'}}}(\Gamma_M)
\le\mathcal J_{\rho_0\iden_{B_\ell}}(\Gamma_N)+C_1\ell^2,
\]
which is~\eqref{eq:lower_completion}.

\end{document}